\documentclass[12pt]{article}

\newcommand{\lcc}{{\widehat\nabla}} 
\newcommand{\Dlcc}{{\widehat D}} 
\newcommand{\Klcc}{{\widehat K}} 
\newcommand{\Rlcc}{{\widehat R}} 
\newcommand{\bnw}{{\nabla}} 
\newcommand{\lvf}{{\ul\lambda}} 
\newcommand{\efnwt}[1]{{{\mr e}^{\widetilde{#1}}}} 
\newcommand{\efnsim}[1]{{{\mr e}^{#1^\sim}}} 

\usepackage{amsmath}
\usepackage{amssymb}
\usepackage{amsthm}
\usepackage{enumerate}
\usepackage{bbm}
\usepackage{bm}
\usepackage{mathtools}
\usepackage{tikz-cd}
\usepackage{hyperref}
\usepackage{xcolor}

\numberwithin{equation}{section}
\theoremstyle{theorem}
\newtheorem{Theorem}{Theorem}[section]
\newtheorem{Proposition}[Theorem]{Proposition}
\newtheorem{Propdef}[Theorem]{Proposition-Definition}
\newtheorem{Lemma}[Theorem]{Lemma}
\newtheorem{Corollary}[Theorem]{Corollary}

\newcommand{\thistheoremname}{}
\newtheorem{Genthm}[Theorem]{\thistheoremname}
\newenvironment{Namedthm}[1]
  {\renewcommand{\thistheoremname}{#1}%
   \begin{Genthm}}
   {\end{Genthm}}

\theoremstyle{definition}
\newtheorem{Definition}[Theorem]{Definition}

\theoremstyle{remark}
\newtheorem{Remark}[Theorem]{Remark}
\newtheorem{Example}[Theorem]{Example}

\newcommand{\btm}{\begin{Theorem}}
\newcommand{\etm}{\end{Theorem}}

\newcommand{\ben}{\begin{enumerate}}
\newcommand{\een}{\end{enumerate}}

\newcommand{\bit}{\begin{itemize}}
\newcommand{\eit}{\end{itemize}}

\newcommand{\bca}{\begin{cases}}
\newcommand{\eca}{\end{cases}}

\newcommand{\bre}{\begin{Remark}\rm}
\newcommand{\ere}{\end{Remark}}

\newcommand*{\bbm}{\begin{Remark}}
\newcommand*{\ebm}{\end{Remark}}

\newcommand{\ble}{\begin{Lemma}}
\newcommand{\ele}{\end{Lemma}}

\newcommand*{\bsz}{\begin{Proposition}}
\newcommand*{\esz}{\end{Proposition}}

\newcommand{\beq}{\begin{equation}}
\newcommand{\eeq}{\end{equation}}

\newcommand{\bbma}{\begin{bmatrix}}
\newcommand{\ebma}{\end{bmatrix}}

\newcommand*{\bbs}{\begin{Example}}
\newcommand*{\ebs}{\end{Example}}

\newcommand*{\bfg}{\begin{Corollary}}
\newcommand*{\efg}{\end{Corollary}}

\newcommand*{\bdf}{\begin{Definition}}
\newcommand*{\edf}{\end{Definition}}

\newcommand*{\bbw}{\begin{proof}}
\newcommand*{\ebw}{\end{proof}}

\newcommand*{\bpf}{\begin{proof}}
\newcommand*{\epf}{\end{proof}}

\newtagform{simple}{}{}
\newtagform{standard}{(}{)}

\newcommand{\bibend}{.}

\allowdisplaybreaks

\newcommand{\II}{\mathbbm{1}}

\newcommand{\CC}{{\mathbb{C}}}

\newcommand{\NN}{{\mathbb{N}}}

\newcommand{\RR}{{\mathbb{R}}}

\newcommand{\R}{\mathbb{R}}

\newcommand{\N}{\mathbb{N}}
\newcommand{\Z}{\mathbb{Z}}

\newcommand{\mb}{\mathbf}

\newcommand*{\res}{\upharpoonright}

\newcommand{\SO}{{\mr{SO}}}

\newcommand{\SU}{{\mr{SU}}}
\newcommand{\su}{\mf{su}}

\newcommand{\al}[1]{\begin{align} #1 \end{align}}
\newcommand{\ala}[1]{\begin{align*} #1 \end{align*}}

\DeclareMathOperator{\Ad}{Ad}

\DeclareMathOperator{\im}{im}

\DeclareMathOperator{\tr}{tr}

\newcommand{\mc}[1]{\mathcal{#1}}
\newcommand{\mf}[1]{\mathfrak{#1}}
\newcommand{\mr}[1]{\mathrm{#1}}

\newcommand{\comment}[1]{}
\newcommand{\verweis}[1]{}
\newcommand{\todo}[1]{}
\renewcommand{\d}{{\mr d}}

\newcommand{\ve}{\varepsilon}

\newcommand{\vp}{\varphi}

\newcommand{\yi}{\upsilon}
\newcommand{\ctg}{\mr T^\ast}
\newcommand{\tg}{\mr T}

\newcommand{\rref}[1]{{\rm \ref{#1}}}

\newcommand{\ol}[1]{\overline{#1}}
\newcommand{\ul}[1]{\underline{#1}}
\newcommand{\abs}{\hspace*{2.5mm}}

\newcommand{\g}{\mathfrak{g}}

\newcommand{\Zeta}{\mathrm{Z}}

\DeclareMathOperator{\End}{End}

\DeclareMathOperator{\Mat}{Mat}

\DeclareMathOperator{\ad}{ad}

\newcommand{\calC}{\mathcal{C}}
\newcommand{\calD}{\mathcal{D}}
\newcommand{\calA}{\mathcal{A}}
\newcommand{\calI}{\mathcal{I}}

\newcommand{\red}{\slash\!\slash}
\newcommand{\id}{\operatorname{id}}
\newcommand{\ext}{\mathsf{e}}
\newcommand{\rest}{\mathsf{r}}
\newcommand{\CE}{\operatorname{CE}}
\newcommand{\Hom}{\operatorname{Hom}}
\newcommand{\imi}{\mr i}

\newcommand{\gp}{G}
\newcommand{\la}{\mf g}

\newcommand{\mm}{J}
\newcommand{\pt}{\Lambda}

\begin{document}

\title{Deformation quantization and homological reduction of a lattice gauge model}

\author{M.J.~Pflaum, G.~Rudolph, M.~Schmidt}

\maketitle

\begin{abstract}
\noindent
For a compact Lie group $G$ we consider a lattice gauge model given by the $G$-Hamiltonian
system which consists of the cotangent bundle of a power of $G$ with its canonical symplectic structure
and standard moment map. We explicitly construct a Fedosov quantization of the
underlying symplectic manifold using the Levi-Civita connection of the Killing
metric on $G$. We then explain and refine quantized homological reduction
for the construction of a star product on the symplectically reduced space in the singular case.
Afterwards we show that for $G = \SU (2)$ the main hypotheses ensuring the method of quantized homological
reduction to be applicable hold in the case of our lattice gauge model. For that
case, this implies that the - in general singular - symplectically reduced phase space
of the corresponding lattice gauge model carries a star product.
\end{abstract}

\newpage

\tableofcontents


\section{Introduction}

In this paper, we apply the homological approach to singular reduction in deformation quantization 
developed in \cite{BHP} to a model of gauge theory obtained via lattice approximation of Yang-Mills theory within the Hamiltonian approach. 
We refer to the classical paper \cite{KS} for the formulation of the full model (including matter fields) on a finite lattice and for its canonical quantization. In geometric terms, 
the underlying classical phase space is a product of copies of the cotangent bundle over the gauge group manifold $G$, endowed with the canonical symplectic structure, and the canonical 
moment map is given by the Gauss constraint generator. In \cite{KR1,KR2}, the canonical quantization procedure of this model was taken up in the language of $C^*$-algebras. 
The authors of these papers studied the structure of both the field and the observable algebras of the model including a discussion of the Gauss law 
and the classification of the irreducible representations of the algebra of observables. The latter is, by definition, the quotient of the algebra of gauge invariant operators by the 
ideal generated by the Gauss law. In  \cite{GR1,GR2}, this structural analysis was continued with emphasis on the construction of the thermodynamical limit including also 
the quantum dynamics of the system. Here, we limit our attention to pure Yang-Mills theory (without matter fields) in the finite lattice context. 

It should be clear that within the above approach the algebra of observables rather than the space of states plays the primary role. On the other hand, by standard $C^*$-algebraic arguments 
or, alternatively, by the theory of systems of imprimitivity, one has a unique field algebra representation (the generalized Schr\"odinger representation) and, therefore, it is quite 
straightforward to reduce the gauge symmetry after quantization yielding an identification of the observable algebra of pure lattice Yang-Mills theory with the algebra of compact operators 
on the Hilbert space of square integrable functions over a product of copies of $G$ (the classical configuration space). 
As we are dealing with reduction after quantization here, this algebra a priori 
does not contain any information about the classical gauge orbit stratification of the reduced phase space, the latter being obtained via singular symplectic reduction for the moment 
map at level zero. However,  using the polar decomposition map, the unreduced phase space may be identified with the product of copies of the complexification of $G$, 
this way aquiring a natural K\"ahler structure. Thus, a concept developed by Huebschmann \cite{H} combined with results of Hall \cite{Hall:SBT} may be applied, yielding a costratification of the physical Hilbert space, which may be viewed as the quantum counterpart of the classical stratification. We refer to \cite{HRS,FuRS,FJRS} for the study of this structure including 
a discussion of its possible physical relevance. Recently \cite{KRS}, we have also clarified how to implement the classical stratification on the level of the observable algebra, leading to a stratification of the latter that is, in a sense, dual to the costratification of the physical Hilbert space. In a sense, the above observable algebra endowed with this additional stratified structure may be viewed as a reasonable substitute for a (sometimes desired) theory obtained via quantization after reduction, which within the above approach has not been worked out yet.

Deformation quantization is another quantization procedure which heavily rests on the Hamiltonian structure of the classical phase space and on Marsden--Weinstein reduction. In this respect, 
it is rather close in spirit to the above described approach. On the other hand, in some aspects it differs drastically
from the $C^*$-algebraic approach. 
To be more precise, what we are dealing with here is Fedosov's formal deformation quantization \cite{F} of the unreduced phase space defined above. Then, various options  for the star product can be chosen, see 
\cite{BNW1,BNW2,G1,G2,GR}. Using the above mentioned K\"ahler structure, a Fedosov star product of Wick type can be taken as well, see \cite{BW,N}.
It would be desirable to compare these options, but in this paper we merely choose one of them, namely the product of the standard order type. In \cite{FedosovLMP}, Fedosov has 
shown that there is a natural deformation quantization analog of classical regular symplectic reduction. Next, this issue was taken up by Bordemann, Herbig and Waldmann \cite{BHW}, 
who developed a deformation quantization formulation of the BRST-method. They proved that, under appropriate regularity properties of the group action, the BRST-procedure induces a star 
product on the reduced phase space. In \cite{BHP}, Bordemann, Herbig and Pflaum showed that this method may be extended to singular symplectic reduction, provided the following assumptions 
are fulfilled:
\begin{enumerate}
\item[(GH)\hspace{-1mm}]
The components of the moment map $J$ generate the vanishing ideal of the zero level set $J^{-1} (0)$.
\item[(AC)\hspace{-1mm}]
The Koszul complex on $J$ in the ring of smooth functions on phase space is acyclic. 
\end{enumerate}
Moreover, the star product of the underlying unreduced quantum deformation theory has to fulfill some equivariance conditions. 
The main ideas of this reduction procedure are as follows. 
\begin{enumerate}
\item
For a given $G$-Hamiltonian system $(M, \omega, \Psi, J)$, one constructs the classical BRST-complex $(\calA^\bullet,\calD)$ by taking the graded tensor product of the Chevalley--Eilenberg complex
$\CE^\bullet (\g, \calC^\infty(M))$ associated with the $\g$-module $\calC^\infty (M)$ with the Koszul complex $(K^\bullet,\partial)$ on the moment map $J$ and endows it with the structure of  a differential graded commutative $\calC^\infty (M)$-algebra. Moreover, one shows that the latter carries a natural Poisson structure. 
Now, one can prove that, under the assumptions (GH) and (AC),
the classical symplectically reduced space  is representable (via a deformation retract) as the zeroth cohomology of this BRST-complex  with its natural structure of a differential graded
Poisson algebra.
\item 
Assume we are given a star product $\star$ obtained by formal deformation quantization of the $G$-Hamiltonian system $(M, \omega, \Psi, J)$, fulfilling some natural invariance conditions to be discussed later. Combining this star product with the natural product on the Gra{\ss}mann part, one can endow the $\CC[[\lambda]]$-module $\calA^\bullet[[\lambda]]$ of formal power series with values in $\calA^\bullet$ with a star product $\ast$. Moreover, one  constructs a deformation $\bm{\calD}$ of the classical BRST-differential, thus arriving at a formal deformation quantization $(\calA^\bullet[[\lambda]], \ast , \bm{\calD})$ (called the quantum BRST algebra) of the classical BRST algebra  $\calA^\bullet$. Finally, one can prove that there exists a deformed version of the contraction mentioned under point 1, giving rise 
to a star product on the symplectically reduced space. 
\end{enumerate}

The main result of the present paper consists in the proof that the above conditions (GH) and (AC),
together with the needed equivariance conditions on the star product $\star$, are fulfilled for the gauge model under consideration with gauge group $G = \SU (2)$, see Section \ref{sec:application}. That is, we have proved that homological reduction may be applied to lattice gauge theory. 
Clearly, the star product on the reduced phase space is given in a complicated implicit way. To make it more explicit, one has to study the deformation retract structure entering the whole construction. This will be done in future work. 

There are two further results holding true for any compact connected gauge group $G$ which should be mentioned. First, we have calculated the (standard order) star product for the unreduced theory in detail (Section \ref{sec:fedosov}), thus, 
in particular extending results contained in \cite{BNW1} and, second, we have provided the reader with a deeper analysis of the assumptions needed for the deformation retract 
method used in various places of the paper, see Theorem \ref{Thm:ExCompatibleEquivariantAnalyticStructures} and Theorem \ref{Thm:KoszulResolutionContraintSurface} which is an improved version of Theorem 3.2 in \cite{BHP}. 

One final remark is in order. Throughout this paper, we have exclusively discussed formal deformation quantization. It is a challenge for future work to clarify whether the homological 
reduction method may be developed for strict deformation quantization (see e.g.\ \cite{Landsman}) as well. This would make it possible to compare the quantum observable algebra structure 
obtained here with the observable algebra obtained via canonical quantization described above in closer terms.  \vspace{2mm}

\textbf{Acknowledgements:}
M.J.P.~thanks  DESY Hamburg and the Max-Planck-Institut f\"ur Mathematik Bonn for hospitality and
support of his research stays. He also thanks the Universities of Leipzig and Bonn for hospitality.
Travel support by the Simons Foundation through award nr.~359389 and support by the NSF through
award OAC 1934725 is gratefully acknowledged.   
M.S.~acknowledges funding by DFG under grant SCHM1652/2.
The authors also thank the referees for constructive advice.


\section{The model}\label{sec:model}

Throughout the paper $G$ will denote a compact Lie group and $\g$ its Lie algebra. The lattice gauge model for which
we construct a deformation quantization is best represented as a particular $G$-Hamiltonian system.
Recall, \cite[Sec.~10.1]{BuchI}, that by a $G$-Hamiltonian system or a Hamiltonian $G$-manifold one
understands a quadruple $(M,\omega,\Psi,J)$ such that $(M,\omega)$ is a symplectic manifold,
$\Psi:G \times M \to M$ a smooth action of $G$ on $M$ by symplectomorphisms and such that
$J :M \to \g^*$ is a smooth map called the \emph{moment map} 
which is $G$-equivariant and which satisfies 
\begin{equation}
   d \mm_X  = - X_M \lrcorner \omega \quad \text{for all  } X \in \g \ . 
\end{equation}
Here, $J_X: M \to \R$ denotes the function which maps a point $p\in M$ to the pairing $\langle J(p),X\rangle$
and $X_M$ is the fundamental vector field of $X\in \g$ on $M$. The symplectically reduced space
$M\red G$ is now defined as the quotient space $M_0 /G$ of the zero level set $M_0 = \mm^{-1} (0)$
by the group action. Note that $M_0$, which often is also called the \emph{constraint surface},
is invariant under the group action by equivariance of the moment map and might possess singularities in case
$0$ is not a regular value of the moment map.

To define our lattice gauge model, let $\Lambda$ be a finite spatial lattice. Its sets of zero-dimensional, one-dimensional and two-dimensional elements are denoted by, respectively, $\Lambda^0$, $\Lambda^1$ and $\Lambda^2$ and are called, respectively, sites, links and plaquettes. We also assume that for the links and plaquettes an
arbitrary orientation has been chosen.
In the Hamiltonian approach to lattice gauge theory, gauge fields, or in other words the
variables, are approximated by their parallel transporters along links. Gauge transformations representing the symmetries are approximated by their values at the
lattice sites. The classical configuration space can then be identified with the space $G^{\Lambda^1}$ of maps
$\Lambda^1 \to G$, the classical symmetry group is the group $G^{\Lambda^0}$ of maps $\Lambda^0 \to G$ with pointwise
multiplication and the action of $g \in G^{\Lambda^0}$ on $a \in G^{\Lambda^1}$ is given by  
\beq
\label{G-Wir-voll}
(g \cdot a)(\ell) := g(x) a(\ell) g(y)^{-1}\,,
\eeq
where $\ell \in \Lambda^1$ and $x$, $y$ denote the starting point and the endpoint of $\ell$, respectively. The classical phase space is given by the associated Hamiltonian $G$-manifold \cite{AbrahamMarsden,BuchI} and the reduced classical phase space is obtained by symplectic reduction \cite{OrtegaRatiu,BuchI,SjamaarLerman}. 
Dynamics is governed by the classical counterpart of the Kogut-Susskind lattice Hamiltonian. After identifying $\ctg G$ with $G \times \mf g$, and thus $\ctg G^{\Lambda^1}$ with $G^{\Lambda^1} \times \mf g^{\Lambda^1}$, by means of left-invariant vector fields, the classical Hamiltonian is given by 
\beq
\label{Hamiltonian-C}
H(a,E)
 = 
\frac{\kappa^2}{2 \delta} \sum_{\ell \in \Lambda^1}^N \|E(\ell)\|^2
 -
\frac{1}{\kappa^2 \delta} \sum_{p \in \Lambda^2} \left(\tr a(p) + \ol{\tr a(p)}\right)
 \,,
\eeq
where $a \in G^{\Lambda^1}$, $\kappa$ denotes the coupling constant, $\delta$ denotes the lattice spacing and $a(p)$ is the product of $a(\ell)$ along the boundary of the plaquette $p \in \Lambda^2$ in the induced orientation. The trace is taken in some chosen unitary representation. Due to unitarity, the  Hamiltonian does not depend on the choice of plaquette orientations. Finally, $ E \in \mf g^{\Lambda^1}$ is the canonically conjugate momentum (classical colour electric field).

In the analysis of the orbit type stratification in continuum gauge theory it is reasonable to first factorize with respect to the free action of pointed gauge transformations. This leads to an action of the compact gauge group $G$ on the quotient manifold. This procedure can also be applied to the case of lattice gauge theory under consideration. Given a lattice site $x_0$, it is easy to see that the normal subgroup 
\beq
\label{G-ptgautrf}
\{g \in G^{\Lambda^0} : g(x_0) = \II\}\,,
\eeq
where $\II$ denotes the unit element of $G$, acts freely on $G^{\Lambda^1}$. Hence, one may pass to the quotient manifold and the residual action by the quotient Lie group of $G^{\Lambda^0}$ with respect to this normal subgroup. Obviously, the quotient Lie group is  isomophic to $G$. Let us explain how to identify the quotient manifold with a direct product of copies of $G$ and the quotient action with the action of $G$ by diagonal conjugation. Choose a maximal tree $\mc T$ in the graph $\Lambda^1$ and define the tree gauge of $\mc T$ as the subset
$$
\{a \in G^{\Lambda^1} : a(\ell) = \II  \text{ for all } \ell \in \mc T\}
$$
of $G^{\Lambda^1}$.
One can easily show that every element of $G^{\Lambda^1}$ is conjugate under $G^{\Lambda^0}$
to an element in the tree gauge of $\mc T$ and that two elements in the tree gauge of
$\mc T$ are conjugate under $G^{\Lambda^0}$ if they are conjugate under the action of $G$
via constant gauge transformations. As a consequence, the natural inclusion map of the
tree gauge into $G^{\Lambda^1}$ descends to a $G$-equivariant diffeomorphism from that
tree gauge onto the quotient manifold of $G^{\Lambda^1}$ with respect to the action of the
subgroup \eqref{G-ptgautrf}.
Finally, by choosing a numbering of the off-tree links in $\Lambda^1$, we can identify the tree gauge with the direct product of $N$ copies of $G$, where $N$ denotes the number of off-tree links. The number $N$ does not depend on the choice of $\mc T$. Under this identification, the action of $G$ on the tree gauge via constant gauge transformations translates into the action of $G$ on $G^N$ by diagonal conjugation
\beq\label{G-Wir-Q}
  \underline{\Psi} : G \times G^N \to G^N, \:
  (g, \ul a) = \big( g, (a_1 , \dots  , a_N) \big) \mapsto
  g \cdot \ul a =  (g a_1 g^{-1}, \dots  , g a_N g^ {-1}) \ .
\eeq
To summarize, for the analysis of the role of orbit types we may pass from the original Hamiltonian system with symmetries, given by the configuration space $G^{\Lambda^1}$, the symmetry group $G^{\Lambda^0}$ and the action \eqref{G-Wir-voll}, to the reduced Hamiltonian system with symmetries given by the configuration space 
$Q := G^N$,
the symmetry group $G$ and the action of $G$ on $Q$ given by diagonal conjugation \eqref{G-Wir-Q}. This is the system we will discuss in this paper. The classical phase space is given by the associated Hamiltonian $G$-manifold and the reduced classical phase space is obtained by symplectic reduction. First, by regular symplectic reduction, we obtain the partially reduced phase space $\ctg Q = \ctg G^N$ endowed with its canonical cotangent bundle projection
$\pi: \ctg Q \to Q$.
The action of $G$ on $Q$ lifts to a symplectic action on $\ctg Q$ admitting the standard moment map
\[
  \mm : \ctg Q \to \mf g^\ast \, , \quad \mm(\xi)\big(X)
  := \langle \xi , X_{\ctg Q}(p) \rangle\,,
\]
where $p\in Q$, $\xi \in \ctg_p Q$, $X \in \mf g$ and $X_{T^*Q}$ denotes
the fundamental vector field on $\ctg Q$ defined by $X$.
So one obtains a $G$-Hamiltonian system $(T^*G^N, \omega, \ul\Psi,J)$ which in
the following we will briefly refer to as the \emph{lattice gauge model} for
the Lie group $G$. Its
reduced phase space is obtained from $\ctg Q$ by singular symplectic reduction at $\mm = 0$,
$$
\ctg Q/\!/G := J^{-1}(0)/G\,.
$$
That is, it is the set of orbits of the $G$-action on the invariant subset $\mm^{-1}(0) \subset \ctg Q$, endowed with the quotient topology induced from the relative topology on this subset. In gauge theory, the condition $\mm=0$ corresponds to the Gau{\ss} law constraint. It turns out that the action of $G$ on $\mm^{-1}(0)$ has the same orbit types as that on $Q$. By definition, the orbit type strata of $\ctg Q/\!/G$ are the connected components of the orbit type subsets of $\ctg Q/\!/G$. They are called strata, because they provide a stratification \cite{PflaumBook} of $\ctg Q/\!/G$ \cite{SjamaarLerman,OrtegaRatiu}. By the theory of singular symplectic reduction, the orbit type strata of $\ctg Q/\!/G$ are endowed with symplectic manifold structures yielding a stratified symplectic space. As $\mm$ is linear on the fibers of $\ctg Q$ and hence $\mm^{-1}(0)$ contains the zero section of $\ctg Q$, the bundle projection $\pi : \ctg Q \to Q$ induces a surjective map 
$\ctg Q/\!/G \to Q/G$.
This map need not preserve the orbit type though.




\section{Fedosov deformation quantization of $\ctg\gp^N$}
\label{sec:fedosov}


We carry out Fedosov deformation quantization with respect to the Levi-Civita
connection of the Killing metric on $G^N$. The subsequent presentation rests
on the results of \cite{BNW1}. For our purposes, we have to discuss some points
in more detail. In particular, we present an explicit formula for the lift of
the Levi-Civita connection to $\ctg G^N$ and we calculate the bidifferential
operators in the corresponding Fedosov star product explicitly.


\subsection{Notation and conventions}


First, we have to develop the necessary calculus on $G^N$ and $\ctg G^N$. Given $k$-vector fields $X_1 , \dots , X_N$ on $\gp$, we can define a $k$-vector field $\ul X = (X_1 , \dots , X_N)$ on $\gp^N$ by 
$$
\ul X_{\ul a} = \big ((X_1)_{a_1} , \dots, (X_N)_{a_N}\big)
 \,,\quad 
\ul a \in \gp^N\,.
$$
By analogy, given $k$-forms $\xi_1 , \dots , \xi_N$ on $\gp$, we can define a $k$-form $\ul\xi = (\xi_1 , \dots , \xi_N)$ on $\gp^N$ by
$$
\ul\xi_{\ul a} = \big ((\xi_1)_{a_1} , \dots, (\xi_N)_{a_N}\big)
 \,,\quad 
\ul a \in \gp^N\,.
$$
 Evaluation of $\ul \xi$ on the $k$-vector field $\ul X$ then yields
\beq\label{G-Eval}
\ul \xi(\ul X) = \sum_{i=1}^N \xi_i(X_i) \in C^\infty(\gp^N)\,.
\eeq
All the vector fields and differential forms we will meet are of this specific type. For example, the left-invariant vector fields on $\gp^N$ are given by $\ul X \in \la^N$ and the left-invariant 1-forms by $\ul\xi \in \la^\ast{}^N$. Clearly,
$$
[\ul X,\ul Y] = \big([X_1,Y_1] , \dots , [X_N,Y_N]\big)\,.
$$
We will identify $\ctg \gp^N \cong \gp^N \times \la^\ast{}^N$ via  the global trivialization by left translation, 
\beq\label{G-Trvis}
\gp^N \times \la^\ast{}^N \to \ctg \gp^N
 \,,\qquad
(\ul a,\ul\alpha) \mapsto \ul\alpha_{\ul a}\,.
\eeq
Accordingly, 
\beq\label{G-Spl-Fa}
\tg_{(\ul a,\ul\alpha)} (\gp^N \times \la^\ast{}^N) 
= 
(\tg_{\ul a} \gp	^N) \oplus (\tg_{\ul\alpha} \la^\ast{}^N)
= 
\la^N \oplus \la^\ast{}^N
\,.
\eeq
Thus, we arrive at the identification
$$
\gp^N \times \la^\ast{}^N \times \la^N \times \la^\ast{}^N 
\cong
\tg (\ctg \gp^N)
\,,
$$
where the tuple $(\ul a , \ul\alpha , \ul X , \ul\xi)$ corresponds to the element of $\tg (\ctg \gp^N)$ which under \eqref{G-Trvis} is represented by the curve 
\beq\label{G-TaV-Ku}
t \mapsto \big(\ul a \exp(t\ul X) , \ul\alpha + t \ul\xi\big)
\,.
\eeq
For $\ul X \in \la^N$ and $\ul\xi \in \la^\ast{}^N$, let $(\ul X,\ul\xi)$ denote the vector field on $\ctg\gp^N$ defined by 
\beq\label{G-D-StdVF}
(\ul X , \ul\xi)_{(\ul a,\ul\alpha)} = (\ul a,\ul\alpha,\ul X,\ul \xi)
 \,,\qquad
(\ul a,\ul\alpha) \in \ctg\gp^N\,.
\eeq
Vector fields of this type will be referred to as standard vector fields on $\ctg\gp^N$. The flow of standard vector fields is given by
\beq\label{G-StdVF-Flow}
\Phi^{(\ul X,\ul\xi)}_t\big((\ul a,\ul\alpha)\big)
 =
\big(\ul a \exp(t\ul X) , \ul\alpha + t \ul\xi\big)
\eeq
and their commutator reads
\beq\label{G-StdVF-Ktr}
\big[(\ul X,\ul \xi) , (\ul Y,\ul \upsilon)\big]
 =
\big([\ul X,\ul Y] , 0\big)
\,.
\eeq
Correspondingly, elements of $\ctg (\ctg\gp^N)$ will be written in the form $(\ul a,\ul\alpha,\ul\xi,\ul X)$, where $(\ul\xi,\ul X)$ represents a cotangent vector at the point $(\ul a,\ul\alpha)$ via \eqref{G-Trvis} and the identification
$$
\ctg_{(\ul a,\ul\alpha)}\big(\gp^N \times \la^\ast{}^N\big) 
= 
\big(\ctg_{\ul a} \gp^N\big) \oplus \big(\ctg_{\ul \alpha} \la^\ast{}^N\big)
=
\la^\ast{}^N \oplus \la^N\,.
$$
In this description, the natural pairing between tangent vectors and cotangent vectors is given by 
$$
\langle (\ul a,\ul\alpha,\ul\xi,\ul X) , (\ul a,\ul\alpha,\ul Y,\ul \yi) \rangle
 =
\langle \ul\xi , \ul Y \rangle + \langle \ul\yi , \ul X \rangle\,.
$$
For $\ul\xi \in \la^\ast{}^N$ and $\ul X \in \la^N$, let $(\ul\xi,\ul X)$ denote the $1$-form on $\ctg\gp^N$ defined by 
\beq
\label{dec-ctg-ctgQ}
(\ul\xi , \ul X)_{(\ul a,\ul\alpha)} = (\ul a,\ul\alpha,\ul\xi,\ul X)
 \,,\quad
(\ul a,\ul\alpha) \in \ctg\gp^N\,.
\eeq
$1$-forms of this type will be referred to as standard $1$-forms on $\ctg\gp^N$. Recall that every vector field $Z$ on $\gp^N$ defines a tautological smooth function $\tilde{Z}$ on $\ctg \gp^N$ by 
\beq\label{G-D-tautFn}
\tilde Z(\eta) := \eta(Z_{\ul a})
 \,,\quad
\eta \in \ctg_{\ul a}\gp^N\,.
\eeq
For left-invariant vector fields $\ul X \in \la^N$,
\beq\label{G-LIVF-Fn}
\tilde{\ul X}(\ul a,\ul\alpha) = \ul\alpha(\ul X)\,.
\eeq
Together with \eqref{G-TaV-Ku}, this yields
$$
\left\langle \mr d \tilde{\ul X} , (\ul a,\ul\alpha,\ul Y,\ul\yi) \right\rangle
=
\left\langle \ul\yi , \ul X \right\rangle
$$
and hence, in the sense of \eqref{dec-ctg-ctgQ},
\beq\label{G-dtautFn}
\mr d \tilde{\ul X} = (0 , \ul X)
\,.
\eeq
Recall further the coadjoint representations $\Ad^\ast$ of $G$ and $\ad^\ast$ of $\la$, defined by 
$$
\langle \Ad^\ast (a) \xi, Y\rangle  = \langle \xi, \Ad(a^{-1})Y\rangle \, \quad 
\,,\qquad
\langle \ad^\ast (X) \xi, Y\rangle = - \langle \xi, [X, Y]\rangle \, \quad  
$$
for all $a \in G$, $X,Y \in \la$ and $\xi \in \la^\ast$.

Finally, for concrete calculations we will occasionally need to fix a basis $\{E_1 , \dots , E_d\}$
in $\la$. We then agree on the following conventions. The corresponding dual basis will always be
denoted $\{\ve^1 , \dots , \ve^d\}$. Let $\mc I := \{1 , \dots , N\} \times \{1 , \dots , d\}$. For
$I = (n , i) \in \mc I$ we write 
\beq\label{G-frame}
\ul E_I := (0, \dots , 0 , E_i , 0 , \dots , 0)
 \,,\qquad
\ul \ve^I := (0, \dots , 0 , \ve^i , 0 , \dots , 0)
\,,
\eeq
with the nonzero entry at the $n$-th place. The families $\{\ul E_I : I \in \mc I\}$ and $\{\ul\ve^I : I \in \mc I\}$ are then dual bases in $\la^N$ and $\la^\ast{}^N$, respectively, and thus provide dual global frames in $\tg G^N$ and
$\ctg G^N$, respectively. Let $C_{ij}^k$ denote the structure constants of $\la$ with respect to
the basis $(E_1 , \dots , E_d)$. Then, the structure constants $C_{IJ}^K$ of $\la^N$ with respect to the basis $\{\ul E_I\}$ are given by
\beq\label{G-D-CIJK}
C_{(n,i),(m,j)}^{(l,k)}
:= 
\begin{cases}
C_{ij}^k & |\abs n=m=l\,,
\\
0 & |\abs \text{otherwise.}
\end{cases}
\eeq


\subsection{Symplectic structure and Poisson structure}
\label{sec:SympStrucCotangBdl}


Let us denote the tautological $1$-form of $\ctg G$ by $\theta$, the corresponding standard symplectic form by $\omega = \d \theta$ and the corresponding standard Poisson tensor by $\pt$. Then, 
\beq\label{G-theta-N}
\ul\theta = (\theta , \stackrel{N}{\dots}, \theta)
 \,,\qquad
\ul\omega = \mr d \ul\theta = (\omega , \stackrel{N}{\dots}, \omega)
 \,,\qquad
\ul\Lambda = (\Lambda , \stackrel{N}{\dots}, \Lambda)
\eeq
represent, respectively, the tautological 1-form, the symplectic form and the Poisson tensor of $\ctg G^N$. As usual, the Hamiltonian vector field generated by a function $f \in C^\infty(\ctg\gp^N)$ will be denoted by $X_f$. We choose the convention $X_f \lrcorner\,  \ul \omega = - \d f$.

We derive formulae for the symplectic structure and the Poisson structure of $\ctg\gp^N$ under the identification \eqref{G-Trvis}. For $f \in C^\infty(\ctg\gp^N)$, define partial differentials 
$$
\mr d_\gp f : \ctg\gp^N \to \la^\ast{}^N 
 \,,\qquad
\mr d_{\la^\ast} f : \ctg\gp^N \to \la^N 
$$
by
\beq\label{G-dGf}
\mr d f \big((\ul a,\ul\alpha,\ul X,\ul\xi)\big)
 =
\langle \mr d_\gp f(\ul a,\ul\alpha) , \ul X \rangle
 + 
\langle \mr d_{\la^\ast} f(\ul a,\ul\alpha) , \ul \xi \rangle
\,.
\eeq

\ble\label{L-Fml-splStr}

For all $(\ul a,\ul\alpha) \in \ctg\gp^N$, all standard vector fields $(\ul X,\ul\xi)$ and all functions $f , g$ on $\ctg\gp^N$, one has
 \al{\label{G-theta}
\ul\theta_{(\ul a,\ul\alpha)}\big((\ul X,\ul\xi)\big)
 & = 
\ul\alpha(\ul X)\,,
\\ \label{G-omega}
\ul\omega_{(\ul a,\ul\alpha)}\big((\ul X,\ul\xi) , (\ul Y,\ul\zeta)\big)
 & = 
\ul\xi(\ul Y) - \ul\zeta(\ul X) - \ul\alpha([\ul X,\ul Y])\,,
\\ \label{G-HaVF}
(X_f)_{(\ul a,\ul\alpha)}
 & = 
 \big(
\ul a,\ul\alpha
 ,
\mr d_{\la^\ast} f(\ul a,\ul\alpha)
 , 
- \ad^\ast(\mr d_{\la^\ast} f(\ul a,\ul\alpha)) \ul\alpha - \mr d_\gp f(\ul a,\ul\alpha)
 \big)\,,
\\ \nonumber
\{f,g\}(\ul a,\ul\alpha)
 & = 
\langle \mr d_\gp g(\ul a,\ul\alpha) , \mr d_{\la^\ast} f(\ul a,\ul\alpha) \rangle
 -
\langle \mr d_\gp f(\ul a,\ul\alpha) , \mr d_{\la^\ast} g(\ul a,\ul\alpha) \rangle
\\ \label{G-PoKl}
 & \hspace{3cm} +
\ul\alpha\big([\mr d_{\la^\ast} f(\ul a,\ul\alpha)\,, \mr d_{\la^\ast} g(\ul a,\ul\alpha)]\big) \ ,
\\ \label{G-ImpAbb}
\mm(\ul a,\ul\alpha) 
 & =
\sum_{i=1}^N \big( \Ad^\ast(a_i) \alpha_i - \alpha_i \big) \ .
 }

\ele

\bpf

\eqref{G-theta} and \eqref{G-omega} follow by straightforward calculation. To prove  \eqref{G-HaVF}, we plug the ansatz $(X_f)_{(\ul a,\ul\alpha)} = (\ul a,\ul\alpha,\ul X,\ul\xi)$ into the equation 
$$
\ul\omega_{(\ul a,\ul\alpha)}\big(X_f , (\ul Y,\ul\zeta)\big)
 = 
- \mr d f\big((\ul Y,\ul\zeta)\big)
$$
with a standard vector field $(\ul Y,\ul\zeta)$. In view of \eqref{G-dGf} and \eqref{G-omega}, this yields
$$
\ul\xi(\ul Y) - \ul\zeta(\ul X) - \ul\alpha([\ul X,\ul Y])
 =
- \langle \mr d_\gp f(\ul a,\ul\alpha) , \ul Y \rangle
- \langle \mr d_{\la^\ast} f(\ul a,\ul\alpha) , \ul \zeta \rangle
$$
for all $\ul Y\in \la^N$ and $\ul\zeta \in \la^\ast{}^N$. Putting $\ul Y = 0$, we read off $\ul X = \mr d_{\la^\ast} f(\ul a,\ul\alpha)$. Putting then $\ul\zeta = 0$, we find $\ul\xi = - \ad^\ast(\mr d_{\la^\ast} f(\ul a,\ul\alpha)) \ul\alpha - \mr d_\gp f(\ul a,\ul\alpha)$.
Formula \eqref{G-PoKl} then follows from $\{f,g\} = \omega(X_f,X_g)$.
To prove \eqref{G-ImpAbb}, we observe that the fundamental vector field on $\gp^N$ generated by $B \in \la$ via the action by diagonal conjugation is given by 
$$
(B_{\gp^N})_{\ul a}
 =
\big(\mr R_{a_1}'B_\II - \mr L_{a_1}'B_\II , \dots , \mr R_{a_N}'B_\II - \mr L_{a_N}'B_\II\big)
\ .
$$
Hence, by left-invariance, 
 \ala{
\langle \mm(\ul a,\ul\alpha),B \rangle
 =
\langle \ul\alpha_{\ul a} , (B_{G^N})_{\ul a}\rangle
 =
\sum\nolimits_i \langle (\alpha_i)_{a_i} , \mr R_{a_i}'B_\II - \mr L_{a_i}'B_\II \rangle
 =
\sum\nolimits_i \langle \alpha_i , \Ad(a_i^{-1}) B - B \rangle \ .
 }
This yields the assertion.
\epf


\subsection{Lift of the Levi-Civita connection}


To derive the Fedosov standard ordered star product with respect to the Levi-Civita connection of the Killing metric on $G^N$, we first have to find a homogeneous and symplectic lift of this connection to $\ctg G^N$. Recall that, given a Riemannian manifold $Q$ with Levi-Civita connection $\lcc$, a torsion-free linear connection $\check\nabla$ on $\ctg Q$ is called

\ben

\item a lift of $\lcc$ if
$
\pi' \circ (\check\nabla_Z W)
=
(\lcc_{\widehat Z} {\widehat W}) \circ \pi
$
for all vector fields $Z$, $W$ on $\ctg Q$ and $\widehat Z$, $\widehat W$ on $Q$ satisfying $\pi' \circ Z = \widehat Z \circ \pi$ and $\pi' \circ W = \widehat W \circ \pi$,

\item symplectic if
$
\check\nabla \omega = 0
$,

\item homogeneous if
$
[\lvf,\check\nabla_U V] - \check\nabla_{[\lvf,U]} V - \check\nabla_U [\lvf V] = 0
$
for all vector fields $U$, $V$ on $\ctg Q$, where $\lvf$ denotes the Liouville vector field.

\een

It turns out that homogeneous symplectic lifts are not unique, see e.g.\ \cite{BCGRS}. As observed in \cite{BNW1}, one option to make the lift unique is to impose the additional condition that
$$
\omega\big(U_1,\check R(V,U_2)U_3+\check R(V,U_3)U_2\big)
+
(\text{cyclic permutations of } U_1, U_2, U_3)
=
0
$$
for all vector fields $U_i$, $V$ on $\ctg Q$, where $\check R$ denotes the curvature tensor of $\check\nabla$, viewed as a $2$-form on $\ctg Q$ with values in the $1,1$-tensor fields on $\ctg Q$. Let us refer to this connection as the BNW lift of $\lcc$ and let us denote it by $\bnw$. To write it down, we need the following lifting operations. First, $\lcc$ defines a horizontal lifting operator $\mr h$ by mapping every vector field $Z$ on $Q$ to a vector field $\mr h Z$ on $\ctg Q$, its horizontal lift, which is uniquely determined by the conditions 
\beq\label{G-hor}
\pi' \circ (\mr h Z) = Z \circ \pi
\,,\qquad
\Klcc \circ (\mr h Z) = 0
\,,
\eeq
where $\Klcc : \tg(\ctg Q) \to \ctg Q$ is the connection mapping of $\lcc$. Second, the structure of the cotangent bundle defines a (metric-independent) vertical lifting operator mapping every $1$-form $\zeta$ on $Q$ to the vertical vector field $\mr v\zeta$ on $\ctg Q$ induced by the complete flow 
$$
\ctg Q \times \RR \to \ctg Q
\,,\qquad
(p,t) \mapsto p + t \zeta\big(\pi(p)\big)
\,.
$$
Third, the lift of $1$-forms and the operation sending vector fields $Z$ on $Q$ to their tautological functions $\tilde Z$ on $\ctg Q$ combine to a lifting operation which turns $1,1$-tensor fields on $Q$ into vertical vector fields on $\ctg Q$, $T \mapsto \mr v T$. By definition, for $1,1$-tensor fields of the form $T = Z \otimes \zeta$ with a vector field $Z$ and a $1$-form $\zeta$,
$$
\mr v(Z \otimes \zeta) = \tilde Z (\mr v \zeta)
\,.
$$
According to \cite{BNW1}, the BNW lift of $\lcc$ to $\ctg Q$ is given by the formulae
\ala{
\bnw_{\mr v \zeta} (\mr v \beta)
 & :=
0
\,,\qquad
\bnw_{\mr v \zeta} (\mr h Z)
 :=
0
\,,\qquad
\bnw_{\mr h Z}\left(\mr v \zeta\right)
 :=
\mr v\left(\lcc_Z \zeta\right)
\,,
\\
\bnw_{\mr h Z}(\mr h W)
 & :=
\mr h \left(\lcc_Z W\right)
 +
\mr v
 \left(
\frac 1 2 \Rlcc(Z,W)\,\cdot
 + 
\frac 1 6 \Rlcc(Z,\cdot)W
 +
\frac 1 6 \Rlcc(W,\cdot)Z
 \right)
 }
holding true for all vector fields $Z$, $W$ on $Q$ and $1$-forms $\zeta$, $\beta$ on $Q$. Here, $\Rlcc$ denotes the Riemann curvature tensor of $\lcc$, and the corresponding terms are $1,1$-tensor fieldes on $Q$, viewed as mappings of vector fields, with the dot representing the variable.

\begin{Remark}\label{Bem-BNW}
  The BNW lift can be obtained by standard symplectification, see e.g.\ \cite{BCGRS}, of the complete lift of $\lcc$ to $\ctg Q$
  in the sense of Yano and Patterson \cite{YP1}. This was observed in \cite{Plebanski} and has also been proved in
  \cite{RudSchmiConnection}.
\end{Remark}

Let us determine $\bnw$ for $Q=G^N$ endowed with the Killing metric. It suffices to do this for $Z=\ul X$ and $W=\ul Y$ with $\ul X$ and $\ul Y$ being left-invariant vector fields on $G^N$ and for $\zeta=\ul\xi$ and $\beta=\ul\upsilon$ with $\ul\xi$ and $\ul\upsilon$ being left-invariant $1$-forms on $G^N$. Recall that for such fields, the Levi-Civita connection is given by 
\beq\label{G-lcc}
\lcc_{\ul X} \ul Y = \frac 1 2 [\ul X , \ul Y]
\,,\qquad
\lcc_{\ul X} \ul\xi = \frac 1 2 \ad^\ast(\ul X) \ul\xi
\,.
\eeq
As a preparation, we derive the lifting operators. Clearly, the vertical lift of a left-invariant $1$-form $\ul\xi$ on $G^N$ is given by 
\beq\label{G-lift-1Fm}
(\mr v \ul\xi)_{(\ul a,\ul\alpha)}
= 
\left( \ul a , \ul \alpha , 0 , \ul \xi \right)
\,.
\eeq
To find the horizontal lifting operator $\mr h$, we have to compute the connection mapping $\Klcc$. We use that 
$$
\Klcc(\zeta' Z) = \lcc_Z \zeta
$$
for any vector field $Z$ and any $1$-form $\zeta$ on $G^N$ \cite[Prop.\ 1.5.6]{BuchII}, and that $\Klcc$ acts on $\tg_{(\ul a,\ul \alpha)}(\ctg_{\ul a} G^N)$ as the natural identification with $\ctg_{\ul a} G^N$. We find
\ala{
\Klcc(\ul a,\ul \alpha,\ul X,\ul \xi)
& = 
\Klcc(\ul a,\ul \alpha,\ul X,0) + \Klcc(\ul a,\ul \alpha,0,\ul \xi)
\\
& = 
\Klcc(\ul \alpha' (\ul a, \ul X) ) + (\ul a,\ul \xi)
\\
& = 
\left(\ul a,\ul \xi + \frac 1 2 \ad^\ast(\ul X)\ul \alpha\right)
\,.
}
Hence, from \eqref{G-hor} we read off that for left-invariant vector fields $Z=\ul X$, the horizontal lift is given by
\beq\label{G-lift-VF}
(\mr h \ul X)_{(\ul a,\ul\alpha)} 
= 
\left(\ul a , \ul \alpha , \ul X , - \frac 1 2 \ad^\ast(\ul X) \ul \alpha\right)
\,.
\eeq

\bsz\label{S-Lift}

For $\ul a \in G^N$, $\ul\alpha \in \mf g^\ast{}^N$ and $\ul X \in \mf g^N$, $\ul\xi \in \mf g^\ast{}^N$,
\al{
\label{G-lift-hv}
\left(\bnw_{\mr h \ul X} \mr v \ul \xi\right)_{(\ul a,\ul\alpha)}
& =
\left(
\ul a
,
\ul\alpha
,
0
,
\frac 1 2 \ad^\ast(\ul X) \, \ul\xi
\right)
,
\\
\label{G-lift-hh}
\left(\bnw_{\mr h \ul X} \mr h \ul Y\right)_{(\ul a,\ul\alpha)}
& =
\left(
\ul a
,
\ul\alpha
,
\frac 1 2 [\ul X,\ul Y]
,
\frac 1 6 \ad^\ast(\ul Y) \ad^\ast(\ul X) \, \ul\alpha
-
\frac{1}{12} \ad^\ast(\ul Y) \ad^\ast(\ul X) \, \ul\alpha
\right)
\,.
}

\esz

\bbw

Eq.\ \eqref{G-lift-hv} is a direct consequence of \eqref{G-lcc} and \eqref{G-lift-1Fm}, because $\ad^\ast(\ul X)\,\ul\xi$ is a left-invariant $1$-form on $G^N$, so that \eqref{G-lift-1Fm} applies. To prove Eq.\ \eqref{G-lift-hh}, it remains to calculate the vertical lifts of the curvature terms. For that purpose, we observe that 
$
\mr v(T) \, \tilde Z = \big(T(Z)\big)^\sim
$
for all $1,1$-tensor fields $T$ and all vector fields $Z$ on a manifold $Q$ and that 
$$
\Rlcc(\ul X,\ul Y) = - \frac 1 4 \ad\big([\ul X,\ul Y]\big)
$$
for all left-invariant vector fields $\ul X$, $\ul Y$ on $G^N$. Using this, we check that
\al{\label{G-vRXY.}
\mr v\big(\Rlcc(\ul X,\ul Y)\big)_{(\ul a,\ul\alpha)}
& =
\left(\ul a,\ul\alpha,0,\frac 1 4 \ad^\ast\big([\ul X,\ul Y]\big) \, \ul\alpha \right)
,
\\
\label{G-vRX.Y}
\mr v\big(\Rlcc(\ul X,\cdot)\ul Y\big)_{(\ul a,\ul\alpha)}
& =
\left(\ul a,\ul\alpha,0,\frac 1 4 \ad^\ast(\ul X) \ad^\ast(\ul Y) \, \ul\alpha \right)
.
}
Now, \eqref{G-lift-hh} follows by plugging \eqref{G-lcc}, \eqref{G-lift-VF}, \eqref{G-vRXY.} and \eqref{G-vRX.Y} into the defining formula for $\bnw_{\mr h \ul X} (\mr h \ul Y)$.
\ebw

Another useful formula can be obtained by calculating $\bnw$ for standard vector fields.

\bsz
\label{S-Lift-SVF}

Let $\ul X$, $\ul Y \in \mf g^N$ and $\ul \xi$, $\ul \upsilon \in \mf g^\ast{}^N$. Then, for all $\ul a \in G^N$ and $\ul\alpha \in \mf g^\ast{}^N$,
\al{\nonumber
\left(\bnw_{(\ul X,\ul \xi)}(\ul Y,\ul \upsilon)\right)_{(\ul a,\ul \alpha)}
 & =
 \left(
\ul a \, , \, \ul \alpha \, , \, \frac 1 2 [\ul X , \ul Y]
 \, , \, 
\frac 1 2 \ad^\ast(\ul X) \ul\upsilon + \frac 1 2 \ad^\ast(\ul Y) \ul\xi
 \right.
\\ \label{G-Zh-Lift}
 & \hspace{2cm}
 \left.
 +
\frac 1 6 \big(\ad^\ast(\ul X) \ad^\ast(\ul Y) + \ad^\ast(\ul Y) \ad^\ast(\ul X)\big) \ul \alpha
 \right)
 .
 }

\esz

\bbw

Choose a basis $\{\ul \ve^I\}$ in $\mf g^\ast{}^N$, expand $\ul \alpha = \ul \alpha_I \ul \ve^I$ (summation convention) and define coefficient functions 
\beq\label{G-star-s-p}
p_I : \ctg G^N \to \RR
 \,,\qquad
p_I(\ul a , \ul \alpha) := \ul \alpha_I
\,.
\eeq
According to \eqref{G-lift-VF},
$$
(\ul X , \ul \xi)
 =
\mr h \ul X
 + 
\mr v \ul \xi
 + 
\frac 1 2  \, p_I \, \mr v\left(\ad^\ast(\ul X) \ul \ve^I\right)
\,.
$$
Plugging this decomposition for $(\ul X , \ul \xi)$ and $(\ul Y , \ul \upsilon)$ into $\bnw_{(\ul X,\ul \xi)}(\ul Y,\ul \upsilon)$, we find
 \ala{
\bnw_{(\ul X,\ul \xi)}(\ul Y , \ul \upsilon)
 & =
\bnw_{\mr h \ul X} (\mr h \ul Y)
 + 
\bnw_{\mr h \ul X} (\mr v \ul \upsilon)
 + 
\frac 1 2 \, \big((\ul X,\ul \xi) p_I\big) \, \mr v\left(\ad^\ast(\ul Y) \ul \ve^I\right)
\\
 & \hspace{5.5cm}
 +
\frac 1 2 \, p_I \, 
\bnw_{\mr h \ul X} 
 \left(
\mr v\left(\ad^\ast(\ul Y) \ul \ve^I\right)
 \right)
\,.
 }
Using the formulae of Proposition \rref{S-Lift} and 
$
(\ul X , \ul \xi)_{(\ul a,\ul \alpha)} p_I
 = 
\ul \xi_I
\,,
$
we obtain the assertion.
\ebw

For later purposes, let us prove that the BNW lift of the Levi-Civita connection defined by the Killing metric on $G^N$ is $G$-invariant. The lifted $G$-action on $\ctg G^N \cong G^N \times \mf g^\ast{}^N$ reads
\beq
\label{G-Trvis-G}
(g,(\ul a,\ul\alpha)) \mapsto \Psi_g(\ul a, \ul \alpha) = (g \ul a g^{-1}, \Ad^*(g) \ul \alpha) \,, \quad g \in G
\eeq
and the induced action on $\tg(\ctg \gp^N) \cong (\gp^N \times \mf g^\ast{}^N) \times (\mf g^N \times \mf g^\ast{}^N)$ via the tangent mapping of $\Psi_g$ is given by 
\beq
\label{G-action-t-B}
\Psi'_g(\ul a, \ul \alpha,\ul X,\ul \xi) 
= 
(g \ul a g^{-1}, \Ad^*(g) \ul \alpha, \Ad(g) \ul X, \Ad^*(g)\ul \xi)
\,.
\eeq

\bsz\label{S-nabla-inv}

The BNW lift of the Levi-Civita connection on $G^N$ defined by the Killing metric is $G$-invariant, that is, $(\Psi_g)_\ast \bnw = \bnw$.

\esz

\bbw

It suffices to show that
$$
\bnw_{(\Psi_g)_*(\ul X,\ul \xi)}(\Psi_g)_*(\ul Y,\ul \upsilon)
=
(\Psi_g')_\ast \left(\bnw_{(\ul X,\ul \xi)}(\ul Y,\ul \upsilon)\right)
$$
for all standard vector fields $(\ul X,\ul\xi)$ and $(\ul Y,\ul\upsilon)$ on $\ctg G^N$. Evaluating both sides at a point $(\ul a,\ul\alpha)$ by means of \eqref{G-Zh-Lift} and \eqref{G-action-t-B} and using the equivariance properties
$$
\ad\big((\Psi_g)_* \ul X\big) = (\Psi_g)_* \circ \ad(\ul X) \circ (\Psi_{g}^{-1})_* 
\,,\quad
\ad^*\big((\Psi_g)_* \ul X\big) = (\Psi_{g}^{-1})^* \circ \ad^*(\ul X) \circ (\Psi_{g})^* \, ,
$$
we obtain the assertion by direct inspection. 
\ebw

In the general case, $G$-invariance of $\bnw$ can be obtained by direct inspection of the defining formulae for $\bnw$ using the equivariance of $\mr h$ and $\mr v$. Alternatively, it follows from the geometric interpretation of $\bnw$ provided by Remark \rref{Bem-BNW}.


\subsection{Fedosov star product}
\label{Fedosov-star}


Now, we are prepared to derive the Fedosov star product of standard order type corresponding to the lifted connection $\bnw$. First, let us briefly recall the Fedosov construction \cite{F}. The starting point is the formal Weyl algebra bundle $W(M)$ over the symplectic manifold $M = \ctg G^N$. Recall that $W(M)$ is fiberwise defined as the $\CC[[\lambda]]$-module 
$$
W(M)_p := \left( \prod_{k = 0}^\infty {\rm S}^k (\ctg_p M) \right) [[\lambda]] \, , 
$$   
that is, elements of $W(M)_p$ must be viewed as formal power series in the parameter 
$\lambda$ and as formal series in the symmetric degree of symmetric tensors over $\ctg_p M $.
Let us denote by ${\cal W}(M) := \Gamma^\infty\big(W(M)\big)$ the corresponding space of sections.
%
%
In the sequel, the basic object will be $W(M)$  tensorized with the bundle 
$ \Lambda^\bullet M$ of exterior forms on $M$, that is, 
$W(M) \otimes  \Lambda^\bullet M$.
This bundle  may be endowed (pointwise) with a natural associative 
and supercommutative product
$$
\mu(a \otimes b) := ab := (f \vee g) \otimes (\alpha \wedge \beta) \, , 
$$
for elements $a = f \otimes \alpha$ and $b = g \otimes \beta$.
Its (graded) algebra of sections is the tensor product 
$\Gamma^\infty\big(W(M) \otimes  \Lambda^\bullet M \big) \cong {\cal W}(M) \otimes_{\calC^\infty (M)} \Omega^\bullet (M)$.
We refer to Section 6.4 of \cite{Wald} for further details. In the next step, one deforms $\mu$ by using a
fiberwise Moyal-type product $a \circ_s b$.
We use the standard ordered type. In the case at hand, it is given in terms of the dual global frames $\{\ul E_I : I \in \mc I\}$ in $\tg G^N$ and $\{\ul\ve^I : I \in \mc I\}$ in $\ctg G^N$ by 
\beq
\label{stand-order-M}
a \circ_s b
 := 
\mu \circ {\mr e}^{ \frac{\lambda}{\imi} 
i_s\left(\mr v\ul \ve^I\right)
 \, \otimes \, 
i_s \left(\mr h\ul E_I\right) } 
a \otimes b
\eeq
(summation convention), where $i_s $ means the operation of symmetric insertion and $\mr v$ and $\mr h$ are given by \eqref{G-lift-1Fm} and \eqref{G-lift-VF}. Formula \eqref{stand-order-M} explains how the connection $\bnw$ enters the Fedosov construction. It is easy to see that the product $\circ_s$ does not depend on the choice of frames. Next, we wish to define the star product of standard ordered type for functions on $M$. For that purpose, we denote by 
$$
  \sigma: 
  {\cal W}(M) \otimes_{\calC^\infty (M)} \Omega^\bullet (M)\to C^\infty(M)[[\lambda]] 
$$ 
the canonical projection onto the part of symmetric and antisymmetric degree zero. 
Now, the key idea of the Fedosov construction consists in distinguishing a subalgebra of 
${\cal W} (M)$ such that $\sigma$ restricted to that subalgebra is bijective. 
Then, the associative product $\circ_s $ may be pulled back to $C^\infty(M)[[\lambda]]$ 
via this bijection yielding an associative $\CC[[\lambda]]$-bilinear product. Such a  subalgebra may be obtained as the kernel of a superderivation 
$D: {\cal W} (M) \to {\cal W} (M) \otimes_{\calC^\infty (M)} \Omega^1 (M)$ of antisymmetric degree one fulfilling $D^2 = 0$, called the Fedosov derivation. It is constructed using the BNW lift $\bnw$, see formula (71) in \cite{BNW1} for the standard order Fedosov derivation $D_s$. Associated with $D_s$, for every $f \in C^\infty (M) [[\lambda]]$, there exists a unique 
element $\tau_s (f) \in \ker D_s \cap {\cal W}(M)$, such that
$\sigma (\tau_s (f) ) = f$ and the  mapping $\tau_s :  C^\infty (M) [[\lambda]] \to {\cal W}(M) $ (called the Fedosov Taylor series) is $\CC [[\lambda]]$-linear. According to Theorem 3.3 in \cite{F} (or Theorem 2 in \cite{BNW1}), it can be determined recursively. Now, the Fedosov star product is defined as follows:
\beq
\label{Fed-star-pr}
f \star g := \sigma( \tau_s(f) \circ_s \tau_s (g)) \, , 
\eeq
for any $f,g \in C^\infty(M)[[\lambda]]$. One can derive an explicit formula for $\circ_s$ in the case of a general cotangent bundle, see Theorem 9 in \cite{BNW1}.

Here, we wish to determine this star product explicitly for the case under consideration. For that purpose, we recall that there is a canonical representation of the star product algebra 
$(C^\infty(M) [[\lambda]], \star)$, called the standard order representation:
\beq
\label{stand-ord-rep}
\rho (f) \psi := i^\ast( f \star \pi^\ast \psi ) \, ,
\eeq
for any $f \in C^\infty(M) [[\lambda]]$ and $\psi \in C^\infty(G^N) [[\lambda]]$. Here, 
$i: G^N  \to M = \ctg G^N$ denotes the canonical embedding via the zero section. Now, a key observation is that the calculations may be performed in the representation $\rho$, see \cite[Cor.~10]{BNW1}. More precisely, this corollary states that the restriction of $\rho$ to the subalgebra of smooth complex-valued functions polynomial in the momenta as well as to the subalgebra of formal power series with coefficients in the functions which are analytic in the fiber variables is injective. 
Thus, let us analyze formula \eqref{stand-ord-rep} for these two classes of functions. 

A function $f \in C^\infty(\ctg Q)$ is called fiber-homogeneous if it is of the form
\beq\label{G-D-fapnmFn}
f = \pi^\ast f^{J_1 \dots J_l} p_{J_1} \cdots p_{J_l} \, ,
\eeq
with a symmetric tensor field $f^{J_1 \dots J_l}$ on $G^N$. Here, $p_I$ denote the coefficient functions
with respect to the global frame $( \ul\ve^I)$ in $\ctg G^N$ given by \eqref{G-star-s-p}.

\bsz\label{S-MultFml}

For fiber-homogeneous functions $f$ of degree $l$, one has
\beq\label{G-star-s-rho}
\rho(f)
 = 
\left(\frac{\lambda}{\imi}\right)^l 
f^{J_1 \dots J_l} \, \ul E_{J_1} \cdots \ul E_{J_l}
\eeq
(symmetric operator ordering). 

\esz

\bbw

Let $\psi \in C^\infty(G^N)$ be given. Using \eqref{stand-ord-rep}, \eqref{Fed-star-pr} and the relation
$
\tau_s \circ \pi^\ast = \pi^\ast \circ \tau_0
\,,
$
where $\tau_0$ is the Fedosov-Taylor series with respect to $\lcc$, we obtain
$$
\rho(f) \psi
 =
\sigma \circ i^\ast \big(\tau_s(f) \circ_s \pi^\ast (\tau_0(\psi))\big)
\,.
$$
By \eqref{stand-order-M}, then
 \ala{
\rho(f) \psi
 & =
\sum_{r=0}^\infty \frac{1}{r!} \left(\frac{\lambda}{\imi}\right)^r 
 \left(
\sigma \circ i^\ast 
 \left\{
i_s(\mr v\ul \ve^{I_1}) \dots i_s(\mr v\ul \ve^{I_r}) \tau_s(f)
 \right\} 
 \right)
\\
 & \hspace{4cm}
 \cdot \left(
\sigma \circ i^\ast 
 \Big\{
i_s\left(\mr h\ul E_{I_1}\right) \dots i_s\left(\mr h\ul E_{I_r}\right)
\pi^\ast(\tau_0(\psi))
 \Big\}\right)
\,.
 }
The second factor yields $\sigma\left(i_s(\ul E_{I_1}) \dots i_s(\ul E_{I_1}) \tau_0(\psi)\right)$. By Theorem 4 in \cite{BNW1}, we have
$\tau_0(\psi) = \mr e^{\Dlcc} \psi$,
where $\Dlcc = \ul \ve^I \vee \lcc_{\ul E_I}$. Since $\sigma$ projects onto degree $0$, in $\mr e^{\Dlcc}$ only the term of order $r$ survives. Thus, 
 \ala{
\rho(f) \psi
 & =
\sum_{r=0}^\infty \frac{1}{r!} \left(\frac{\lambda}{\imi}\right)^r 
 \left(
\sigma \circ i^\ast 
 \left\{
i_s(\mr v\ul \ve^{I_1}) \dots i_s(\mr v\ul \ve^{I_r}) \tau_s(f)
 \right\} 
 \right)
 \left(
\ul E_{I_1} \cdots \ul E_{I_r} \psi
 \right) \,,
 }
where we have used that the first factor is symmetric under permutation of indices. In the first factor, we use $\sigma \circ i^\ast = i^\ast \circ \sigma$. According to \cite[Lem.~7]{BNW1}, given $f$ and $r$, there exist local sections $\vp_I$ such that the term of order $r$ of $\tau_s(f)$ can locally be written as
$$
\tau_s(f)^{(r)} = D^{(r)} f + (\ul\ve^I,0) \vee \vp_I
\,,
$$
where $D = (\ul \ve^I,0) \vee \bnw_{(\ul E_I,0)} + (0,\ul E_I) \vee \bnw_{(0,\ul \ve^I)}$. Using this and
$
\mr v\ul \ve^I = (0,\ul\ve^I) = \frac{\partial}{\partial p_I}
\,,
$
for the first factor we obtain 
$$
i^\ast
 \left(
\frac{\partial}{\partial p_{I_1}} \cdots \frac{\partial}{\partial p_{I_r}} f
 \right)
\,.
$$
In view of \eqref{G-D-fapnmFn}, this trivially vanishes for $r > l$. It vanishes for $r < l$, too, because $i^\ast p_I = 0$. Thus, the first factor yields
$
\delta_{rl} \, l! \, f^{I_1 \dots I_l}
\,.
$
This proves \eqref{G-star-s-rho}. 
\ebw

We immediately read off the following special case.

\bfg
\label{st-rep-lin}

For a function $f$ which is linear in the momenta,
\beq
\label{st-rep-lin-1}
\rho(f) \psi = \frac{\lambda}{\mr i} {\ul X} (\psi), 
\eeq
with $ \ul X$ defined by $f ( \ul a, \alpha) = {\ul \alpha} (\ul X)$. 
\qed

\efg

Following \cite{BNW1}, we first derive a formula for the star product $\star$ of  exponentials of tautological functions of left-invariant vector fields on $G^N$ and then, using this formula, we extend $\star$ to arbitrary functions on $\ctg G^N$.

\ble[Bordemann, Neumaier, Waldmann {\cite[Sec.~8, Lem.~10]{BNW1}}]
\label{sstarprod}

For functions of the form $\efnwt{\ul X}$ with $\ul X$ being a left-invariant vector field on $G^N$, the standard ordered star product is given by
\beq\label{G-star-s-e}
\efnwt{\ul X} \star \efnwt{\ul Y}
=
\efnsim{
\frac{\imi}{\lambda} H(\frac{\lambda}{\imi} \ul X , \frac{\lambda}{\imi}\ul Y)
}
\,,
\eeq
where $H$ denotes the Baker-Campbell-Hausdorff series. 

\ele

\bbw

By \eqref{st-rep-lin-1}, for all $\psi \in C^\infty (G^N)$, we have 
\beq
\label{rho-e-X}
\rho\left(\efnwt{\ul X}\right) \psi 
= 
\exp \left( \circ \frac{\lambda}{\mr i} \ul X \right) \psi
\,,
\eeq
where $\circ $ denotes the composition of vector fields viewed as differential operators on $C^\infty (G)$. Using the representation property and \eqref{rho-e-X}, we calculate
$$
\rho
\left(\efnwt{\ul X} \star \efnwt{\ul Y}\right) 
\psi 
= 
\rho 
\left(
\efnsim{
\frac{\mr i}{\lambda} 
H\left(\frac{\lambda}{\mr i} \ul X, \frac{\lambda}{\mr i} \ul Y \right)
}
\right)
\psi
\,.
$$
In view of Corollary 10 in \cite{BNW1}, this yields the assertion.
\ebw

Define operators
$$
B_m : C^\infty(\ctg G^N) \times C^\infty(\ctg G^N) \to C^\infty(\ctg G^N)
$$
by
\beq
\label{G-star-s-M}
\efnwt{\ul X} \star \efnwt{\ul Y}
 =
\sum_{m=0}^\infty 
\left(\frac{\lambda}{\imi}\right)^m 
B_m\left(\efnwt{\ul X} , \efnwt{\ul Y}\right)
\,.
\eeq
Explicit expressions for $B_m$ will be derived below. 

Now, we can extend formula \eqref{G-star-s-e} to arbitrary functions on $\ctg G^N$.

\bsz[Bordemann, Neumaier, Waldmann {\cite[Sec.~8, Prop.~11]{BNW1}}]\label{S-star-s}

For $f,g \in C^\infty(\ctg G^N)$, one has
\beq\label{G-star-s}
f \star g
 =
\sum_{m=0}^\infty 
\left(\frac{\lambda}{\imi}\right)^m
\sum_{n=0}^m
\frac{1}{n!}
B_{m-n}\left((0,\ul \ve^{J_1}) \cdots (0,\ul \ve^{J_n}) f
 , 
(\ul E_{J_1},0) \cdots (\ul E_{J_n},0) g\right)
\,.
\eeq

\esz

\bbw

We follow the proof of Proposition 11 in loc.~cit. It suffices to check \eqref{G-star-s} for fiber-homogeneous functions $f$ and $g$ of degree $l$ and $k$, respectively. We show that \eqref{G-star-s} holds under application of $\rho$. Since both sides of formula \eqref{G-star-s} belong to the subspace generated by the fiber-homogeneous functions, this formula then follows from \cite[Cor.~10]{BNW1}. By the representation property and by \eqref{G-star-s-rho},
$$
\rho(f \star g)
 =
\rho(f) \rho(g)
 =
\left(\frac{\lambda}{\imi}\right)^{l+k} 
\left(f^{J_1 \dots J_l} \, \ul E_{J_1} \cdots \ul E_{J_l}\right)
\left(g^{I_1 \dots I_k} \, \ul E_{I_1} \cdots \ul E_{I_k}\right)
\,.
$$
Since $f^{J_1 \dots J_l}$ is symmetric under permutation of indices, we can apply the Leibniz rule
to rewrite the right hand side as
$$
\left(\frac{\lambda}{\imi}\right)^{l+k} 
\sum_{n=0}^l
{l \choose n}
f^{J_1 \dots J_l} 
\left(\ul E_{J_1} \cdots \ul E_{J_n} g^{I_1 \dots I_k}\right)
\ul E_{J_{n+1}} \cdots \ul E_{J_l} \ul E_{I_1} \cdots \ul E_{I_k}
\,.
$$
Using the symmetry of $f^{J_1 \dots J_l}$ and \eqref{G-star-s-rho}, we can replace $\ul E_{J_{n+1}} \cdots \ul E_{J_l}$ by
$$
\frac{1}{(l-n)!} \sum_{\sigma \in \mr S_{l-n}} 
\ul E_{J_{n+\sigma(1)}} \cdots \ul E_{J_{n+\sigma(l-n)}}
 =
\left(\frac{\lambda}{\imi}\right)^{n-l} \rho\left(p_{J_{n+1}} \cdots p_{J_l}\right)
\,. 
$$
By analogy, we can replace $\ul E_{I_1} \cdots \ul E_{I_k}$ by
$\left(\frac{\lambda}{\imi}\right)^{-k} \rho\left(p_{I_1} \cdots p_{I_k}\right)$.
Hence,
\ala{
\rho&(f \star g) =
\sum_{n=0}^l
\left(\frac{\lambda}{\imi}\right)^n 
{l \choose n}
f^{J_1 \dots J_l} 
\left(\ul E_{J_1} \cdots \ul E_{J_n} g^{I_1 \dots I_k}\right)
\rho\left(p_{J_{n+1}} \cdots p_{J_l} \star p_{I_1} \cdots p_{I_k}\right)
\,.
 }
Since $p_{J_{n+1}} \cdots p_{J_l}$ and $p_{I_1} \cdots p_{I_k}$ are invariant under the point transformations generated by left translations on $G^N$, we can apply \eqref{G-star-s-M} to get
$$
( p_{J_{n+1}} \cdots p_{J_l}) \star (p_{I_1} \cdots p_{I_k})
 =
\sum_{r=0}^\infty 
\left(\frac{\lambda}{\imi}\right)^r 
B_r\left(p_{J_{n+1}} \cdots p_{J_l} , p_{I_1} \cdots p_{I_k}\right)
\,.
$$
Since for functions $\vp$ on $G^N$ and $h$ on $\ctg G^N$ we have
$
\vp \rho(h) = \rho\big((\pi^\ast \vp) h\big)
\,,
$
and since the differential operators $B_r$ vanish on functions of the form $\pi^\ast \vp$, we obtain
 \ala{
\rho(f \star g)
 & =
\sum_{r=0}^\infty \sum_{n=0}^l
\left(\frac{\lambda}{\imi}\right)^{n+r} 
{l \choose n}
\rho
 \left(
B_r
 \left(
\left(\pi^\ast f^{J_1 \dots J_l}\right) p_{J_{n+1}} \cdots p_{J_l} 
\,,\,
 \right.
 \right.
\\
 & \hspace{5cm}
 \left.
 \left.
\pi^\ast \left(\ul E_{J_1} \cdots \ul E_{J_n} g^{I_1 \dots I_k}\right) 
p_{I_1} \cdots p_{I_k}
 \right)
 \right)
\,.
 }
The second argument of $B_r$ can be rewritten as
$$
 \left(
(\ul E_{J_1},0) \cdots (\ul E_{J_n},0) \left(\pi^\ast g^{I_1 \dots I_k}\right) 
 \right) 
p_{I_1} \cdots p_{I_k}
 =
(\ul E_{J_1},0) \cdots (\ul E_{J_n},0) g
\,,
$$
because $(\ul E_J,0) p_I = 0$. The first argument can be rewritten as
$$
\frac{(l-n)!}{l!} 
 \,
\frac{\partial}{\partial p_{J_1}} \cdots \frac{\partial}{\partial p_{J_n}} f
 =
\frac{(l-n)!}{l!}
 \,
(0,\ul \ve^{J_1}) \cdots (0,\ul \ve^{J_n}) f
\,.
$$
Thus,
$$
\rho(f \star g)
=
\rho
 \left(
\sum_{r=0}^\infty \sum_{n=0}^l
\left(\frac{\lambda}{\imi}\right)^{n+r}
\frac{1}{n!}
B_r
 \left(
(0,\ul \ve^{J_1}) \cdots (0,\ul \ve^{J_n}) f
 ,
(\ul E_{J_1},0) \cdots (\ul E_{J_n},0) g
 \right)
 \right)
.
$$
The summation over $n$ can be extended to $\infty$, because $(0,\ul \ve^{J_1}) \cdots (0,\ul \ve^{J_n}) f = 0$ for $n > l$. Finally, we replace the summation variable $r$ by $m = r+n$. Then,
$$
\rho(f \star g)
=
\rho
 \left(
\sum_{m=0}^\infty \sum_{n=0}^m
\left(\frac{\lambda}{\imi}\right)^m
\frac{1}{n!}
B_{m-n}
 \left(
(0,\ul \ve^{J_1}) \cdots (0,\ul \ve^{J_n}) f
 ,
(\ul E_{J_1},0) \cdots (\ul E_{J_n},0) g
 \right)
 \right)
.
$$
This proves \eqref{G-star-s}.
 \ebw

We use the Baker-Campbell-Hausdorff formula to determine the bi-differential operators $B_m$ explicitly. In what follows, let $\mu : C^\infty(\ctg G^N) \times C^\infty(\ctg G^N)\to C^\infty(\ctg G^N)$ be the multiplication mapping. Writing $\NN$ for the set of nonnegative integers, we define $\mc K_r$ to be the set of all triples 
$\mb k = (\vec k_1,\vec k_2,k\big) \in \NN^\kappa \times \NN^\kappa \times \NN$, where $\kappa \in \NN$, satisfying the conditions
$$
k_{1i} + k_{2i} > 0 \text{ for all } i = 1 , \dots , \kappa
 \,,\qquad
|\vec k_1| + |\vec k_2| + k = r-1
\,,
$$
where $|\vec k_a| = k_{a1} + \cdots + k_{a\kappa}$. Given $\mb k = (\vec k_1 , \vec k_2 , k) \in \mc K_r$, let $\mc B_{\mb k}$ be the set of all pairs $(\tilde I , \tilde J)$, where
 \ala{
\tilde I
 & = 
 \big(
(I_{1,1} , \dots , I_{1,k_{11}})
 , \dots , 
(I_{\kappa,1} , \dots , I_{\kappa,k_{1\kappa}})
 ,
(I_1 , \dots , I_k)
 \big)
\,,
\\
\tilde J
 & = 
 \big(
(J_{1,1} , \dots , J_{1,k_{21}})
 , \dots , 
(J_{\kappa,1} , \dots , J_{\kappa,k_{2\kappa}})
 ,
J
 \big)
\,,
 }
with $I_{i,j} , I_i , J_{i,j} , J \in \mc I$ belonging to the same copy of $G$, i.e., having coinciding first entries. Given $(\tilde I , \tilde J) \in \mc B_{\mb k}$ and $\ul X$, define functions $\ul E_{\tilde I , \tilde J}$ by 
 \ala{
\ul E_{\tilde I , \tilde J}
 & :=
 \frac{
(-1)^\kappa
 }{
(\kappa + 1)(|\vec k_2| + 1) k_{11}! k_{21}! \cdots k_{1\kappa}! k_{2\kappa}! k!
 }
\\
 & \hspace{2cm}
 \bigg(
\ad\left(\ul E_{I_{1,1}}\right) \cdots \ad\left(\ul E_{I_{1,k_{11}}}\right) 
\ad\left(\ul E_{J_{1,1}}\right) \cdots \ad\left(\ul E_{J_{1,k_{21}}}\right) 
 \cdots 
\\
 & \hspace{3cm} 
 \cdots
\ad\left(\ul E_{I_{\kappa,1}}\right)\cdots\ad\left(\ul E_{I_{\kappa,k_{1\kappa}}}\right) 
\ad\left(\ul E_{J_{\kappa,1}}\right)\cdots\ad\left(\ul E_{J_{\kappa,k_{2\kappa}}}\right) 
\\
 & \hspace{2cm} 
\ad\left(\ul E_{I_1}\right)\cdots\ad\left(\ul E_{I_k}\right) 
\ul E_J
 \bigg)^\sim
.
 }

\bsz
\label{S-Bm}

The bidifferential operators $B_m$ are given by 
 \al{\nonumber
B_m
 & =
{\sum_{\vec n}}^\ast
\frac{B_{\vec n}}{|\vec n|!} \prod_{r=2}^\infty 
\sum_{~~\mb k^1 , \dots , \mb k^{n_r} \in \mc K_r ~~} 
\sum_{(\tilde I^i,\tilde J^i) \in \mc B_{\mb k^i}} 
\ul E_{\tilde I^1 , \tilde J^1} \cdots \ul E_{\tilde I^{n_r} , \tilde J^{n_r}} 
\\ \label{G-Mm-ve}
 & \hspace{6cm}
 \, \mu \circ
 \left(
\left(\ul \ve^{\tilde I^1} \cdots \ul \ve^{\tilde I^{n_r}}\right)
 \otimes 
\left(\ul \ve^{\tilde J^1}  \cdots \ul \ve^{\tilde J^{n_r}}\right)
 \right)
\,,
 }
where $\sum_{\vec n}^\ast$ stands for the sum over all finite sequences $\vec n = (n_2 , \dots , n_s)$ of nonnegative integers satisfying $\sum_{r=2}^s (r-1) n_r = m$ and
 \ala{
B_{\vec n} 
& = 
{|\vec n| \choose n_2} {|\vec n| - n_2 \choose n_3} {|\vec n| - n_2 - n_3 \choose n_4} \cdots
\\
\ul \ve^{\tilde I}
 & =
(0,\ul\ve^{I_{1,1}}) \cdots (0,\ul\ve^{I_{1,k_{11}}})
 \cdots\cdots
(0,\ul\ve^{I_{\kappa,1}}) \cdots (0,\ul\ve^{I_{\kappa,k_{1\kappa}}})
(0,\ul\ve^{I_1}) \cdots (0,\ul\ve^{I_k})
\,,
\\
\ul \ve^{\tilde J}
 & =
(0,\ul\ve^{J_{1,1}}) \cdots (0,\ul\ve^{J_{1,k_{21}}})
 \cdots\cdots
(0,\ul\ve^{J_{\kappa,1}}) \cdots (0,\ul\ve^{J_{\kappa,k_{2\kappa}}})
(0,\ul\ve^J)
\,.
 }

\esz

\bbw

Recall that
\beq\label{G-BCH}
H(\ul X,\ul Y) = \ul X + \ul Y + \sum_{r = 2}^\infty H_r(\ul X,\ul Y)
\,,
\eeq
where the Lie algebra elements $H_r(\ul X,\ul Y)$ are given by
\beq\label{G-BCH-r}
H_r(\ul X,\ul Y)
 =
\sum_{\mb k \in \mc K_r} 
(-1)^\kappa
 \frac{
\ad(\ul X)^{k_{11}} \ad(\ul Y)^{k_{21}} \cdots \ad(\ul X)^{k_{1\kappa}} \ad(\ul Y)^{k_{2\kappa}} \ad(\ul X)^k Y
 }{
(\kappa + 1)(|\vec k_2| + 1) k_{11}! k_{21}! \cdots k_{1\kappa}! k_{2\kappa}! k!
 }
\,.
\eeq
Plugging \eqref{G-BCH} into \eqref{G-star-s-e}, we find
$$
\left(\efnwt{\ul X} \star \efnwt{\ul Y}\right) (\ul a , \ul \alpha)
 =
\mr e^{\ul \alpha(\ul X)} \mr e^{\ul \alpha(\ul Y)}
\mr e^{\sum_{r=2}^\infty (\frac{\lambda}{\imi})^{r-1} \ul \alpha(H_r(\ul X , \ul Y))}
\,.
$$
Expanding the last exponential and using the binomial formula, we obtain 
$$
\left(\efnwt{\ul X} \star \efnwt{\ul Y}\right) (\ul a , \ul \alpha)
 =
\mr e^{\ul \alpha(\ul X)} \mr e^{\ul \alpha(\ul Y)}
\sum_{m=0}^\infty 
\left(\frac{\lambda}{\imi}\right)^m
{\sum_{\vec n}}^\ast
\frac{B_{\vec n}}{|\vec n|!} \prod_{r=2}^\infty 
\big(\ul \alpha(H_r(\ul X , \ul Y))\big)^{n_r}
\,,
$$
with $\sum_{\vec n}^\ast$ and $B_{\vec n}$ given as in the proposition. Comparison with \eqref{G-star-s-M} then yields
\beq\label{G-BCH-Mm-1}
B_m(\efnwt{\ul X},\efnwt{\ul Y})(\ul a,\ul \alpha)
 =
\left(\mr e_{\ul X} \mr e_{\ul Y}\right)(\ul a,\ul \alpha)
{\sum_{\vec n}}^\ast
\frac{B_{\vec n}}{|\vec n|!} \prod_{r=2}^\infty 
\big(\ul \alpha(H_r(\ul X , \ul Y))\big)^{n_r}
\,.
\eeq
To read off a formula for $B_m$ in terms of a bidifferential operator, we expand $\ul X$ and $\ul Y$ with respect to the basis $\{\ul E_I\}$ in $\mf g^N$ and plug this into \eqref{G-BCH-r}. In the condensed notation
\ala{
\ul X^{\tilde I}
 & :=
\ul X^{I_{1,1}} \cdots \ul X^{I_{1,k_{11}}}
 \cdots\cdots
\ul X^{I_{\kappa,1}} \cdots \ul X^{I_{\kappa,k_{1\kappa}}}
\ul X^{I_1} \cdots \ul X^{I_k} \ ,
\\
\ul Y^{\tilde J}
 & :=
\ul Y^{J_{1,1}} \cdots \ul Y^{J_{1,k_{21}}}
 \cdots\cdots
\ul Y^{J_{\kappa,1}} \cdots \ul Y^{J_{\kappa,k_{2\kappa}}}
\ul Y^J
,
}
this yields
$$
\ul \alpha(H_r(\ul X , \ul Y))
 =
 \left(
\sum_{\mb k \in \mc K_r} 
\sum_{(\tilde I , \tilde J) \in \mc B_{\mb k}} 
\ul E_{\tilde I , \tilde J} \ul X^{\tilde I} \ul Y^{\tilde J}
 \right)
(\ul a, \ul \alpha)
$$
and thus
 \al{\nonumber
\big(\ul \alpha&(H_r(\ul X , \ul Y))\big)^{n_r}
\\ \label{G-BCH-aHr}
 & =
 \left(
\sum_{\mb k^1 , \dots , \mb k^{n_r} \in \mc K_r ~~} 
\sum_{(\tilde I^i,\tilde J^i) \in \mc B_{\mb k^i}} 
\ul E_{\tilde I^1 , \tilde J^1} \cdots \ul E_{\tilde I^{n_r} , \tilde J^{n_r}} 
\ul X^{\tilde I^1} \ul Y^{\tilde J^1} \cdots \ul X^{\tilde I^{n_r}} \ul Y^{\tilde J^{n_r}}
 \right)
(\ul a, \ul \alpha) \ .
}
Plugging \eqref{G-BCH-aHr} into \eqref{G-BCH-Mm-1} and using that $\ul X^I \efnwt{\ul X} = (0 , \ul\ve^I)\efnwt{\ul X}$, we obtain the assertion.
\ebw

\begin{Remark}\label{Bem-Mr}
For $B_0$, $B_1$ and $B_2$, we obtain 
 \ala{
B_0 & = \mu
\,,
\\
B_1
 & = 
\frac 1 2 \,\, [\ul E_I , \ul E_J]^\sim
 \,\,
\mu \circ \left((0,\ul\ve^I) \otimes (0,\ul\ve^J)\right)
\,,
\\
B_2
 & = 
\frac{1}{24} 
 \left\{
2 \, [\ul E_I , [\ul E_J , \ul E_K]]^\sim \,\,
 \mu \circ
\left(
(0,\ul\ve^I) (0,\ul\ve^J) \otimes (0,\hat\ve^K)
 - 
(0,\ul\ve^K) \times (0,\ul\ve^I) \hat\ve^J
 \right)
 \right.
\\
 & \hspace{2.5cm} +
 \left.
3 \, [\ul E_I , \ul E_L]^\sim \, [\ul E_K , \ul E_L]^\sim \,\,
\mu
 \circ
 \left(
(0,\ul\ve^I) (0,\ul\ve^K) \otimes (0,\ul\ve^J) (0,\ul\ve^L)
 \right)
 \right\}
\,.
 }
\end{Remark}



\section{Homological reduction }
\subsection{The method}

Classical homological reduction of a $G$-Hamiltonian system essentially goes back to the work of
Batalin--Fradkin--Vilkoviski \cite{BF,BV77,BV83,BV85} and was later interpreted mathematically
in terms of the tensor product of a Koszul-Tate resolution of the constraint ideal with the
Chevalley--Eilenberg complex of the Lie algebra of the symmetry group \cite{McMullan,Stasheff}.
In the case of a regular $G$-Hamiltonian system Bordemann--Herbig--Waldmann \cite{BHW} constructed
a star product on the reduced symplectic space via homological perturbation of the classical 
homological reduction \'a la Batalin--Fradkin--Vilkoviski; see also \cite{GW,EspKraWal}.
In \cite{Reichert}, Reichert relates the characteristic classes of the unreduced with
the reduced star product and thus shows that, under reasonable
assumptions on the initial data of the Hamiltonian system, deformation quantization commutes
with homological reduction. The method from \cite{BHW} was generalized by Herbig \cite{HPhD} and
Bordemann--Herbig--Pflaum \cite{BHP} to the singular case under the condition that the zero level set
is a complete intersection and that its vanishing ideal is generated by the components
of the moment map. Let us explain the main ideas behind classical homological reduction and its quantized version 
within the framework of deformation theory. For the necessary
tools from homological algebra and homological perturbation theory we refer the reader to \cite{HilStaCHA,GelMan,Weibel,LamHPT,CrainicPerturbation}
and to Appendix \ref{ToolsHomAlg}.

Assume that $(M,\omega,\Psi,J)$ is a $G$-Hamiltonian system where $G$, as before, is assumed to be a compact Lie group. Denote by  $\pi : M_0 \to M\red G$ the canonical projection from the the zero level set
$M_0 = J^{-1}(0)$ onto the symplectically reduced space. The reduced phase space $M\red G$ becomes in a
natural way a commutative locally ringed space with structure sheaf
$\calC^{\infty}_{M\red G}$ given by
\[
  \calC^{\infty}_{M\red G}  (U)=\left( \calC^\infty (\widetilde{U})\right)^G \big/ \left(\calI_{M_0} (\widetilde{U})\right)^G \ .
\]
Here, $U$ runs through the open sets of $M\red G$, $\widetilde{U}$ denotes for given $U \subset M\red G$
an open subset of $M$ such that $\widetilde{U}\cap M_0 = \pi^{-1} (U)$,
$\calI_{M_0} \subset \calC^\infty_M$ is the vanishing ideal sheaf of the constraint surface,
and $\big( - \big)^G$ denotes the $G$-invariant part. One can prove that the ringed space
$\left(M\red G, \calC^{\infty}_{M\red G} \right)$ is a differentiable space in the sense of Spallek \cite{SpallekDR},
cf.~also \cite{NGonzalezSanchoBook}, and that it has a natural minimal Whitney stratification \cite{SjamaarLerman}.
More importantly from the point of view of geometric mechanics is the observation by Sjamaar and Lerman \cite{SjamaarLerman}
that the so-called algebra of smooth functions 
$\calC^{\infty} (M\red G) := \calC^{\infty}_{M\red G} (M\red G)$ on the reduced space carries a Poisson structure
\[
  \big\{ - , - \big\}_{\! M\red G} : \calC^{\infty} (M\red G) \times \calC^{\infty} (M\red G) \to \calC^{\infty} (M\red G) \ .
\]   
This Poisson structure is uniquely determined by the condition that it is compatible with the natural Poisson
bracket $\big\{ - , - \big\}_{\! M}$ on the symplectic manifold $(M,\omega)$.
This means that the Poisson bracket of two elements $f,g\in \calC^{\infty} (M\red G)$
is given by
\begin{equation}
  \label{eq:Poissonbracketreducedspace}
  \big\{ f,g\big\}_{\! M\red G} \circ \pi =
  \big\{\widetilde{f},\widetilde{g}\big\}_{\! M} \big|_{J^{-1} (0)} \ , 
\end{equation}
where  $\widetilde{f},\widetilde{g} \in \calC^{\infty} (M)$
are chosen to be $G$-invariant and to satisfy
$\widetilde{f} |_{M_0}  = f\circ \pi$ and $\widetilde{g} |_{M_0}  = g\circ \pi$. It was shown in \cite{SjamaarLerman} that the strata $S$ of the natural stratification of $M\red G$ are symplectic manifolds and that the embeddings $(S,\calC^\infty_S) \to \left(M\red G, \calC^{\infty}_{M\red G} \right)$ are Poisson.

In homological reduction, the so constructed Poisson algebra of smooth functions on a
symplectically reduced space is expressed in terms of the zeroth cohomology
of a certain cochain complex carrying the structure of a graded Poisson algebra.
Under certain assumptions, the latter can be deformed along the graded Poisson
structure and the zeroth cohomology of the deformed algebra is a deformation
quantization of the original Poisson algebra.  
Before we can describe the details of this method we need the following.

\subsection{A tool combining real algebraic with symplectic geometry}
A crucial ingredient for homological reduction to work in the singular case is a
certain solution to (a variant of) the so-called extension
problem in real algebraic geometry, cf.~\cite{WhitneyAEDFDCS,WhitneyEDF,MalIDF,Schwarzbier,Fefferman}.
By that one understands the following. Assume that $Z$ is a closed subset of a
smooth manifold $M$, $I_Z\subset \calC^\infty (M)$ the vanishing ideal, and
$\rest : \calC^\infty (M) \to \calC (Z)$, $f\mapsto f|_Z$ the restriction map.
Then $I_Z$ is a closed ideal, so one obtains a short exact sequence of
Fr\'echet algebras
\[
  0 \longrightarrow I_Z \longrightarrow \calC^\infty (M)
  \overset{\rest}{\longrightarrow} \calC^\infty (Z) \longrightarrow 0 \ ,
\]
where $\calC^\infty (Z) \subset \calC (Z)$ denotes the image of $\rest$ equipped
with the quotient topology. The question now arises under which conditions on
$M$ and $Z$ this sequence has a continuous split, meaning that a continuous map
$\ext:  \calC^\infty (Z) \to \calC^\infty (M)$ exists such that
$\rest \circ \ext = \id$. If such a continuous split exists, one says that
$Z\subset M$ has the \emph{extension property} \cite[Sec.~7.1]{Schwarzbier}.
According to the solution of the extension problem by Bierstone and Schwarz \cite[Thm.~0.2.1]{Schwarzbier},
every Nash subanalytic subset $Z$ of a  real analytic manifold $M$ has the extension property;
see \cite[Def.~0.1.2]{Schwarzbier} for the definition of Nash subanalytic sets. 
Note that every semianalytic hence every analytic subset of a real analytic manifold
is Nash subanalytic by \cite[\S 17]{Lojasiewicz}. 

Two important results which entail that the  extension theorem by Bierstone and Schwarz can be applied to our
situation are the observation by Kutzschebauch and Loose \cite{KutzschebauchLoose} that every symplectic manifold
carries a real analytic structure in which the symplectic form is real analytic and
\cite[Theorem 1.3]{MatumotoShiota} by Matumoto and Shiota that every smooth manifold with a compact Lie group action carries an analytic structure
in which the $G$-action is real analytic, see also \cite{Illman}. Note that in either case
the real analytic structure is not unique but only unique up to isomorphism. Therefore it is not immediately
clear that a real analytic structure on the underlying space of a given $G$-Hamiltonian system
can be chosen so that both the group action and the symplectic form are real analytic.
Below we show that this is indeed the case. We also verify that, as a consequence, the moment map
of a $G$-Hamiltonian system equipped with such a compatible real analytic structure
is real analytic as well, so its  zero level set is analytic and therefore
has the desired extension property. Note that hereby we assume that all manifolds are
second countable.

\begin{Theorem}
  \label{Thm:ExCompatibleEquivariantAnalyticStructures}
  Let $(M,\omega)$ be a symplectic manifold. Then the following holds true:
  \begin{enumerate}[{\rm (i)}]
  \item\label{ite:analyticstr}
    There exists a real analytic structure on $M$ that means an atlas of $M$
    with real analytic transition maps in regard to which $\omega$ becomes a real
    analytic $2$-form.
  \end{enumerate}
  Under the assumption that $G$ is a compact Lie group with a Hamiltonian action $\Psi$ on
  $M$ and $J: M \to \g^*$ the corresponding moment map the following additional
  statements are satisfied:
  \begin{enumerate}[{\rm (i)}]
  \setcounter{enumi}{1}
  \item\label{ite:analyticaction}
    The real analytic structure in {\rm (\ref{ite:analyticstr})}
    can be chosen so that the $G$-action on $M$ and the symplectic form $\omega$ are real analytic.
    In regard to such a real analytic structure the moment map $J$ is
    real analytic as well. 
  \item\label{ite:zerolevelextprop}
    The zero level set $M_0 = \{p \in M: \mm(p) =0 \}$ has the extension property.
    Moreover, the extension map $\ext:  \calC^\infty (M_0) \to \calC^\infty (M)$
    can be chosen to be equivariant.
  \end{enumerate}
\end{Theorem}

To prove the theorem, we need some preliminary results. As before, $G$ will always denote a compact
Lie group with its canonical real analytic structure.
%
%
Recall first from \cite[Chapter 2]{Hirsch} or \cite{KutzschebauchLoose} the definition of the Whitney topology on $\calC^\infty (M)$ for  a smooth $n$-dimensional manifold $M$.
Let $\Phi = (U_i,x_i=(x_i^1,\ldots,x_i^n))_{i\in I}$ with $I\subset \N$ be a locally finite smooth atlas of $M$, $K =(K_i)_{i\in I}$
a family of compact subsets $K_i\subset U_i$, $m = (m_i)_{i\in I}$ a family of positive integers, and
$\varepsilon = (\varepsilon_i)_{i\in I}$ a family of positive real numbers. We call such a quadruple $(\Phi,K,m,\varepsilon)$
a \emph{limiting cover} of $M$. 
Associated to every limiting cover and every $f\in \calC^\infty (M)$ is the basic neighborhood
\[
  N (f; \Phi,K,m,\varepsilon) =
  \left\{ g \in \calC^\infty (M) : \sup_{p\in K_i \atop \alpha \in \N^n, \: |\alpha|\leq m_i}
  \left| \frac{\partial^{|\alpha|} g }{\partial x^\alpha}(p) - \frac{\partial^{|\alpha|} f}{\partial x^\alpha}(p)\right|
  \leq  \varepsilon_i \text{ for all } i \in I\right\} \ .
\]
One verifies that the basic neighborhoods $N (f; \Phi,K,m,\varepsilon)$ where $f$ runs through the
elements of $\calC^\infty (M)$ and $(\Phi,K,m,\varepsilon)$ through the limiting covers of $M$ forms
a basis of a  topology.

The topology generated by this basis on $\calC^\infty (M)$ is  translation invariant by construction. 
It is called the \emph{Whitney topology}.
The definition of the Whitney topology can be extended in a straightforward way to
the space $\Omega^k (M)$ of smooth $k$-forms on $M$. A fundamental observation by Whitney
\cite[Lem.~6]{WhitneyAEDFDCS} was that for an open subset $U\subset \R^n$ the space of real analytic
functions on $U$ is dense in $\calC^\infty (U)$ with respect to the Whitney topology.
More generally, the Grauert--Morrey embedding theorem \cite{Grauert} together with Whitney's
result imply that $\calC^\omega (M)$ is dense in $\calC^\infty (M)$ in the Whitney topology
for any real analytic manifold $M$, see \cite[Thm.~13.4]{Illman}.

\begin{Lemma}
  \label{Lem:KL2}
  Let $M$ be a real analytic manifold equipped with an analytic $G$-action, $k \in \N_{>0}$,
  and $N \subset \Omega^{k-1} (M)$ an open zero neighborhood in the Whitney topology.

  If $\omega$ is a smooth and $G$-invariant closed
  $k$-form on $M$, then there exists an invariant $\theta\in N$ such that
  $\omega^\textup{a} = \omega - d\theta$ is real analytic. In particular this means that
  one can find a $G$-invariant analytic representative within the 
  de Rham cohomology class of $\omega$.
\end{Lemma}

\begin{proof}[Proof of Lemma \ref{Lem:KL2}]
  Consider the averaging operator $A:\Omega^\bullet (M) \to \Omega^\bullet (M)$ which is defined by
  integration with respect to the normalized Haar measure on $G$:
  \begin{equation}
    \label{eq:averagingop}
     A \varrho = \int_G {\mr L}^*_g \varrho \, dg \quad \text{for all } \varrho \in  \Omega^k (M) \ .
  \end{equation}
  The operator $A$ then is a projection onto the space of $G$-invariant forms, commutes with the
  exterior differential, and maps real analytic forms to real analytic forms by \cite[Prop.~14.4]{Illman}.
  By \cite[Theorem 15.4]{Illman}, $A:\Omega^\bullet (M) \to \Omega^\bullet (M)$
  is also continuous with respect to the Whitney topology. So $N_1 = A^{-1}N$ is
  a zero neighborhood, and there exists, by \cite[Lem.~2]{KutzschebauchLoose},
  an element $\theta_1 \in N_1$ so that $\omega_1 = \omega - d\theta_1$ is real analytic.
  Then $\omega^\textup{a} = A\omega_1$ is  $G$-invariant by construction and real analytic by
  \cite[Prop.~14.4]{Illman}. Moreover,  $\omega^\textup{a} = \omega  - d \theta$, where
  $\theta = A\theta_1 \in N$.   
\end{proof}

\begin{Lemma}
\label{Lem:DerAnaImpliesAna}
  Let $M$ be a real analytic manifold and $f : M \to \R$ a smooth function
  such that $df$ is a real analytic $1$-form. Then $f$ is real analytic.
\end{Lemma}

\begin{proof}[Proof of Lemma \ref{Lem:DerAnaImpliesAna}]
  Since the problem is local, it suffices to assume that $M$ is an
  open subset $U$ of some $\R^n$. Recall, for example from
  \cite[Prop.~2.2.10]{Krantz}, the well-known criterion for real analyticity
  which says that $g: U \to \R$ is real analytic if and only if for each $a\in U$
  there exists an open ball $V\subset U$ around $a$ together
  with constants $C,R>0$ such that
  \[
    \left| \frac{\partial^{|\alpha|} g}{\partial x^\alpha} (v)\right|
    \leq C \frac{\alpha ! }{R^{|\alpha|}} \quad \text{for all } v \in V,\:
    \alpha \in \N^n \ . 
  \]
  Here $x=(x^1,\ldots,x^n)$ denote the standard coordinates of $\R^n$.
  Now assume that $f : U \to \R $ is smooth and that the partial derivatives
  $\partial_i f = \frac{\partial f}{\partial x^i} : U \to \R$, $i=1,\dots ,n$ are real
  analytic. By the mentioned criterion there exist for every point
  $a\in U$ open balls  $V_i\subset  U$ around $a$ and constants
  $C_i,R_i>0$, $i=1,\dots ,n$, such that
  \[
    \left| \frac{\partial^{|\alpha|} \partial_i f}{\partial x^\alpha} (v)\right|
    \leq C_i \frac{\alpha ! }{R_i^{|\alpha|}} \quad \text{for all } v \in V_i,\:
    \alpha \in \N^n \ . 
  \]
  Choose an open ball $V$ relatively compact in $U$ such that $a \in V \subset V_1 \cap \ldots \cap V_n$
  and put $R =\min \{R_1,\ldots ,R_n \} $. Choose $C>0$ which is larger than $\sup_{v\in \overline{V}} \{ |f(v)| \}$ and
  larger than each of the products $R \cdot C_i$.
  Then the estimate 
  \begin{equation*}
    \label{eq:criterion}
    \left| \frac{\partial^{|\alpha|} f}{\partial x^\alpha} (v)\right|
    \leq C \frac{\alpha ! }{R^{|\alpha|}} \ , \quad v \in V  \ ,
  \end{equation*}
  holds true for $\alpha=0$ by definition of $V$ and $C$. Let us show that it also holds for non-zero
  $\alpha \in \N^n$. Then $\alpha_j >0$ for some $j  \in \{ 1, \ldots , n\}$. Put
  \[
    \beta_i =
    \begin{cases}
      \alpha_i & \text{if } i \neq j \\
      \alpha_j - 1 & \text{if } i=j \ .
    \end{cases}
  \]
  One obtains
  \[
    \left| \frac{\partial^{|\alpha|} f}{\partial x^\alpha} (v)\right|
    =
    \left| \frac{\partial^{|\beta|} \partial_j f}{\partial x^\beta} (v)\right|
    \leq C_j  \frac{\beta ! }{R_j^{|\beta|}} \leq C \frac{\alpha ! }{R^{|\alpha|}} \quad \text{for all } v \in V  \ ,
  \]
  hence $f$ satisfies the analyticity criterion and the claim is proved.
\end{proof}

\begin{proof}[Proof of Theorem \ref{Thm:ExCompatibleEquivariantAnalyticStructures}]
  {\it ad} {\rm (\ref{ite:analyticstr})}.
  This has been proved in \cite{KutzschebauchLoose}. The main idea in that work was to verify
  a non-equivariant version of Lemma \ref{Lem:KL2} and then apply Moser's trick. We generalize
  this ansatz to the equivariant case.
  \newline
  {\it ad} {\rm (\ref{ite:analyticaction})}.
  By \cite[Theorem 1.3]{MatumotoShiota} there exists an analytic structure on $M$ with respect to which
  the  $G$-action $\Psi$ is real analytic.
  To show the claim it now suffices to construct an analytic $G$-invariant symplectic form
  $\omega^\textup{a}$ on $M$  and a $G$-equivariant diffeomeorphism $f:M\to M$ so that
  $f^*\omega^\textup{a} = \omega$. Following \cite{KutzschebauchLoose} we will apply Moser's trick
  to construct $f$. First choose a zero neighborhood $N$ in the Whitney topology on $\Omega^1 (M)$
  so that $\omega_t = \omega - t d \theta$ is a non-degenerate $2$-form for all $\theta \in N$
  and $t\in [0,1]$. For each such  $\theta$ and $t$ there then exists a uniquely defined smooth vector
  field $X_t: M \to TM$ so that
  \[
                    X_t \lrcorner\, \omega_t = \theta \ .
  \]
  Note that $X_t$ depends smoothly on $t$.
  After possibly shrinking the neighborhood $N$ one can achieve that the non-autonomous
  vector field $X_t$ is integrable up to $t=1$ which means that there exists a family of diffeomorphism
  $(\varphi_t)_{t\in[0,1]}$ of $M$ which is smooth in $t$ so that $f_0=\id_M$ and
  \[
        \frac{d}{dt} \varphi_t = X_t \circ \varphi_t \quad \text{for all } t\in [0,1] \ .
  \]
  Note that $X_t$ and hence $\varphi_t$ are $G$-equivariant in case $\theta$ is $G$-invariant. 
  By Lemma \ref{Lem:KL2} one can now find a real analytic $G$-invariant form $\theta \in N$
  so that $\omega^\textup{a} = \omega - d\theta$ is real analytic. By construction,
  $\omega^\textup{a}$ then has to be $G$-invariant as well. Moreover, the vector fields $X_t$
  and the diffeomeorphisms $\varphi_t$ are $G$-equivariant as well for all $t\in [0,1]$.
  By Moser's trick, 
  \[
     \omega = \varphi_1^* \omega^\textup{a}
  \]
  and the first claim of {\rm (\ref{ite:analyticaction})} is proved.
  Since the moment map satisfies $d\mm_Z = - Z_M \lrcorner \, \omega$ for all
  $Z \in \g$ and since both the $G$-action and $\omega$ are real analytic the remaining claim
  now follows from Lemma \ref{Lem:DerAnaImpliesAna}.
  \newline
  {\it ad} {\rm (\ref{ite:zerolevelextprop})}.
  Choose the analytic structure as in {\rm (\ref{ite:analyticaction})}. Then $M_0 = \mm^{-1} (0)$
  is an analytic subset of $M$, hence is Nash subanalytic by \cite[\S 17]{Lojasiewicz} and so
  has the extension property by \cite[Thm.~0.2.1]{Schwarzbier}.
  By averaging over the unique normalized Haar measure on $G$ one can achieve that the
  extension map $\ext :\calC^\infty (M_0)\to\calC^\infty (M)$ is $G$-equivariant.
\end{proof}

\subsection{Classical homological reduction}

Next we explain algebraic reduction \cite{AGJ} which underlies classical
homological reduction. Observe that by definition of the constraint
surface the functions
$\mm_Z = \langle \mm (-),Z \rangle$ with $Z \in\g^*$ vanish on the constraint surface. The ideal
$I(\mm) \subset \calC^\infty $    
generated by these functions $\mm_Z$ is contained in the vanishing ideal
$\calI_{M_0} (M) = \{ f \in \calC^\infty (M) : f|_{M_0} = 0 \} $ which we will denote from now on 
by $I_{M_0}$ as in \cite{AGJ}. Equality of the ideals $I(J)$ and $I_{M_0}$ then
holds under the following condition. 
\begin{enumerate}[(A\text{H})]
\setcounter{enumi}{6}  
\item\label{GenCond}
  \textbf{Generating Hypothesis.}\abs
  The functions $\mm_Z$ with $Z \in \la$ generate the vanishing ideal 
  $I_{M_0}$  of the constraint surface.
\end{enumerate}  
Note that a generating system of $I(J)$ is also given by the components
$J_l : M \to \R$, $l=1,\ldots , d$ of the representation $J = \sum_{l=1}^d J_l \, \varepsilon^l $
in terms of a basis $ (\varepsilon^1,\ldots , \varepsilon^d)$ of the dual 
$\g^*$. In classical homological reduction the Poisson algebra
$\left( \calC^{\infty} (M\red G) , \big\{ - , - \big\}_{M\red G} \right)$ is expressed -- under the assumption
of the generating condition and acyclicity of the Koszul complex on $J$ -- as the zeroth cohomology of the so-called BRST complex constructed below. 
In addition to being a differential graded algebra, the BRST complex carries  a graded Poisson structure which
it inherits from the natural Poisson bracket on $\calC^\infty (M)$. The particular virtue of the BRST complex
now is that it admits under the assumptions made a formal deformation quantization which leads to a star
product on the reduced phase space.

The first ingredient to the BRST complex is the Koszul complex
$\left( K_\bullet (\calC^\infty(M),J),\partial \right)$ on the map $J : M \to \g^*$, see
Example \ref{AppKoszulCplx}.
Its degree $k$ component is the free $\calC^\infty (M)$-module
\[ K_k = K_k  (\calC^\infty(M),J) = \calC^\infty (M) \otimes \wedge^k \g \ ,  \]
and the differential is given by contraction with $J$:
\[
  \partial : K_k (\calC^\infty(M),J)  \to  K_{k-1} (\calC^\infty(M),J), \:
  \alpha \mapsto \langle J, \alpha \rangle = \sum_{l=1}^d J_l (\varepsilon^l \lrcorner \, \alpha)  \ .
\]
As before, $(E_1,\ldots , E_d)$ denotes here a basis of the Lie algebra $\g$,
$(\varepsilon^1,\ldots , \varepsilon^d)$ its dual basis in $\g^*$, and the
$J_l \in \calC^\infty (M)$, $l=1,\ldots , d$ are the uniquely determined maps so that
$J = \sum_{l=1}^d J_l \varepsilon^l $.
The second condition  on the $G$-Hamiltonian system which
is needed to entail that the zeroth homology of the Koszul complex coincides with the algebra
$\calC^\infty (M_0)$ of smooth functions on the constraint surface is the following:
\begin{enumerate}[(A\text{C})]
\setcounter{enumi}{0}  
\item\label{AcycCond}
  \textbf{Acyclicity Condition.}\abs
  The Koszul complex $\left( K_\bullet (\calC^\infty(M),\mm),\partial\right)$ is acyclic.
\end{enumerate}

\begin{Proposition}
  Let $(M,\omega,\Psi,J)$ be a $G$-Hamiltonian system which satisfies conditions  \hyperref[GenCond]{{\rm(GH)}} and
  \hyperref[AcycCond]{{\rm (AC)}}. Then the complex
  \[
    0 \longrightarrow K_d \overset{\partial}{\longrightarrow} \ldots
    \overset{\partial}{\longrightarrow} K_1 \overset{\partial}{\longrightarrow}  K_0 = \calC^\infty (M)
    \longrightarrow \calC^\infty (M_0) \longrightarrow 0
  \]  
  is contractible, so $K_\bullet (\calC^\infty (M), J)$ is a free resolution of
  \[
    \calC^\infty (M_0) = \calC^\infty (M) / I_{M_0} \cong  H_0 \big(K_\bullet (\calC^\infty(M),\mm)\big) 
  \]
  in the category of $\calC^\infty(M)$-modules.
\end{Proposition}

\begin{proof}
  This is immediate by definition of the Koszul complex, since $ I_{M_0} = I(J)$ by the generating hypothesis and
  since $\im \left( \partial: K_1 \to K_0 \right) = I(J)$ by the acyclicity condition.
\end{proof}

In \cite{BHP} it was observed that under the assumptions \hyperref[GenCond]{{\rm(GH)}} and
\hyperref[AcycCond]{{\rm (AC)}} the Koszul complex allows for a contracting homotopy consisting of
linear maps continuous with respect to the natural Fr\'echet topologies on $\calC^\infty (M)$ and
its quotient $\calC^\infty (M_0)$. Here we provide a strengthening
  of that result. By virtue of Theorem \ref{Thm:ExCompatibleEquivariantAnalyticStructures}, a $G$-Hamiltonian system
  always carries a real analytic structure so that the symplectic form and the group action are both
  real analytic. This observation implies that one can leave out the technical assumption of ''local analyticity''
  in the statement of \cite[Thm.~3.2]{BHP}. More precisely, the following holds.

\begin{Theorem} 
\label{Thm:KoszulResolutionContraintSurface}
  Let $(M,\omega,\Psi,J)$ be a $G$-Hamiltonian system with $G$ compact.
  Assume that the Koszul complex $K_\bullet (\calC^\infty (M), J)$ is
  a free resolution of $\calC^\infty (M_0)$. Then there exists an equivariant continuous
  linear section $\ext:\calC^\infty (M_0) \to \calC^\infty (M)$,
  called \emph{extension map}, of the restriction map $\rest: \calC^\infty (M) \to \calC^\infty (M_0)$, $f \mapsto f|_{M_0}$
  together with a family $h =(h_k)_{k\in \N}$ of continuous linear maps $h_k : K_k \to K_{k+1}$ such that
  \[
    \left(  (\calC^\infty(M_0),0 ) \overset{\ext}{\underset{\rest}{\rightleftarrows}} (K_\bullet,\partial) , h \right)
  \]
  is a deformation retract. 
  This means that $\ext : (\calC^\infty(M_0),0 )  \to (K_\bullet,\partial)$
  and $\rest : (K_\bullet,\partial) \to (\calC^\infty(M_0),0 )$ are chain maps
  fulfilling $\rest \circ \ext = \id $ and 
  $\id - \ext \circ \rest = \partial h + h \partial$.
  Moreover, one can achieve that the $h_k$ are equivariant and that the side conditions
  $h\circ h = 0$, $ h_0 \circ \ext = 0$ and $ \rest \circ h_{-1} =0$ hold true.   
\end{Theorem}

\begin{proof}
  Let us provide a  proof emphasizing where Theorem \ref{Thm:ExCompatibleEquivariantAnalyticStructures} comes in.
  According to that theorem there exists an analytic structure on $M$ so that that the $G$-action $\Psi$ and the moment map $J$
  are real analytic. By (\ref{ite:zerolevelextprop}) in the same theorem there exists an equivariant
  extension map $\ext: \calC^\infty (M_0) \to \calC^\infty (M)$.
  It remains to construct a chain homotopy $h=(h_k)_{k\in \N}$  with the desired properties. 
  To this end we follow the idea in the proof of \cite[Thm.~3.2]{BHP} and apply the
  division theorem by Bierstone and Schwarz \cite[Thm.~0.1.3.]{Schwarzbier} which says that for any matrix $\Phi \in \Mat_{k \times l} (\calC^\omega (M) )$ of real analytic functions on an analytic manifold the image of the 
  map $\Phi_\#: \calC^\infty (M)^l \to  \calC^\infty (M)^k$ induced by matrix multiplication with $\Phi$ is
  closed and has a continuous linear split $\sigma : \im \Phi \to \calC^\infty (M)^l$. The latter means
  that $\Phi_\#\sigma =\id_{\im \Phi}$. Note that the image of such a splitting is
  closed since $\sigma \Phi_\#$ acts as identity on $\im \sigma$ which by continuity implies
  \begin{equation}
    \label{eq:closednessimage}
    \overline{\im \sigma} =  \sigma \Phi_\# (\overline{\im \sigma}) \subset \im \sigma \ .
  \end{equation}
  Now consider the following sequence which is exact by assumption:
  \begin{equation}
    \label{eq:augmentedkoszul}
    0 \longleftarrow \calC^\infty (M_0) \overset{\rest}{\longleftarrow} K_0 \overset{\partial_1}{\longleftarrow}
    K_1  \overset{\partial_2}{\longleftarrow} \ldots \overset{\partial_d}{\longleftarrow} K_d
    \longleftarrow 0 \ .
  \end{equation}
  By the division theorem of Bierstone and Schwarz,  $\im \partial_k \subset K_{k-1}$ is closed
  for $k=1,\ldots, d$ and there exists  for each such  $k$ a continuous linear splitting
  $\sigma_{k-1} : \im \partial_k \to K_k$ of $\partial_k$.
  By equivariance of the $\partial_k$ and after possibly averaging over $G$ one can assume that each
  $\sigma_k$ is equivariant. For the particular case $k=-1$ we put $\sigma_{-1} = \ext$.
  Finally we assume $\sigma_l$ to be $0$ for those $l$  for which it has not been defined yet.
  By exactness of the sequence \ref{eq:augmentedkoszul} one obtains the direct sum decompositions
  $K_k = \im \sigma_{k-1} \oplus \im \partial_{k+1}$ for $k=0,\ldots d$. 
  We know already by the division theorem that 
  $\im \partial_{k+1}$ is closed. The subspace $\im \sigma_{k-1}$ is so, too,
  by the above argument involving Eq.~\eqref{eq:closednessimage}.
  Let $\pi_k : K_k \to \im \partial_{k+1}$ denote the canonical projection along 
  $\im \sigma_{k-1}$ for $k=0,\ldots d$ and put $\pi_{-1} := \id_{\calC^\infty (M_0)}$.
  Then $\pi_k$ is continuous and equivariant since  $\im \partial_{k+1}$  and $\im \sigma_{k-1}$
  are closed $G$-invariant subspaces of $K_k$. Furthermore,
  $\sigma_{k-1} \partial_k = \id_{K_k} - \pi_k$ for $k = 0,\ldots ,d$ since both sides act in the same way
  on $\im \sigma_{k-1}$ and $\im \partial_{k+1}$. Now let
  \[ h_k :=
    \begin{cases}
      \sigma_k \pi_k & \text{for } k= 0,\ldots , d-1,\\
      0 & \text{else} \ .
    \end{cases}
  \]
  Then compute
  \[
        \id_{\calC^\infty (M_0)} - \ext \circ \rest = \pi_0 = \partial_1 h_0 = \partial_1 h_0 + h_{-1} \partial_0 
  \]
  and for $k>0$
  \[
    \partial_{k+1}h_k + h_{k-1}\partial_k = \partial_{k+1}\sigma_k\pi_k + \sigma_{k-1}\pi_{k-1}\partial_k=
     \pi_k + \sigma_{k-1} \partial_k = \id_{K_k} \ . 
  \]
  Thus $h=(h_k)_{k\in \N}$ is the desired chain homotopy.
  Since $h_{k+1} h_k = \sigma_{k+1}\pi_{k+1}\sigma_k\pi_k =0$ for $k=0 ,\ldots ,d$
  and $h_0\, \ext = \sigma_0\pi_0\sigma_{-1}=0$ ,
  the first and second side conditions are fulfilled. Since $h_{-1}=0$ by construction, the third side
  condition holds trivially. 
\end{proof}

It later will turn out to be convenient to write the Koszul complex as a cohomological complex that is we put
$K^k = K_{-k} (\calC^\infty (M),J)$ for $k \in - \N$ and $K^k = 0$ for $k\in \N \setminus \{ 0 \}$. Note that
$K^\bullet$ is a bounded cochain complex. 

The second crucial ingredient in the construction of the BRST complex is the Chevalley--Eilenberg complex
$\left( \CE^\bullet (\g, \calC^\infty(M)) , \delta\right)$ of the $\g$-module $\calC^\infty (M)$,
see Example \ref{AppCEcomplex}.
Observe that the space $\calC^\infty (M)$ of smooth functions on $M$ carries a natural structure of
$\g$-module. An element $X\in \g$ hereby acts by the associated fundamental vector field $X_M$. More
precisely, the $\g$-module structure on  $\calC^\infty (M)$ is given by the map
\begin{equation}
  \label{Eq:ModuleStructureMap}
    L : \g \times \calC^\infty (M) \to \calC^\infty (M),\: (X,f) \mapsto L_X f = X_M f \ . 
\end{equation}

\begin{Remark}
  Since $\calC^\infty (M)$ is a commutative algebra, the Chevalley--Eilenberg complex becomes
  a differential graded algebra with the algebra structure given by the tensor product
  of the graded commutative algebra $\Lambda^\bullet\g^*$ and the commutative algebra $\calC^\infty (M)$.
  It is straightforward to check that the product of this algebra structure is graded commutative
  and that the Chevalley--Eilenberg coboundary then coincides with the unique graded linear  map
  $\delta : \CE^\bullet (\g, \calC^\infty(M)) \to \CE^\bullet (\g, \calC^\infty(M))$ of degree $+1$
  which satisfies the graded Leibniz identity, acts on elements $f \in \calC^\infty (M) $ of degree $0$ by
  \[
    \delta f (\xi) = \xi_M f \quad \text{for all } \xi \in \g
  \]
  and on elements of degree $1$ of the form  $\alpha \otimes 1$ with $\alpha \in \g^*$ by
  \[
    \delta (\alpha \otimes 1) (\xi,\zeta)= - \alpha ( [\xi, \zeta]) \quad \text{for all } \xi,\zeta \in \g \ .
  \]
\end{Remark}

\begin{Lemma}
\label{Lem:InducedModuleStructureConstraintSet}
  The $\g$-module structure $L$ on $\calC^\infty (M)$ leaves the vanishing ideal $I_{M_0}$ invariant and
  hence induces a  $\g$-module structure $L^0 : \g \times \calC^\infty (M_0) \to \calC^\infty (M_0)$
  on the quotient  $\calC^\infty (M_0) \cong  \calC^\infty (M) / I_{M_0}$. The action of an element $X \in \g$
  on $\calC^\infty (M_0)$ is then given by
  \begin{equation}
    \label{eq:ActionLieAlgElementFctConstraintSurf}
      L^0_X  = \rest \circ L_X \circ \ext \ , 
  \end{equation}
  where $\ext:  \calC^\infty (M_0) \to \calC^\infty (M)$ is a $G$-equivariant extension map.   
  In case $G$ is a connected compact Lie group one has with respect to this $\g$-module structure:
  \begin{equation}
    \label{Eq:ZeroLieAlgebraCohomologyConstraintSurface}
    H^0 (\g,\calC^\infty (M_0)) = \left(\calC^\infty (M_0)\right)^\g =
    \left(\calC^\infty (M_0)\right)^G \cong \calC^\infty (M\red G) \ .
  \end{equation}
\end{Lemma}
\begin{proof}
  Let $f\in I_{M_0}$, $X\in \g$ and $p\in M_0$. Then
  \[
    X_Mf (p) = \frac{d}{dt} \left. f(\exp(tX) \cdot p) \right|_{t=0} = 0 
  \]
  since $t \mapsto \exp(tX) \cdot p$ is a smooth path in $M_0$ by $G$-invariance.
  This means that $L$ leaves the ideal $I_{M_0}$ invariant.
  The induced $\g$-module structure on the quotient $\calC^\infty (M_0)$
  can be written in the form \eqref{eq:ActionLieAlgElementFctConstraintSurf} 
  since by Theorem \ref{Thm:ExCompatibleEquivariantAnalyticStructures} an extension map $\ext$ exists.  
   
To prove \eqref{Eq:ZeroLieAlgebraCohomologyConstraintSurface} observe that
$H^0 (\g,\calC^\infty (M)) = \calC^\infty (M)^\g$ for any $G$-manifold $M$ and
that $H^0 (\g,\calC^\infty (M_0)) = \calC^\infty (M_0)^\g$ for the constraint surface of the
Hamiltonian system. In the case of a connected compact Lie group  $G$ this implies that
$H^0 (\g,\calC^\infty (M)) = \calC^\infty (M)^G$ and that
$H^0 (\g,\calC^\infty (M_0)) = \calC^\infty (M_0)^G$.
\end{proof}

So, finally, we have all the tools to construct the classical BRST complex
$\calA^\bullet$ of a $G$-Hamiltonian system $(M,\omega,\Psi,J)$. As a graded algebra, $\calA^\bullet$
is defined as the graded tensor product of the Chevalley--Eilenberg complex
$\CE^\bullet (\g, \calC^\infty(M))$ with the Koszul complex
$(K^\bullet,\partial)$, that is
\begin{equation}
  \label{eq:classicalBRSTcomplex}
  \calA^\bullet =\CE^\bullet (\g, \calC^\infty(M)) \otimes_{\calC^\infty (M)} (K^\bullet,\partial) \ .
\end{equation}   
Expanding the right hand side one obtains for $n\in \Z$  
\begin{equation}
  \label{eq:BRSTcomponentsCE}
  \begin{split}
    \calA^n \,
    & = \calC^\infty (M) \otimes \bigoplus_{k,l \in \Z \atop k+l = n } \Lambda^k \g^* \otimes \Lambda^{-l} \g =\\
    & = S^n_{\calC^\infty (M)}  (\g^*[-1] \oplus \g[1]) =
    \bigoplus_{k,l \in \Z \atop k+l = n } \CE^k (\g,S^l_{\calC^\infty(M)} (\g[1]))  \ .
  \end{split}
\end{equation}
Elements of $\g^*$  thus have degree $+1$ and are called \emph{ghosts} in the physics literature,
whereas elements of $\g$ have degree $-1$ and are named \emph{antighosts}.
Note that $\bigoplus\limits_{k,l \in \Z \atop k+l = n } \Lambda^k \g^* \otimes \Lambda^{-l} \g$ can be
interpreted as the degree $n$ vector space underlying the free graded commutative algebra
\[ S^\bullet (\g^*[-1] \oplus \g[1]) \]
on the graded vector space $\g^*[-1] \oplus \g[1]$. Under this identification, the product map $\mu$ on
$S^\bullet (\g^*[-1]\oplus \g[1])$ is the unique graded commutative associative bilinear operation
fulfilling the equalities
\[
  \mu ( \alpha \otimes \beta) = \alpha \wedge \beta , \quad  \mu ( X \otimes Y) = X \wedge Y,
  \quad \text{and} \quad
  \mu ( \alpha \otimes X ) = - \mu ( X \otimes \alpha) = \alpha \otimes X  
\]
for all $\alpha,\beta \in \g^*$ and $X,Y \in \g$. Sometimes we will write $v \wedge w$
for the product $\mu(v,w)$ of two elements $v,w \in S^\bullet (\g^*[-1]\oplus \g[1])$. 
The BRST complex can now be written in the form
\begin{equation}
  \label{eq:BRSTrepresentation}
  \calA^\bullet = 
  \calC^\infty (M) \otimes S^\bullet (\g^*[-1] \oplus \g[1]) \ .
\end{equation}
The differentials $\partial : K^\bullet \to K^\bullet$ and
$\delta : \CE^\bullet (\g, \calC^\infty(M)) \to \CE^\bullet (\g, \calC^\infty(M))$  
extend in a natural way to graded derivations on $\calA^\bullet$ by letting them act trivially on
$\g^*[-1]$ and $\g[1])$, respectively. The thus extended differentials supercommute, so
\[
   \calD = \delta + 2 \partial 
\]
is a differential of degree $+1$ on $\calA^\bullet$. In addition, $\calA^\bullet$ inherits from
$S^\bullet (\g^*[-1]\oplus \g[1])$ a graded commutative associative product which we also denote
by $\mu$. Thus $\big( \calA^\bullet, \mu , \calD\big)$ becomes a differential graded commutative
$\calC^\infty (M)$-algebra which one calls the \emph{classical BRST algebra}.
The BRST algebra also carries a natural Poisson bracket. For its definition we need some more notation.
To this end let
\[
  i : \g^*\oplus \g \to \End\left( S^\bullet (\g^*[-1]\oplus \g[1]) \right)
\]
be \emph{left insertion} which means let $i$ be the unique linear map
from $\g^*\oplus \g$ to the graded endomorphism ring of
$S^\bullet (\g^*[-1]\oplus \g[1])$ such that
\[
  i(X) (\omega \otimes \Zeta) = (X \lrcorner \omega) \otimes \Zeta \quad\text{and}\quad
  i(\alpha) (\omega \otimes \Zeta) = (-1)^k \omega \otimes (\alpha \lrcorner \Zeta) 
\]
for all $\alpha \in \g^*$, $X \in \g$, $\omega \in \Lambda^k \g^*$ and $\Zeta \in  \Lambda^l \g $. 
By \emph{right insertion} we understand the unique linear map
\[
  j : \g^*\oplus \g \to \End\left( S^\bullet (\g^*[-1]\oplus \g[1]) \right)
\]
such that $j(v) x = (-1)^{n+1} i(v)x$ for all $v\in \g^*\oplus \g$ and
$x \in S^n (\g^*[-1]\oplus \g[1])$. Then we define the \emph{Poisson endomorphisms}
$ P$ and $ P^*$ on $ S^\bullet (\g^*[-1]\oplus \g[1]) \otimes S^\bullet (\g^*[-1]\oplus \g[1])$
by
\begin{equation}
  \label{eq:PoissonEndos}
  P = \sum_{l=1}^d j(\varepsilon^l) \otimes i(E_l) \quad\text{and}\quad
  P^* = \sum_{l=1}^d j(E_l) \otimes i(\varepsilon^l) \ ,
\end{equation}
where as before $(E_l,\ldots,E_d)$ is a basis of $\g$ and $(\varepsilon^l,\ldots,\varepsilon^d)$ its dual basis.
Note that $P$ and $P^*$ do not depend on the particular choice of these bases.
Now we can subsume and define the Poisson bracket on the BRST algebra.
See \cite[3.10]{HPhD}, \cite[Sec.~4]{BHW} and \cite[Sec.~4]{BHP} for further details
and a proof. 

\begin{Proposition}\label{prop:PoissonbracketBRSTalgebra} 
  As a graded algebra, the classical BRST algebra  $\calA^\bullet$ of a $G$-Hamiltonian system $(M,\omega,\Psi,J)$ coincides
  with the free graded commutative $\calC^\infty (M)$-algebra generated by $\g^*[-1] \oplus \g[1]$. 
  Moreover, $\calA^\bullet$ carries an even graded Poisson bracket $\{ - ,- \}_\calA$ given by
  \[
    \{ f \, v ,g\, w\}_\calA =
    \{ f,g\}_M \, \mu (v, w) + 2 fg \, \mu \! \left( (P + P^*) (v\otimes w) \right) 
  \]
  for all $f,g \in \calC^\infty (M)$ and $v,w\in  S^\bullet(\g^*[-1] \oplus \g[1])$.
  Finally, the element
  \begin{equation}
    \label{eq:classicalBRSTcharge}   
    \theta = - \frac 14 [ - , - ] + J \in \calA^1 
  \end{equation}
  satisfies $\{ \theta, \theta\}_\calA =0$ and $\calD = \{ \theta , - \}_\calA$ which again entails that
  $\calD^2 =0$ and that $(\calA^\bullet,\mu,\calD)$ is a differential graded algebra. One calls $\theta$ the
  \emph{classical BRST charge} and $\calD$ the \emph{classical BRST differential}. 
\end{Proposition}

The crucial observation from \cite[Thm.~4.1]{BHP} now is that under the assumption of the generating hypothesis
\hyperref[GenCond]{(GH)} and the acyclicity hypothesis \hyperref[AcycCond]{(AH)} the BRST cochain complex
$(\calA^\bullet,\calD)$ and the Chevalley--Eilenberg complex  $\left( \CE^\bullet (\g, \calC^\infty(M_0)), \delta^0\right)$
with values in the $\g$-module of smooth functions on the constraint surface are quasi-isomorphic in the additive
category of Fr\'echet spaces. Note that by $\delta^0$ we denote here the Chevalley--Eilenberg coboundary
with respect to the $\g$-representation $L^0$ on $\calC^\infty(M_0)$.

\begin{Theorem}\label{thm:homologicalreductionPoissonbracket}
  Let $(M,\omega,\Psi,J)$ be a $G$-Hamiltonian system for which  the Koszul complex
  $K_\bullet (\calC^\infty(M),J)$ is a free resolution of $\calC^\infty(M_0)$.
  Choose an equivariant continuous extension map $\ext:\calC^\infty (M_0) \to \calC^\infty (M)$
  and an equivariant continuous homotopy $h=(h_k)_{k\in \N}$ according to
  Theorem \ref{Thm:KoszulResolutionContraintSurface}.   Then 
  \begin{equation}
    \label{eq:retractclassicalBRSTcomplex}
     \left(  \left(\CE^\bullet (\g, \calC^\infty(M_0)),\delta^0\right)
    \overset{\ext}{\underset{\rest}{\rightleftarrows}} (\calA^\bullet,\calD) , \frac 12 h \right)
  \end{equation}
  is a deformation retract. If $G$ is connected, the Poisson bracket of two elements $f,g\in \calC^\infty (M\red G)$
  can be recovered by the identity
  \begin{equation}
    \label{eq:reducedPoissonbracketBRSTbracket}
     \{f,g \}_{M\red G} = \rest \{\ext(f),\ext(g) \}_{\cal A} 
   \end{equation}
   under the natural identifications
   $\calC^\infty (M\red G) \cong \calC^\infty (M_0)^G = \calC^\infty (M_0)^\g = H^0 (\g ,\calC^\infty (M_0))$.
\end{Theorem}

\begin{proof}
  Treating $\calD$ as a perturbation of $2\partial$, we can apply the Perturbation Lemma  \ref{thm:perturbationlemma} to the deformation
  retract provided by Theorem \ref{Thm:KoszulResolutionContraintSurface}.
  This yields that \eqref{eq:retractclassicalBRSTcomplex} is a deformation retract.
  See the proof of \cite[Thm.~4.1]{BHP} for details.
  The equality \eqref{eq:reducedPoissonbracketBRSTbracket} follows immediately from \eqref{eq:Poissonbracketreducedspace},
  equivariance of the extension map $\ext$ and the definition of the Poisson bracket $\{ - ,- \}_{\cal A}$
  in Proposition \ref{prop:PoissonbracketBRSTalgebra}.  
\end{proof}

\begin{Remark}
  The preceding result says in other words that the symplectically reduced space $(M\red G,\calC^\infty_{M\red G},\{-,-\})$ is
  representable as the zeroth cohomology of the BRST complex with its natural structure of a differential graded
  Poisson algebra.
\end{Remark}

\subsection{The quantized version}
\label{sec:quantizedversion}

Under the assumption that the star product $\star$ on a $G$-Hamiltonian system
$(M,\omega,\Psi,J)$ satisfies certain invariance conditions described below
and that the conditions
\hyperref[GenCond]{{\rm(GH)}} and \hyperref[AcycCond]{{\rm (AC)}} hold true,
the classical BRST algebra allows for a formal deformation quantization which
then induces a star product on $H^0 (\g,\calC^\infty (M_0))$.
Let us describe this ansatz in more detail. The main assumption is 
that $\star$ is a $G$-\emph{invariant star product} which means that 
\begin{equation} 
\label{eq:invariancestarproduct}
  L_g^* (f_1 \star f_2) = L_g^* (f_1) \star L_g^* (f_2)
  \quad \text{for all } g \in G, \: f_1,f_2 \in \calC^\infty (M) \ , 
\end{equation}
where $L_g$ denotes the left action of a group element $g \in G$. 
$G$-invariance of $\star$  implies that the star product is also $\g$-\emph{invariant} meaning that 
\begin{equation}
\label{eq:liealgebrainvariancestarproduct}
\{ J_\xi , f_1 \star f_2 \} = \{ J_\xi , f_1 \}  \star f_2
 + f_2 \star \{ J_\xi , f_2 \}
  \quad \text{for all } \xi \in \g, \: f_1,f_2 \in \calC^\infty (M) \ .
\end{equation}
In case the Lie group $G$ is connected and simply-connected, a $\g$-invariant star product is also
$G$-invariant. 

To construct a quantized version of the BRST complex \eqref{eq:classicalBRSTcomplex} we need
a $\g$-module structure on the deformed algebra $\left( \calC^\infty (M)[[\lambda]],\star \right)$.
To this end we assume the star product to be \emph{covariant} which means that 
\begin{equation}\label{eq:covariance}
  \mm_X \star \mm_Y - \mm_Y \star \mm_X = \lambda \mm_{[X,Y]}
   \quad \text{for all } X,Y \in \g \ .
\end{equation}

\begin{Lemma}\label{lem:quantizedrep}
  Let $\star$ be a covariant star product on the $G$-Hamiltonian system $(M,\omega,\Psi,J)$.
  Then the operation
  \[
   \bm{L} : \g \times \calC^\infty (M)[[\lambda]] \to \calC^\infty (M)[[\lambda]], \: (X,f) \mapsto
   \bm{L}_X f = \frac{1}{\lambda} \ad_\star(J_X) f := \frac{1}{\lambda} \left(J_X \ast f - f \ast J_X \right)
  \]
  is a $\g$-representation called the \emph{quantized representation} of $\g$ on $\calC^\infty (M)[[\lambda]]$.
\end{Lemma}

\begin{proof}
  One immediately computes that for $X,Y\in \g$ and $f\in \calC^ \infty (M)$
  \[
    \left[ \ad_\star (J_X) , \ad_\star (J_Y) \right] f =
    \left( J_X\star J_Y\star f - J_Y\star J_X\star f  - f \star J_X\star J_Y + f \star J_Y\star J_X\right)
  \]
  and that
   \[
    \ad_\star (J_{[X,Y]})  f =  \left( J_{[X,Y]} \star  f - f \star  J_{[X,Y]}\right) \ .
   \]
   By covariance of the star product the equality
   \[
     \frac{1}{\lambda^2}\left[ \ad_\star (J_X) , \ad_\star (J_Y) \right] f
     = \frac{1}{\lambda} \ad_\star (J_{[X,Y]})  f 
   \]
   follows. This proves the claim. 
\end{proof}

\begin{Remark}
  \begin{enumerate}[(a)]
  \item   Covariance of the star product on a $G$-Hamiltonian system implies in particular
  that the classical moment map $J$ is a quantum moment map, see \cite{Xu,G3}.
\item
  According to \cite{BHW,G3} a covariant star product exists for every $G$-Hamiltonian system 
  with a compact Lie group action; see also \cite[Sec.~5.8]{FedosovBook}. 
\item
  The natural action of $\g$ on $\calC^\infty (M)$ extends by 
  \[
   L : \g \times \calC^\infty (M)[[\lambda]] \to \calC^\infty (M)[[\lambda]], \: (X,f) \mapsto
   L_X f = \sum_{n\in \N} X_M f_n  \lambda^n    \text{ with } f = \sum_{n\in \N} f_n \lambda^n
  \]
  to another $\g$-representation on $\calC^\infty (M)[[\lambda]]$ which we call the \emph{classical} one.
  By construction, the quantized representation $\bm{L}$ is a deformation of the classical representation $L$
  which   means that
  \[
    \bm{L}_X f - X_M f \in \lambda \calC^\infty (M)[[\lambda]]
    \quad \text{for all } X  \in \g, \: f \in \calC^\infty (M) \ .
  \]
  In general,  $\bm{L}$  and $L$ do not coincide, though. If they do, which in other words means that
  \begin{equation}\label{eq:stronginvariance}
    \mm_X \star f - f \star \mm_X = \lambda \{J_X,f \}
    \quad \text{for all } X  \in \g, \: f \in \calC^\infty (M) \ ,
  \end{equation}
  then one calls the star product \emph{strongly invariant}.
\end{enumerate}
\end{Remark}

\begin{Remark}
  The ring of formal power series $\R[[\lambda]]$ and modules over it of the form $V[[\lambda]]$, where $V$ is a
  real vector space, carry a natural translation invariant topology called the $\lambda$-\emph{adic topology}. A fundamental
  system of $0$-neighborhoods is given by the family of subspaces  $\left( \lambda^k V[[\lambda]]\right)_{k\in \N}$.
  As remarked in \cite[Sec.~2.1]{BorDQS},  $\R[[\lambda]]$ and modules of the form $V[[\lambda]]$
  thus become completely metrizable. We will silently make use of this fact several times in the following. 
\end{Remark}

In the next step we define a product $\cdot$ on the space $S^\bullet (\g^*[-1]\oplus\g[1])[[\lambda]]$
of power series in $\lambda$ with coefficients in the free graded algebra over $\g^*[-1]\oplus\g[1]$
and then extend it to a formal deformation of the classical BRST algebra.
The product $\cdot$ is given by
\[
  v \cdot w = \mu \left( e^{-2\lambda P} (v \otimes w) \right) \quad \text{for }
  v,w\in S^\bullet (\g^*[-1]\oplus\g[1])  \ ,
\]
where $P$ denotes the endomorphism from Eq.~\eqref{eq:PoissonEndos} and the graded commutative product
$\mu$ on  $S^\bullet (\g^*[-1]\oplus\g[1])$ has been extended in a unique way to a $\lambda$-adically continuous 
and $\R[[\lambda]]$-bilinear associative product on $S^\bullet (\g^*[-1]\oplus\g[1])[[\lambda]]$ which we again
denote by $\mu$. Combination of the product $\cdot$ with the covariant star product on $M$ 
gives rise to the formal deformation of the classical BRST algebra we are looking for. More precisely, the
formally deformed product on $\calA^\bullet[[\lambda]]$ is defined by
\begin{equation}
  \label{eq:quantumBRSTstarproduct}
   (f \, v) * (g \, w) = (f\star g) \, (v \cdot w) \quad
  \text{for } f,g\in \calC^\infty (M), \: v,w \in S (\g^*[-1]+\g[1])  
\end{equation}
and then extended in a canonical way to a continuous and $\R[[\lambda]]$-bilinear associative product.

Now we equip $\calA^\bullet[[\lambda]]$ with a differential called the \emph{quantum BRST differential}
which will turn out to be  a deformation of the classical BRST differential. To this end we first extend the $\g$-module
structure on $\calC^\infty (M)[[\lambda]]$ to one on $S^\bullet_{\calC^\infty (M)}(\g[1])[[\lambda]]$  by the map
\begin{equation*}
  \begin{split}
    \bm{L}: \: & \g \times S^\bullet_{\calC^\infty (M)}(\g[1])[[\lambda]] \to  S^\bullet_{\calC^\infty (M)}(\g[1])[[\lambda]], \\
    & (X , f \, v)  \mapsto
      \bm{L}_X (f \, v) = \frac{1}{\lambda} (\ad_*(J_X) f )\, v +  f \, \ad (X) v \ ,
  \end{split}
\end{equation*}
where $X \in \g$, $f\in \calC^\infty(M)[[\lambda]]$ and $v \in S^\bullet (\g[1])$. By definition, it is clear
that $\bm{L}$ is a deformation of the representation
$L :  \g \times S^\bullet_{\calC^\infty (M)}(\g[1]) \to S^\bullet_{\calC^\infty (M)}(\g[1])$.
By the identification from Equation \eqref{eq:BRSTcomponentsCE} and the fact that $\bm{L}$ leaves the symmetric
degree invariant, the $\g$-module structure $\bm{L}$ on $S^\bullet_{\calC^\infty (M)}(\g[1])$ gives rise to a
Chevalley--Eilenberg coboundary
\[
  \bm{\delta} : \calA^\bullet [[\lambda]] \to \calA^\bullet [[\lambda]] \ .
\]
Secondly, we need a deformation of the Koszul complex. To this end put
\[
  \bm{K}^{-k} = \bm{K}_k = \bm{K}_k (\calC^\infty[[\lambda]],\star,J) =
  \begin{cases}
    \calC^\infty (M,\Lambda^{-k}\g )[[\lambda]] & \text{for } k\in \N, \\
    0 & \text{for } k\in \Z_{< 0}. 
  \end{cases}
\]  
Following \cite{BHP,GW}, the \emph{quantized Koszul differential} $\bm{\partial}$ is now given in degree $k$ by
\begin{equation} \label{eq:deformedKoszuldifferential}
  \begin{split}
  \bm{\partial} : \: & \bm{K}_\bullet (\calC^\infty(M)[[\lambda]],\star,J) \to \bm{K}_\bullet (\calC^\infty(M)[[\lambda]],\star,J) \ ,\\
  & f\, v \mapsto \sum_l  (f \star J_l)  \, i(\varepsilon^l)v + \frac{\lambda}{2}
  f \left( \sum_{j,k,l=1}^d  C^l_{jk} E_l \wedge i(\varepsilon^j) i(\varepsilon^k) v + i(\Delta) v \right) \ ,
  \end{split}
\end{equation}
where $f\in \calC^\infty (M)[[\lambda]]$, $v \in \Lambda^\bullet\g$, $ C^l_{jk} = \varepsilon^l ([E_j,E_k])$ are the structure constants
of the Lie algebra $\g$ with respect to the basis $(E_k)_{1\leq k\leq d}$ of $\g$, and $\Delta \in \g^*$ is the modular $1$-form defined by
\[
    \Delta (X) = \tr \ad_X \quad \text{for all } X \in \g \ .
\]
The degree $k$ component of the quantized Koszul differential will be denoted $\bm{\partial}_k$. By definition,
$\bm{\partial}_k$ maps $\calC^\infty (M,\Lambda^k\g  )[[\lambda]]$ to $\calC^\infty (M,\Lambda^{k-1}\g  )[[\lambda]]$. 

\begin{Proposition}\label{prop:quantizedKoszulcomplex}
  Let $(M,\omega,\Psi,J)$ be a $G$-Hamiltonian system satisfying conditions  \hyperref[GenCond]{{\rm(GH)}} and
  \hyperref[AcycCond]{{\rm (AC)}}. Choose an equivariant continuous linear extension map $\ext :\calC^\infty(M_0) \to\calC^\infty (M)$
  and a continuous equivariant homotopy $h =(h_k)_{k\in \N}$ as in Theorem \ref{Thm:KoszulResolutionContraintSurface} such that
  the side conditions $h_0 \circ \ext =0$ and $h_{k+1} \circ h_k = 0$ for $k\in \N$ are fulfilled. 
  Assume further that $\star$ is an invariant and covariant star product on $\calC^\infty (M)[[\lambda]]$.
  Then $\left( \bm{K}^\bullet (\calC^\infty[[\lambda]],\star,J) , \bm{\partial}\right)  $
  is an acyclic cochain complex called the \emph{quantized Koszul complex}. Its $0$-degree homology is given by 
  $H^0 (\bm{K}^\bullet) \cong \calC^\infty (M_0)[[\lambda]]$. Moreover,
  \[
    \left(  (\calC^\infty(M_0)[[\lambda]],0 ) \overset{\ext}{\underset{\bm{\rest}}{\rightleftarrows}} (\bm{K}_\bullet,\bm{\partial}) , \bm{h} \right)
  \]
  is a special deformation retract, where the map $\bm{\rest}$ and the homotopy $\bm{h}$ are defined, recursively, as follows:
  \begin{equation*}
    \begin{split}
      \bm{\rest} & = \rest \left( \id + (\bm{\partial}_1 -\partial_1)h_0 \right)^{-1} \ , \\
      \bm{h}_0 & = h_0 \left( \id + (\bm{\partial}_1 -\partial_1)h_0 \right)^{-1} \ , \\
      \bm{h}_k & = h_k \left(h_{k-1}\bm{\partial}_k + \bm{\partial}_{k+1} h_k \right)^{-1} \ .
    \end{split}
  \end{equation*}
  Finally, $\bm{\partial}$, $\bm{\rest}$ and $\bm{h}$ are deformations of $\partial$, $\rest$ and the homotopy $h$,
  respectively, which means that $\bm{\partial} = \partial + O(\lambda)$, $\bm{\rest} = \rest + O(\lambda)$
  and $\bm{h}^k = h^k + O(\lambda)$ for all $k$. 
\end{Proposition}

\begin{proof}
  By \cite[Thm.~4.1]{HPhD} or \cite[Lem.~3.4]{GW}, we have $\bm{\partial}^2 =0$, so $\bm{\partial}$ is a differential.
  By construction, $\bm{\partial}$ is a deformation of the classical Koszul differential, which in particular implies
  that the Neumann series $\sum_{k\in \N} \left((\partial_1 - \bm{\partial}_1)h_0\right)^k$  converges in the $\lambda$-adic topology.
  Its limit is  $\left( \id + (\bm{\partial}_1 -\partial_1)h_0 \right)^{-1}$, so $\bm{\rest}$ is well-defined
  and a deformation of $\rest$ as claimed. In the same way one shows that $\bm{h}_0 $ is well-defined and a  deformation of $h_0$.
  By induction one verifies the corresponding claim for $\bm{h}_k$. Since $\bm{\partial}$ is a perturbation of $\partial$
  in the sense of homological perturbation theory, application of the perturbation lemma \cite[Lemma A.1]{BHP} (see also \ref{thm:perturbationlemma} and \cite[2.4 \& 3.2]{CrainicPerturbation})
  now entails that $(\ext,\bm{\rest})$ is a special deformation retract with retracting homotopy $\bm{h}$.   
\end{proof}

 Observe that the quantized Koszul differential extends to a graded derivation on $\calA^\bullet [[\lambda]] $ by letting it act trivially on 
 $\g^*[-1]$. We will denote this extension again by  $\bm{\partial}$.
 Now we can formulate the following crucial result
 originally proved in \cite{HPhD} and \cite{BHW}. 
 
\begin{Propdef}
  Let $(M,\omega,\Psi,J)$ be a $G$-Hamiltonian system and $\star$ a $G$-invariant and covariant star product
  on $M$. Let $\ast$ be the associative product on
  $\calA^\bullet[[\lambda]]$ defined by
  \eqref{eq:quantumBRSTstarproduct},  $\bm{\partial}$ the deformed Koszul differential,
  and $\bm{\delta}$ the Chevalley--Eilenberg differential induced by the $\g$-representation $\bm{L}$. 
  Then the differentials $\bm{\delta}$ and $\bm{\partial}$ supercommute, so the
  \emph{quantum BRST differential} defined by 
  \begin{equation}
    \label{eq:qBRSTdifferential}
    \bm{\calD} = \bm{\delta} + 2\bm{\partial}
  \end{equation}
  is a differential on $\calA^\bullet[[\lambda]]$.   Moreover, the \emph{quantized BRST charge}
  \begin{equation}
    \label{eq:quantizedBRSTcharge}
    \bm{\theta} = - \frac 14 [ - , - ] + J + \frac{\lambda}{2} \Delta \in \calA^1[[\lambda]]
  \end{equation}
  satisfies $ \bm{\theta} \ast \bm{\theta} = 0 $ and $\bm{\calD}  = \frac{1}{\lambda} \ad_\ast (\bm{\theta})$,
  hence $\bm{\calD}$ is a graded derivation. 
  The triple $(\calA^\bullet[[\lambda]], \ast , \bm{\calD})$ thus becomes a differential graded algebra
  called the \emph{quantum BRST algebra}. It is a deformation of the classical BRST algebra  $\calA^\bullet$.
\end{Propdef}

\begin{proof}
  That $\bm{\delta}$ and $\bm{\partial}$ supercommute follows by straightforward but lengthy computation, see
  e.g.~\cite[Thm.~4.1.2]{HPhD}. 
\end{proof}
 
Before coming to quantum reduction of the star product we need one more tool, namely a deformed
representation of $\g$ on $\calC^\infty (M_0)[[\lambda]]$.

\begin{Lemma}
  Under the assumptions of Proposition \ref{prop:quantizedKoszulcomplex} on the $G$-Hamiltonian system $(M,\omega,\Psi,J)$
  and with the quantized representation from Lemma \ref{lem:quantizedrep},
  the operation 
  \[
   \bm{L}^0 : \g \times \calC^\infty (M_0)[[\lambda]] \to \calC^\infty (M_0)[[\lambda]], \: (X,f) \mapsto
   \bm{L}^0_X f = \bm{\rest} \bm{L}_X \ext (f)
  \]
  is a representation of $\g$ on $\calC^\infty (M_0)[[\lambda]]$ which we call a \emph{quantized representation} as well.
  It is a deformation of the representation $L^0$ of $\g$ on $\calC^\infty (M_0)$ defined in
  Lemma \ref{Lem:InducedModuleStructureConstraintSet}.
\end{Lemma}

The Chevalley--Eilenberg differential induced by $\bm{L}^0$ will be denoted $\bm{\delta}^0$.
We can now formulate the method of quantum reduction of the star product on the quantized BRST algebra.

\begin{Theorem}\label{thm:homologicalreductionstarproduct}
  Let $(M,\omega,\Psi,J)$ be a $G$-Hamiltonian system for which  the Koszul complex
  $K_\bullet (\calC^\infty(M),J)$ is a free resolution of $\calC^\infty(M_0)$.
  Let $\star$ be an invariant and covariant star product on $\calC^\infty (M)[[\lambda]]$.
  Let $\ext:\calC^\infty (M_0) \to \calC^\infty (M)$ be an extension map and  $h=(h_k)_{k\in \N}$ an equivariant continuous homotopy
  as in Theorem \ref{Thm:KoszulResolutionContraintSurface} so that the side conditions are fulfilled. Further, let $\bm{\rest}$ and $\bm{h}$
  be the deformed restriction map and deformed homotopy, respectively, from Proposition \ref{prop:quantizedKoszulcomplex}. Then 
  \begin{equation}
    \label{eq:retractquantumBRSTcomplex}
     \left(  \left(\CE^\bullet (\g, \calC^\infty(M_0)[[\lambda]]),\bm{\delta}^0\right)
     \overset{\ext}{\underset{\bm{\rest}}{\rightleftarrows}} (\calA^\bullet[[\lambda]],\bm{\calD}) , \frac 12 \bm{h} \right)
  \end{equation}
  is a deformation retract. Hence the star product $\ast$ on the quantized BRST algebra induces an associative product
  $\tilde{\star}$ on $H^0 (\g, \calC^\infty(M_0)[[\lambda]])$ by 
  \begin{equation}
    \label{eq:reducedBRSTstarproduct}
    f \,\tilde{\star}\, g = \bm{\rest} \left( \ext(f) \ast \ext(g) \right) \quad \text{for }
    f,g \in Z^0 (\g, \calC^\infty(M_0)[[\lambda]]) \ .
   \end{equation}
 \end{Theorem}

 \begin{proof}
   Since $\bm{\calD}$ is a perturbation of the differential $2\bm{\partial}$ fulfilling the assumptions
   of \ref{thm:perturbationlemma} one can apply that version of the perturbation lemma.
   By equivariance of $h$ one obtains the particular form of the homotopy in the perturbed deformation
   retract. See \cite[Thm.~6.1]{BHP} for further details.
   It remains to prove  associativity of the operation 
   \eqref{eq:reducedBRSTstarproduct}. This has been achieved in \cite[Thm.~4.3.3]{HPhD}.
 \end{proof}
 
 \begin{Corollary}
\label{cor:stronglyinvariant}
   If in addition to the assumptions of the preceding theorem the $\star$ product on the $G$-Hamiltonian system  $(M,\omega,\Psi,J)$
   is strongly invariant, then the representations $L^0$ and $\bm{L}^0$ of $\g$ on $\calC^\infty (M_0)[[\lambda]]$ coincide
   and the invariance spaces $H^0 (\g, \calC^\infty(M_0)[[\lambda]])$ and $H^0 (\g, \calC^\infty(M_0))[[\lambda]]$ are naturally isomorphic.
   If furthermore $G$ is connected, then, under the natural identifications
   \[ H^0 (\g, \calC^\infty(M_0)[[\lambda]]) \cong H^0 (\g, \calC^\infty(M_0))[[\lambda]]\cong \calC^\infty(M_0)^\g [[\lambda]] =
     \calC^\infty(M_0)^G[[\lambda]] \ , \]
   the formula
    \begin{equation}
    \label{eq:reducedstronglyinvariantBRSTstarproduct}
    f \, \tilde{\star}\,  g = \bm{\rest} \left( \ext(f) \ast \ext(g) \right) \quad \text{for }
    f,g \in \calC^\infty(M_0)^G[[\lambda]] 
   \end{equation}
   defines a star product on the symplectically reduced phase space $\left(\calC^\infty (M\red G),\{- , -\}_{M\red G} \right)$.
 \end{Corollary}

 \begin{proof}
   By definition in Eq.~\eqref{eq:stronginvariance}, strong invariance of the star product implies that
   the representations  $L$ and $\bm{L}$  coincide. Since by definition for $x \in \g$
   \[
     \bm{L}^0_X =
     \bm{\rest} L^0_X  \ext = \rest (\id + (\bm{\partial}_1 -\partial_1)h_0)^{-1} L^0_X \ext
   \]  
   and since $L_X$ commutes with $ \bm{\partial}_1 $, $\partial_1$, and $h_0$, one concludes  that
   $\bm{L}^0_X = \rest L_X \ext =L_X^0$. 
   But this implies that
   $H^0 (\g, \calC^\infty(M_0)[[\lambda]]) \cong H^0 (\g, \calC^\infty(M_0))[[\lambda]]$.
   The rest of the claim is a straightforward consequence of this. 
 \end{proof}
 
 The final result is due to Herbig \cite{HPhD}. See loc.~cit.~for a proof. 

\begin{Theorem}[cf.~{\cite[Prop.~4.3.6]{HPhD}}]
\label{thm:homologicalreduction}
  Let $G$ be a compact connected semisimple Lie group and $(M,\omega,\Psi,J)$ a $G$-Hamiltonian system satisfying
  the generating condition \hyperref[GenCond]{{\rm(GH)}} and the acyclicity condition \hyperref[AcycCond]{{\rm (AC)}}.
  Assume further that $\star$ is an invariant and covariant star product on $\calC^\infty (M)[[\lambda]]$.
  Then there exists a sequence of continuous linear maps \[ S_k :\calC^\infty (M_0)^G \to \calC^\infty (M_0) \]
  such that 
  \[
    S = \sum_{k\in \N} \lambda^k S_k : H^0 (\g, \calC^\infty (M_0) [[\lambda]] = \calC^\infty (M_0)^G[[\lambda]] \to
    \calC^\infty (M_0)[[\lambda]]
  \]
  has image $H^0 (\g, \calC^\infty (M_0)[[\lambda]])$ and is a topologically linear isomorphism onto its image.
  Moreover, if $\ast$ is the star product on the quantized BRST algebra and $\bm{\rest}$ the deformed restriction map
  from Proposition \ref{prop:quantizedKoszulcomplex}, then 
  \[
    f \,\hat{\star}\, g := S^{-1} \left( S(f )\,\tilde{\star}\, S(g) \right) = S^{-1} \bm{\rest} \left( \ext (S(f)) \ast \ext(S(g)) \right)
    \:\: \text{for }  f,g \in \calC^\infty (M\red G) \cong \calC^\infty (M_0)^G
  \]
  defines a star product on the symplectically reduced phase space $\left(\calC^\infty (M\red G),\{- , -\}_{M\red G} \right)$.
\end{Theorem}


\section{Application to the model}
\label{sec:application}
We now want to apply the homological quantization method to the quantum lattice gauge model obtained by
deformation quantization, see Section \ref{sec:fedosov}.
By the discussion in the previous section, we have to check the generating
hypothesis \hyperref[GenCond]{(GH)} and the acyclicity condition
\hyperref[AcycCond]{(AC)} for the $G$-Hamiltonian system
$({T^*G}^N,\omega,\ul\Psi,J)$, together with the invariance and covariance of the star product
derived in Subsection \ref{sec:SympStrucCotangBdl}. As a consequence, we will be able to conclude that Theorems \ref{thm:homologicalreductionstarproduct} and \ref{thm:homologicalreduction} hold true for our model. This will be accomplished in the following subsections for the case $G=\SU (2)$. 

It is not evident whether these properties also hold for other star products, notably for the star product of Weyl type. To answer this question for the latter, one may use the analysis of the relation with the standard ordered star product as provided in Section 8 of \cite{BNW1}. This will be discussed elsewhere.

In the sequel, we will denote by $M$ the cotangent bundle $T^*G^N$ with the
natural analytic structure inherited from the Lie group $G$
and by $\omega$ the canonical symplectic form on $M = T^*G^N$.

\subsection{The generating hypothesis}


We wish to apply \cite[Thm.~6.3]{AGJ}, which relates the generating hypothesis to algebraic conditions for a covering by local models. 

Let $(\ul a, \ul A) \in M$ be given. The tangent space $\tg_{(\ul a , \ul A)} M$ is acted upon by the isotropy representations of the stabilizer subgroup $\gp_{(\ul a , \ul A)}$ of $(\ul a , \ul A)$ and the corresponding Lie subalgebra $\la_{(\ul a , \ul A)}$ (where the latter representation is the Lie algebra representation induced by the former one). Choose a $\gp_{(\ul a , \ul A)}$-invariant vector space complement $V$ of the tangent space of the orbit at $(\ul a , \ul A)$ in $\ker\big(\mm'_{(\ul a, \ul A)}\big) \cong \tg_{(\ul a,\ul A)} \mm^{-1}(0)$. For example, we may choose the orthogonal complement  with respect to the Riemannian metric induced by the scalar product on $\la$. By the theory of symplectic reduction \cite[Ch.\ 10]{BuchI}, $V$ is a symplectic subspace, called a symplectic slice, and  the induced action of $\gp_{(\ul a , \ul A)}$ on $V$ is Hamiltonian with momentum mapping
\beq\label{G-JV}
\mm^V : V \to \la_{(\ul a , \ul A)}^\ast
 \,,\quad
\langle \mm^V(v) , B \rangle = \frac 1 2 \omega_{(\ul a , \ul A)}\left( v , B \cdot v \right) ~\text{ for all }~ v \in V, ~ B \in \la_{(\ul a , \ul A)}\,,
\eeq
where $B \cdot v$ means the action via the isotropy representation. By the Symplectic Tubular Neighbourhood Theorem, the Hamiltonian Lie group action so defined is a local model for the original Hamiltonian Lie group action in a neighbourhood of $(\ul a , \ul A)$ in the sense of Theorem 4.1 in \cite{AGJ}. Therefore, Theorem 6.3 in this article yields that the generating hypothesis holds if and only if for every element of a covering of $M$ by symplectic tubular neighbourhoods, the ideal $I(\mm^V)$ generated in the polynomial ring $\RR[V]$ by the functions $\mm^V_B$ with $B \in \mf g_{(\ul a , \ul A)}$ is a real radical ideal. The latter means that $I(\mm^V)$ coincides with its real radical, that is,
$$
I(\mm^V)
 =
\left\{f \in \RR[V] : f^{2k} + \sum_{j=1}^r g_j^2 \in I(\mm^V) \text{ for some $k$ and some $g_j \in \RR[V]$} \right\}.
$$
In the special situation of a cotangent bundle, it suffices to consider symplectic tubular neighbourhoods about orbits of points in the zero section. 
Thus, let us determine $V$ and $\mm^V$ for elements $(\ul a , \ul 0)$ of the zero section. First, consider the general situation of the Hamiltonian Lie group action $(\ctg Q , \gp, \mm)$ associated with a Lie group action $(Q,\gp)$. Let 
$$
s_0 : Q \to \ctg Q \,,\quad q \mapsto s_0(q) := 0_q\,,
$$
denote the zero section. For every $q \in Q$, we have the natural splitting
\beq\label{G-spl-ctg}
\tg_{0_q} (\ctg Q) = \tg_q Q \oplus \ctg_q Q\,,
\eeq
given by the tangent mapping $(s_0)'_q : \tg_q Q \to \tg_{0_q}(\ctg Q)$ and the inclusion $\ctg_q Q \subset \tg_{0_q}(\ctg Q)$. One has
\al{\label{G-tg0q}
\tg_{0_q}(\gp \cdot 0_q) & = \tg_q (\gp \cdot q) \oplus \{0\}\,,
\\ \label{G-kerJ}
\ker\big(\mm'_{0_q}\big) & = \tg_q Q \oplus \left(\mm^{-1}(0) \cap \ctg_q Q\right)\,,
\\ \label{G-splFm-ctg}
\omega_{0_q}\big((X_1,\eta_1) , (X_2,\eta_2)\big) & = \eta_1(X_2) - \eta_2(X_1)\,,
 }
where $X_i \in \tg_qQ$ and $\eta_i \in \ctg_qQ$.
The last equation means that the symplectic form $\omega_{0_q}$ is given by the natural symplectic form of $\tg_q Q \oplus \ctg_q Q$, Moreover, the isotropy representation of $B \in \la_{0_q} = \la_q$ is given by 
\beq\label{G-ItpDst-ctg}
\bbma D(B) & 0 \\ 0 & - D(B)^{\mr T} \ebma\,,
\eeq
where $D :  \la_q \to \End(\tg_q Q)$ is the isotropy representation defined by the action of $\gp$ on $Q$. By \eqref{G-tg0q} and \eqref{G-kerJ}, if $V_q$ is a $G_q$-invariant vector space complement of $\tg_q(\gp \cdot q)$ in $\tg_q Q$, then the subspace $V \subset \tg_{0_q}(\ctg Q)$ defined relative to the splitting \eqref{G-spl-ctg} by
$$
V = V_q \oplus \left(\mm^{-1}(0) \cap \ctg_q Q\right)
$$
is a $G_{0_q}$-invariant complement of $\tg_{0_q}(\gp \cdot 0_q)$ in $\ker(\mm'_{0_q})$. On the one hand, by the special form of $\mm$ in the cotangent bundle situation, the subspace $\left(\mm^{-1}(0) \cap \ctg_q Q\right) \subset \ctg_q Q$ coincides with the annihilator of the subspace $\tg_q(\gp \cdot q) \subset \tg_q Q$. On the other hand, this annihilator may be identified with $V_q^\ast$. Thus, we may write 
\beq\label{G-V-allg-Idtfiz}
V = V_q \oplus V_q^\ast\,.
\eeq
Under this identification, according to \eqref{G-ItpDst-ctg}, the restriction to $V$ of the isotropy representation of $B \in \la_{0_q}$ is given by 
\beq\label{G-ItpDst-V}
\bbma D(B)_{\res V_q} & 0 \\ 0 & - \big(D(B)_{\res V_q}\big)^{\mr T} \ebma\,.
\eeq
Thus, by \eqref{G-splFm-ctg}, the momentum mapping $\mm^V$ defined by \eqref{G-JV} reads 
\beq\label{G-JV-allg}
\langle \mm^V(X,\eta) , B \rangle = \eta\big(D(B) X\big)\,.
\eeq
Now, we apply this to our model. Here, $Q = \gp^N$ and $q = \ul a$. Since the fundamental vector field generated by $B \in \la$ of the action by diagonal conjugation is given by 
$$
B_{\gp^N}(\ul a) = \mr R_{\ul a}' B - \mr L_{\ul a}' B = \mr L_{\ul a}' \big(\Ad(\ul a^{-1}) B - B\big)\,,
$$
we have 
$$
\tg_{\ul a} (\gp \cdot \ul a) = \mr L_{\ul a}' \big\{\Ad(\ul a^{-1}) B - B : B \in \la\big\}\,.
$$
For the complement $V_{\ul a}$, we choose the orthogonal complement with respect to the metric defined by some $G$-invariant scalar product $\langle \cdot , \cdot \rangle$ on $\la$. This leads to
$$
V_{\ul a} = \mr L_a' \{\ul X \in \la^N : \sum_i \Ad(a_i) X_i - X_i = 0\}\,.
$$
This means that under the metric isomorphism, $V_{\ul a}$ corresponds to $\mm^{-1}(0) \cap \ctg_{\ul a} \gp^N$. Using the metric to identify $V_{\ul a}^\ast$ with $V_{\ul a}$, we obtain $V = V_{\ul a} \oplus V_{\ul a}$, with the pairing given by the metric. In analyzing $J^V$, we may omit the transport to $\ul a$. Thus, we may work with 
 \al{\label{G-Va}
V_{\ul a} & = \{\ul X \in \la^N : \sum_i \Ad(a_i) X_i - X_i = 0\}\,,
\\  \label{G-V}
V
 & =
\left\{(\ul X , \ul Y) \in \la^N \oplus \la^N : \sum\nolimits_i \Ad(a_i) X_i - X_i = 0 = \sum\nolimits_i \Ad(a_i) Y_i - Y_i\right\}\,.
 }
The isotropy representation of $B \in \la_{\ul a}$ is given by $D(B)  = \ad(B)$. Thus, \eqref{G-JV-allg} yields
\beq\label{G-JV-2}
\langle \mm^V(\ul X , \ul Y) , B \rangle
 = 
\sum_i \langle Y_i , [B , X_i] \rangle
 =
\left\langle \sum\nolimits_i  [X_i , Y_i ] , B\right\rangle
 \,,\quad
B \in \la_{\ul a}\,.
\eeq

Next, we restrict attention to the case $G = \SU(2)$. Under the identification of $\su(2)$ with $\RR^3$ endowed with the cross product, see e.g.\ Example 5.2.8 in \cite{BuchI}, \eqref{G-JV-2} reads 
\beq\label{G-JV-R3}
\langle \mm^V(\ul X , \ul Y) , B \rangle = \sum_i (X_i  \times Y_i) \cdot B
 \,,\quad
B \in \la_{\ul a} \subset \RR^3\,.
\eeq
Here, the dot denotes the standard scalar product in $\RR^3$. Let $Z$ and $T$ denote the center and the toral subgroup of diagonal matrices, respectively. The following stabilizers occur \cite{FJRS}.

\ben

\item[$(Z)$] $\gp_{\ul a} = Z$. This is the generic case provided $N \geq 2$.

\item[$(T)$] $\gp_{\ul a}$ is conjugate to $T$. This happens if and only if all $a_i$ commute but at least one of them is not $\pm \II$. Since a tubular neighbourhood is $\gp$-invariant, without loss of generality we may assume $a_i \in T$ and $\gp_{\ul a} = T$.

\item[$(G)$] $\gp_{\ul a} = \SU(2)$. This holds in the case $\ul a \in Z^N$, that is, $a_i = \pm \II$ for all $i$.

\een

In case $\gp_{\ul a} = Z$, we have $\mf g_{\ul a} = 0$, so that $\mm^V \equiv 0$ and hence $I(\mm^V) = 0$, which is a real radical ideal, indeed. 

\paragraph{Stabilizer  $T$}

Here,  $a_1 , \dots , a_N \in T$. Under the isomorphism $\su(2) \cong \RR^3$, the Lie subalgebra of $T$ corresponds to the subspace spanned by $\vec e_1 \in \RR^3$. Hence, by \eqref{G-JV-R3},
\beq\label{G-JV-T}
\mm^V(\ul X , \ul Y) = \sum_i \big((X_i \times Y_i) \cdot \vec e_1\big) \vec e_1\,.
\eeq
Since the action is Abelian, we may use one of the criteria of Theorem 6.8 in \cite{AGJ} to show that the ideal generated by the functions $\mm^V_B$ with $B \in \mf g_{\ul a}$ in $\calC^\infty(V)$ is a real radical ideal. Then, Theorem 6.3 of that work ensures that the ideal generated by these functions in $\RR[V]$ is a real radical ideal, too\todo{proving theorem 6.8 for $\calC^\infty$ is harder than proving it for $\RR[V]$, since we only need it for the latter, we may give a simpler direct proof?}. The criterion provided by Theorem 6.8 we use is that the following condition holds at every point $(\ul X,\ul Y) \in (\mm^V)^{-1}(0)$.
\medskip

{\bf Nonpositivity condition}.
{\it 
For every $B \in \mf g_{\ul a}$, either $\mm^V_B = 0$ in some neighbourhood of $(\ul X,\ul Y)$ in $V$ or any neighbourhood of $(\ul X,\ul Y)$ contains points $(\ul X_\pm,\ul Y_\pm)$ such that $\mm^V_B(\ul X_-,\ul Y_-) < 0$ and $\mm^V_B(\ul X_+,\ul Y_+) > 0$.
}
\medskip

Clearly, we may restrict attention to $B = \vec e_1$. Without loss of generality, we may assume that $a_N \neq \pm\II$. Then, we may find $\vec x$ in the $x_2$-$x_3$ plane such that 
\beq\label{G-nonpos-1}
\sum_{i=1}^{N-1} \big(\Ad(a_i) \vec e_2 - \vec e_2\big) + \Ad(a_N) \vec x - \vec x = 0
\,.
\eeq
Then, denoting the rotation by $\pi/2$ about the $x_1$-axis by $R$, we have
\beq\label{G-nonpos-2}
\sum_{i=1}^{N-1} \big(\Ad(a_i) \vec e_3 - \vec e_3\big) + \Ad(a_N) R \vec x - R \vec x = 0
\,.
\eeq
Define curves
$$
\gamma_\pm(t)
 := 
 \left(
\ul X + t (\vec e_2 , \dots , \vec e_2 , \vec x)
 ,
\ul Y \pm t (\vec e_3 , \dots , \vec e_3 , R \vec x)
 \right)
\,.
$$
By \eqref{G-nonpos-1} and \eqref{G-nonpos-2}, $\gamma_\pm(t) \in V$ for all $t$. We compute
 \ala{
\mm^V_B\big(\gamma_\pm(t)\big)
 & =
\sum_{i=1}^{N-1} 
\big((X_i + t \vec e_2) \times (Y_i \pm t \vec e_3)\big) \cdot \vec e_1
 +
\big((X_N + t \vec x) \times (Y_N \pm t R \vec x)\big) \cdot \vec e_1
\\
 & =
\alpha_\pm t \pm \beta t^2
 }
with
 \ala{
\alpha_\pm
 & =
\sum_{i=1}^{N-1} (\vec e_3 \cdot Y_i \pm \vec e_2 \cdot X_i)
 +
(\vec e_1 \times \vec x) \cdot Y_N \pm \big((R\vec x) \times \vec e_1\big) \cdot X_N
\,,
\\
\beta 
 & =
N + \|\vec x\|^2 - 1
\,.
 }
If $\alpha_+ \neq 0$, the zeros of the polynomial $\alpha_+ t + N t^2$ are distinct and hence $\mm^V_B\big(\gamma_\pm(t)\big)$ takes both positive and negative values in any neighbourhood of $t=0$. A similar argument applies if $\alpha_- \neq 0$. Finally, if both $\alpha_\pm = 0$, then $\mm^V_B\big(\gamma_+(t)\big) > 0$ and $\mm^V_B\big(\gamma_-(t)\big) < 0$ for any $t \neq 0$.

This shows that, in the case where $\ul a$ has stabilizer $T$, the ideal generated by the functions $\mm^V_B$ with $B \in \mf g_{\ul a}$ in $\calC^\infty(V)$ is a real radical ideal.

\paragraph{Stabilizer $\SU(2)$}

Here, $\ul a \in Z^N$ and hence $\mf g_{\ul a} = \mf g \equiv \RR^3$. Hence, by \eqref{G-V} and \eqref{G-JV-R3},
\beq\label{G-JV-G}
V = \la^N \oplus \la^N
 \,,\qquad
\mm^V(\ul X , \ul Y) = \sum_i X_i \times Y_i\,.
\eeq

In this case, we have the following

\bsz\label{S-ErzBed-G}

For every $\ul a \in Z^N$, the ideal $I(\mm^V)$ is a real radical ideal.

\esz

\bpf

Denote $\mm^V_k := \mm^V \cdot \vec e_k$. Thus,
$$
\mm^V_k(\ul X,\ul Y) = \sum_i (X_i \times Y_i) \cdot \vec e_k = \sum_i X_{im} Y_{in} \ve_{mnk}
$$
(summation convention). By letting $\ul X , \ul Y \in \la_\CC^N$, we can extend $\mm^V_k$ to polynomial functions $\tilde\mm^V_k$ on $V_\CC$. Let $I_\CC$ denote the ideal in $\CC[V_\CC]$ generated by these extensions. 

To prove the assertion, we apply the criterion of Theorem 6.5 in \cite{AGJ} \todo{Originalzitat!}, which states that $I(\mm^V)$ is a real radical ideal in $\RR[V]$ if and only if 

\ben

\item $I_\CC$ is radical, meaning that
$$
I_\CC = \{f \in \CC[V_\CC] : f^k \in I_\CC \text{ for some } k\}\,,
$$

\item for every irreducible component $W$ of the zero locus $(\tilde\mm^V)^{-1}(0)$ of the $\tilde\mm^V_k$, the real dimension of (the smooth part of) $W \cap V$ coincides with the complex dimension of (the smooth part of) $W$.

\een

To check these conditions, we apply Theorem 7.8 in \cite{AGJ} \todo{Originalzitat!}, which states that if all $W$ contain a point $(\ul X,\ul Y)$ where the differentials (tangent mappings) $\mr d \tilde\mm^V_k(\ul X,\ul Y)$ are linearly independent, then $I_\CC$ is radical and the complex dimension of $W$ is $3(2N-1)$. We compute
$$
\mr d \tilde\mm^V_k(\ul X,\ul Y)
 = 
\sum_i  (\mr d X_i \times Y_i + X_i \times \mr d Y_i) \cdot \vec e_k
 = 
\sum_i \big( (Y_i \times \vec e_k) \cdot \mr d X_i + (\vec e_k \times X_i) \cdot \mr d Y_i\big) \ .
$$
Hence, for $\lambda_1 , \lambda_2 , \lambda_3 \in \CC$,
$$
\sum_k \lambda_k \mr d \mm^V_k(\ul X,\ul Y)
 =
\sum_i  (Y_i \times \vec\lambda) \cdot \mr d X_i + (\vec\lambda \times X_i) \cdot \mr d Y_i\,.
$$
This vanishes if and only if
$$
\vec\lambda \times X_i = 0
 ~\text{ and }~
\vec\lambda \times Y_i = 0
 ~\text{ for all $i$.}
$$
This system of linear equations has a nontrivial solution $\vec\lambda$ if and only if all $X_i$ and $Y_i$ are parallel. We check that the subset $M^V \subset (\tilde\mm^V)^{-1}(0)$ of points violating this condition is dense. Let $(\ul X,\ul Y) \in (\tilde\mm^V)^{-1}(0)$ such that all $X_i$ and $Y_i$ are parallel, i.e., $X_i = \xi_i \vec a$ and $Y_i = \upsilon_i \vec a$ with $\vec a \in \CC^3 \setminus 0$ and $\xi_i , \upsilon_i  \in \CC$. We construct a curve $\gamma(t)$ such that $\gamma(0) = (\ul X,\ul Y)$ and $\gamma(t) \in M^V$ for all $t \neq 0$. Choose $\vec b \in \CC^3 \setminus 0$ so that $\vec a$ and $\vec b$ are not parallel. We have to distinguish the following cases. If $\xi_1 , \upsilon_2 \neq 0$, we put
$$
\gamma(t) := \big((X_1 , t \xi_1 \vec b + (1-t) X_2 , X_3 , \dots , X_N) , (t \upsilon_2 \vec b  + (1-t) Y_1 , Y_2 , Y_3 , \dots , Y_N)\big)\,.
$$
If $\xi_1 \neq 0$ and $\upsilon_2 = 0$, then 
$$
\gamma(t) := \big((X_1 , t \xi_1 \vec b + (1-t) X_2 , X_3 , \dots , X_N) , (t^2 \vec b + (1-t) Y_1 , t \vec a , Y_3 , \dots , Y_N)\big)
$$
and analogously for $\xi_1 = 0$ and $\upsilon_2 \neq 0$. Finally,  if $\xi_1 = \upsilon_2 = 0$, then 
$$
\gamma(t) := \big((t \vec a , t \vec b + (1-t) X_2 , X_3 , \dots , X_N) , (t \vec b + (1-t) Y_1 , t \vec a , Y_3 , \dots , Y_N)\big)\,.
$$
We leave it to the reader to check that in each case, $\tilde\mm^V\big(\gamma(t)\big) = 0$ for all $t$. Then, $\gamma(t) \in M^V$ for all $t \neq 0$. 

As a consequence, Theorem 7.8 in \cite{AGJ} cited above yields that $I_\CC$ is radical and that the irreducible components of $(\tilde\mm^V)^{-1}(0)$ have complex dimension $3(2N-1)$. In view of Theorem 6.5 in \cite{AGJ} cited above, to prove the assertion it remains to show that for all irreducible components $W$ of $(\tilde\mm^V)^{-1}(0) $, the real dimension of $W \cap V$ is $3(2N-1)$, too. Now, $W \cap V$ is an irreducible component of $(\mm^V)^{-1}(0)$ and the argument showing that $M^V$ is dense in $(\tilde\mm^V)^{-1}(0)$ applies without change to the subset of $(\mm^V)^{-1}(0)$ of points $(\ul X,\ul Y)$, where the differentials $\mr d \mm^V_k$ are linearly independent. This yields the assertion.
\epf


\subsection{The acyclicity condition}


We apply Theorem 3.1 in \cite{BHP}. According to this theorem, if the generating hypothesis is satisfied, a sufficient condition for the acyclicity condition to hold is that the set of points $(\ul a , \ul A)$ where $\mm'_{(\ul a,\ul A)}$ is surjective is dense in $\mm^{-1}(0)$. To check this, we need the following lemma. Given $a \in G$, let $\mr C_\la(a)$ denote the centralizer of $a$ in $\la$, i.e.,
$$
\mr C_\la(a) := \{X \in \la : \Ad(a) X = X\}
\,.
$$

\ble\label{L-AzyBed-1}

The orthogonal complement of $\im\big(\mm'_{(\ul a,\ul A)}\big)$ in $\la$ is
$$
\left(\bigcap\nolimits_i \mr C_\la(a_i)\right) \cap \left(\bigcap\nolimits_i \mr C_\la\big(\Ad(a_i) A_i\big)\right)\,.
$$

\ele

\bpf

We compute
$$
\mm'_{(\ul a,\ul A)} (\mr L_{\ul a}' \ul X , \ul Y)
 =
\sum_i \big( \Ad(a_i)[X_i,A_i] + \Ad(a_i) Y_i - Y_i\big) \ .
$$
Hence, $B \in \la$ is orthogonal to $\im\big(\mm'_{(\ul a,\ul A)}\big)$ if and only if
$$
\left\langle\sum\nolimits_i \big( \Ad(a_i)[X_i,A_i] + \Ad(a_i) Y_i - Y_i , B \big) \right\rangle = 0
 ~\text{ for all }~ \ul X , \ul Y \in \la^N\,.
$$
This is equivalent to 
 \ala{
\left\langle \Ad(a_i)[X_i,A_i] , B \right\rangle
 & = 
\left\langle X_i , [A_i , \Ad(a_i^{-1}) B] \right\rangle  = 0
 ~\text{ for all }~ X_i \in \la\,,
\\ 
\left\langle \Ad(a_i) Y_i - Y_i , B \right\rangle
 & = 
\left\langle Y_i , \Ad(a_i^{-1}) B - B] \right\rangle  = 0
 ~\text{ for all }~ Y_i \in \la\,,
 }
that is, to
$$
[\Ad(a_i) A_i , B] = 0 ~\text{ and }~ \Ad(a_i) B = B ~\text{ for all }~ i\,.
$$
This yields the assertion.
\epf

Now, as before, let $G=\SU(2)$ and let $T \subset \SU(2)$ denote the subgroup of diagonal matrices.

\ble\label{L-AzyBed-2}

Let $a_1 , a_2 \in T \setminus \{\pm\II\}$. For every $B_1 \in \la$, there exists $B_2 \in \la$ such that $\Ad(a_1) B_1 - B_1 + \Ad(a_2) B_2 - B_2 = 0$.

\ele

\bpf

We identify the adjoint action of $\gp$ on $\la$ in the usual way with the action on $\RR^3$ defined via the covering homomorphism $\SU(2) \to \SO(3)$. Then, the Lie subalgebra $\mf t$ associated with $T$ is given by the $x_1$-axis and the linear transformations $\Ad(a)$ with $a \in T$ correspond bijectively to the rotations about this axis. If $\Ad(a_1) B_1 - B_1 = 0$, we may choose $B_2=0$. Otherwise, $B_1 \notin \mf t$ and hence also $\Ad(a_2) B_1 - B_1 \neq 0$. Since both $\Ad(a_1) B_1 - B_1$ and $\Ad(a_2) B_1 - B_1$ belong to the $x_2$-$x_3$-plane, there is $a \in T$ such that 
$$
\Ad(a_1) B_1 - B_1 = \lambda \Ad(a) (\Ad(a_2) B_1 - B_1)
\,.
$$
Putting $B_2 := - \lambda \Ad(a) B_1$, we obtain the desired result.
\epf

\bsz\label{S-AzyBed}

For $\gp = \SU(2)$, the acyclicity condition is satisfied.

\esz

\bpf

As noted above, in view of Theorem 3.1 in \cite{BHP}, it suffices to check that the set of points $(\ul a , \ul A)$ where $\mm'_{(\ul a,\ul A)}$ is surjective is dense in $\mm^{-1}(0)$ . In view of Lemma \rref{L-AzyBed-1}, $\mm'_{(\ul a,\ul A)}$ is surjective if and only if 
$$
I(\ul a,\ul A) := \left(\bigcap\nolimits_i \mr C_\la(a_i)\right) \cap \left(\bigcap\nolimits_i \mr C_\la\big(\Ad(a_i) A_i\big)\right)
$$
vanishes. Let $(\ul a,\ul A) \in \mm^{-1}(0)$ be given. The subspace $I(\ul a,\ul A)$ can be $0$, $\mf t$ or $\la$. In the first case, nothing has to be shown. In the other two cases, we will construct a path $ t \mapsto \ul A(t)$ such that 
\beq\label{G-S-AzyBed}
\ul A(0) = \ul A
 \,,\qquad
(\ul a , \ul A(t)) \in \mm^{-1}(0) ~~ \forall~t
 \,,\qquad
I(\ul a , \ul A(t)) = 0~~ \forall~t \neq 0\,.
\eeq
This will yield the assertion. In case $I(\ul a,\ul A) = \la$, one has $a_i = \pm \II$ and $A_ i= 0$ for all $i$. There exist $B_1$, $B_2 \in \la$ such that $\mr C_\la(B_1) \cap \mr C_\la(B_2) = 0$. Define 
$$
\ul A(t) := (t B_1 , t B_2 , 0 , \dots , 0) \,,\qquad t \in \RR\,.
$$
In case $I(\ul a,\ul A) = \mf t$, one has $a_i \in T$ and $A_i \in \mf t$ for all $i$. Then, $\Ad(a_i) A_i = A_i$ for all $i$. If $a_i = \pm\II$ for all $i$, then $\bigcap\nolimits_i \mr C_\la(a_i) = \mf g$, so for some $j$ we must have $\mr C_{\mf g}(A_j) = \mf t$. Assuming without loss of generality that $j=1$, we may choose $B \in \mf g \setminus \mf t$ and put 
\beq\label{G-S-AzyBed-1}
\ul A(t) := (A_1 , A_2 + t B , A_3 , \dots , A_N) \,,\qquad t \in \RR\,.
\eeq
Then, $I(\ul a , \ul A(t)) \subset \mr C_\la(A_1) \cap \mr C_\la(A_2 + t B) = \{0\}$ for all $t \neq 0$. If there are $j$ such that $a_j \neq \pm \II$ and $k$ such that $a_k = \pm\II$, then, assuming $j=1$ and $k=2$, we may choose $B \in \mf g \setminus \mf t$ and define $A(t)$ by \eqref{G-S-AzyBed-1}. Then, $I(\ul a , \ul A(t)) \subset \mr C_\la(a_1) \cap \mr C_\la(A_2 + t B) = \{0\}$ for all $t \neq 0$. Finally, if $a_i \neq \pm \II$ for all $i$, then we may choose $B_1 \in \la \setminus \mf t$, apply Lemma \rref{L-AzyBed-2} to find $B_2$ such that  $\Ad(a_1) B_1 - B_1 + \Ad(a_2) B_2 - B_2 = 0$, and put 
$$
\ul A(t) := (A_1 + t B_1 , A_2 + t B_2 , A_3 , \dots , A_N) \,,\qquad t \in \RR\,.
$$
Clearly, $A_1 + \Ad(a_1) B_1 \notin \mf t$. Thus, $I(\ul a , \ul A(t)) \subset \mr C_\la(a_1) \cap \mr C_\la(A_1 + t B) = \{0\}$ for all $t \neq 0$.
\epf


\subsection{Invariance conditions}


In this subsection, we check the invariance and covariance conditions introduced in subsection \ref{sec:quantizedversion}. 
In the sequel, we assume that $G$ is connected.

Concerning the invariance condition, we have the following well-known criterion, see \cite{GR,MuellerBahnsNeumaier,BBG}.

\bsz
\label{G-inv-star}

Let $(M,\omega, \mf g)$ be a symplectic $\mf g$-manifold, $\nabla$ a torsion-free, symplectic connection on $M$ and $\Omega \in \nu Z^2(M)[[\nu]]$ a series of closed two-forms. Then, the star-product constructed from these data is $\mf g$-invariant if and only if $\nabla$ and $\Omega$ are $\mf g$-invariant.

\esz

By the assumption of connectedness of $G$, we obtain the corresponding statement for the $G$-action. In the case at hand we have $\Omega = 0$ and, by Proposition \ref{S-nabla-inv}, the lifted connection is $G$-invariant. This implies the following. 
\bfg
\label{S-Ivr}

The Fedosov star product of the standard ordered type is $G$-invariant.

\efg

\begin{Remark}
Using the concrete form of the family of bi-differential operators $B$ derived in Subsection \ref{Fedosov-star}, 
one can check the $G$-invariance of the star product also by direct inspection. We leave this as an exercise to the 
reader. 
\end{Remark}

Covariance will be implied by the following lemma.

\ble\label{L-Cvr}

For $B \in \mf g$ and $f \in C^\infty(\ctg G^N)$, we have
$$
\mm_B \star f - f \star \mm_B = \mr i \lambda \{f , \mm_B\} + O(\lambda^2)
\,,
$$
where in case $f$ is fiber-homogeneous of order $n$, the highest order in $\lambda$ is $\lambda^n$.

\ele

\bbw

We observe that according to \eqref{G-ImpAbb}, the function $\mm_B$ is linear in the fiber variable $\ul\alpha$. Hence, $(0,\ul\ve^I) (0,\ul\ve^K) (\ul E_{I_1},0) \cdots (\ul E_{I_r},0) \mm_B = 0$ for any $I$, $K$ and $I_1 , \dots , I_r$. As a consequence, in the expansions \eqref{G-star-s} of $\mm_B \star f$ and of $f \star \mm_B$, only the contributions of $B_0$ and $B_1$ survive. Thus,
 \ala{
\mm_B & \star f - f \star \mm_B
\\
 = &
\frac{\lambda}{\mr i} 
 \bigg(
B_1(\mm_B,f) - B_1(f,\mm_B)
 + 
B_0\big((0,\ul\ve^I) \mm_B , (\ul E_I,0) f\big)
 -
B_0\big((0,\ul\ve^I) f , (\ul E_I,0) \mm_B\big)
 \bigg)
\\
 &
 -
\sum_{m=2}^\infty
\left(\frac{\lambda}{\mr i} \right)^m
 \left(
\frac{1}{(m-1)!} 
B_1
 \big(
(0,\ul\ve^{I_1}) \cdots (0,\ul\ve^{I_{m-1}}) f
 , 
(\ul E_{I_1},0) \cdots (\ul E_{I_{m-1}},0) \mm_B
 \big)
\right.
\\
 & \hspace{3cm} + \left.
\frac{1}{m!} 
B_0
 \big(
(0,\ul\ve^{I_1}) \cdots (0,\ul\ve^{I_m}) f
 , 
(\ul E_{I_1},0) \cdots (\ul E_{I_m},0) \mm_B
 \big)
 \right)
\,.
 }
Clearly, if $f$ is fiber-homogeneous of order $n$, the sum over $m$ runs up to $m=n$. Consider the contributions of first order in $\lambda$. The formulae for $B_0$ and $B_1$ given in Remark \rref{Bem-Mr} yield 
 \ala{
\mm_B \star f - f \star \mm_B
 & =
\frac{\lambda}{\mr i} 
 \bigg(
\big((0,\ul\ve^I) \mm_B\big) \cdot \big((\ul E_I,0) f\big)
 -
\big((0,\ul\ve^I) f\big) \cdot \big((\ul E_I,0) \mm_B\big)
\\
 & \hspace{1cm} +
[\ul E_I , \ul E_K]^\sim \cdot \big((0,\ul\ve^I) \mm_B\big) \cdot \big((0,\ul\ve^K) f\big)
 \bigg)
 +
O(\lambda^2)
\,.
 }
On the other hand, expanding the partial differentials $\mr d_G f$ and $\mr d_{\mf g^\ast} f$ defined by \eqref{G-dGf} wrt.\ the bases $\{\ul\ve^I\}$ in $(\mf g^\ast)^N$ and $\{\ul E_I\}$ in $\mf g^N$, we find
$$
\mr d_G f = \big((\ul E_I,0)f\big) \ul\ve^I
\,,\qquad
\mr d_{\mf g^\ast} f = \big((0,\ul\ve^I)f\big) \ul E_I
\,.
$$
Then,
$$
[\mr d_{\mf g^\ast} \mm_B , \mr d_{\mf g^\ast} f]
 =
\big((0,\ul\ve^I)\mm_B\big) \cdot \big((0,\ul\ve^K)f\big) [\ul E_I,\ul E_K]
$$
and \eqref{G-PoKl} yields 
 \ala{
\{\mm_B , f\}
 & =
\big((\ul E_I,0)f\big) \cdot \big((0,\ul\ve^I)\mm_B\big)
 -
\big((\ul E_I,0)\mm_B\big) \cdot \big((0,\ul\ve^I)f\big)
\\
 & \hspace{1cm} +
\big((0,\ul\ve^I)\mm_B\big) \cdot \big((0,\ul\ve^K)f\big) \cdot [\ul E_I,\ul E_K]^\sim
\,.
 }
This yields the assertion.
\ebw

\bsz\label{S-Cvr}

The Fedosov star product of the standard ordered type is covariant.

\esz

\bbw

We have to show that 
$$
\mm_B \star \mm_C - \mm_C \star \mm_B = \mr i \lambda \mm_{[B,C]}
$$
for all $B, C \in \mf g$. Since both $\mm_B$ and $\mm_C$ are linear in the fiber variable, Lemma \rref{S-Cvr} yields
$$
\mm_B \star \mm_C - \mm_C \star \mm_B = \mr i \lambda \{\mm_C , \mm_B\}
\,.
$$
Since $\mm$ is equivariant, $\{\mm_C , \mm_B\} = \mm_{[B,C]}$, cf.\ eg.\ \cite[Prop.\ 10.1.14]{BuchI}.
\ebw


To summarize, the standard ordered Fedosov star product on $\calC^\infty (T^*G^N)$
is invariant and covariant for $G=\SU(2)$, but by Lemma \ref{L-Cvr}
it is not necessarily strongly invariant. Combining this with the fact that conditions
\hyperref[GenCond]{{\rm(GH)}} and \hyperref[AcycCond]{{\rm (AC)}}
are satisfied, Theorems \ref{thm:homologicalreductionstarproduct} and \ref{thm:homologicalreduction}
can be applied in this situation, but not Corollary \ref{cor:stronglyinvariant}.
More precisely, one concludes the following.

\begin{Corollary}
    Let $G=\SU(2)$ and $(T^*G^N,\omega,\ul\Psi,J)$ be the associated lattice gauge model.
    Then the  standard ordered Fedosov star product $\star$ on $\calC^\infty (T^*G^N)$ gives rise to
    a star product $\ast$ on the quantized BRST algebra which then, after appropriate choices of an
    extension map $\ext$, a homotopy $h$ and deformations $\bm{\rest}$ and $\bm{h}$ of the restriction map
    and homotopy, respectively, induces a star product $\tilde{\star}$ on
    $H^0 (\g, \calC^\infty (T^*G^N)[[\lambda]])$ by
  \begin{displaymath}
    f \,  \tilde{\star}\,  g = \bm{\rest} (\ext (f) \ast \ext (g))
    \quad\text{for } f, g \in Z^0 (\g, \calC^\infty (T^*G^N)[[\lambda]]) \  .
  \end{displaymath}
  Moreover, 
  there exists a star product $\hat{\star}$ on the reduced phase space $T^*G^N\red G$ 
  of the form
  \begin{displaymath}
    f \, \hat{\star} \, g = S^{-1} (S (f)\, \tilde{\star} \, S(g))
    \quad\text{for } f, g \in \calC^\infty (T^*G^N\red G) \cong \calC^\infty (T^*G^N_0)^G \ ,
  \end{displaymath}
  where
  $S : \calC^\infty(T^*G^N_0)^G [[\lambda]] \to H^0 (\g,\calC^\infty (T^*G^N_0)[[\lambda]])$ is a
  topological linear isomorphism onto
  $H^0 (\g,\calC^\infty (T^*G^N_0)[[\lambda]]) \subset \calC^\infty (T^*G^N_0) [[\lambda]]$ of the form 
  $S = \sum_{k\in \N} \lambda^k S_k$ with $S_k : \calC^\infty (T^*G^N_0) \to \calC^\infty (T^*G^N_0)$
  continuous linear. 
\end{Corollary}


\section{Outlook}

There is a variety of challenging open problems which may be subject to future work:
\ben
\item
Clearly, the star product on the reduced phase space is given in a complicated implicit way. To make it more explicit, one has to study the deformation retract structure entering the whole construction in more detail. 

\item
The model under consideration carries a natural K\"ahler structure. Thus, it will be interesting to derive the corresponding 
star product of Wick type. Moreover, it will be easy to find the star product of Weyl type. Thereafter, it will be possible 
to compare the properties of these products with the star product of standard order type dealt with in this paper. 

\item
  One should try to extend the results of Section \ref{sec:application} to other Lie groups, notably to
  $\SU(n)$ for $n>2$. In particular, it would be interesting to find examples for which the conditions
  \hyperref[GenCond]{{\rm(GH)}} and \hyperref[AcycCond]{{\rm (AC)}} are not fulfilled. 
  These are nontrivial tasks, because in each new case one deals with a new, different stratified structure
  and, thus, it seems to be hard to find general arguments. Moreover, it is likely that in some cases
  \hyperref[GenCond]{{\rm(GH)}} will be fulfilled, but \hyperref[AcycCond]{{\rm (AC)}} not,
  or even the other way around.
  For the analysis of Condition \hyperref[GenCond]{{\rm(GH)}} it is crucial to study the ideal generated
  by the components of the linearized moment map $J^V$ for various Lie groups or classes of Lie groups to
  see whether it is radical or not. This should be possible by means of real algebraic geometry.  
  If \hyperref[GenCond]{{\rm(GH)}} is fulfilled, then by Theorem 3.1 in \cite{Schwarzbier}, checking condition
  \hyperref[AcycCond]{{\rm (AC)}} boils down to checking that the set of points for 
  which the tangent mapping of the moment map is surjective is dense in the zero level set of the moment map.
  In case Condition \hyperref[AcycCond]{{\rm (AC)}} is not satisfied, there is no finite resolution
  of the classical observable algebra of the reduced phase space (as a module of the observable algebra of
  the unreduced space), but there still is a resolution of inifinite length, namely the
  Koszul--Tate resolution \cite{Avramov}. This method already has proved a powerful tool in the
  quantization of gauge  theories \cite{BarBraHen}.
  One can expect that it will be one too for quantized homological
  reduction where Condition \hyperref[AcycCond]{{\rm (AC)}} is not satisfied.
\item
  Our paper deals with formal deformation quantization only. It is a challenge to clarify whether the
  homological reduction method may be developed for strict deformation quantization
  (see e.g.\ \cite{Landsman}) as well. As already mentioned in the introduction, this would make it possible
  to compare the quantum observable algebra structure obtained here with the observable algebra obtained via
  canonical quantization described above in closer terms.  
  It appears to be promising to use methods from complex analysis as they were used in \cite{Schmitt}
  for a strict quantization of coadjoint orbits. 
\end{enumerate}


\appendix

\section{Tools  from homological algebra}
\label{ToolsHomAlg}
For the convenience of the reader we recall here some examples of complexes from homological algebra
which are crucial for our paper and the fundamental concepts of homological perturbation theory.
For more details on the former we refer the reader to \cite{HilStaCHA,GelMan,Weibel}, for the latter to
\cite{Stasheff,LamHPT,CrainicPerturbation}. 

\begin{Example}
  Let $R$ be a unital ring. Then every (ungraded) $R$-module $M$ can be understood as a
  cochain complex $M^\bullet$ concentrated in degree $0$ by putting
\[
  M^k =
  \begin{cases}
    M & \text{for } k=0 ,\\
    0 & \text{else.} 
  \end{cases}
\]
Likewise one constructs the chain complex $M_\bullet$ concentrated in degree $0$.
\end{Example}

\begin{Example}\label{AppKoszulCplx}
  Let $R$ be a commutative ring, $E$ a free $R$-module of finite rank $d$, and $x: E \to R$
  an $R$-linear map. Then the \emph{Koszul complex} on $x$ is the chain complex of $R$-modules
  \[
    K_\bullet (x) : 0 \longleftarrow R \overset{\partial}{\longleftarrow} \Lambda^1 E 
    \overset{\partial}{\longleftarrow} \Lambda^2 E \ \overset{\partial}{\longleftarrow}  \ldots \overset{\partial}{\longleftarrow}
    \Lambda^d E  \longleftarrow 0 \ ,
  \]
  where the \emph{Koszul differential} $\partial : \Lambda^k E \to  \Lambda^{k-1} E$ is given by
  \[
    \partial ( e_1 \wedge \ldots \wedge e_k) =  \sum_{l=1}^k x(e_l) \,  e_1 \wedge \ldots \wedge\widehat{e_l}\wedge\ldots\wedge e_k 
    \quad\text{for all } e_1 \ldots , e_k \in E \ .
  \]
  Under an isomorphism $E \cong R^d$, the map $x$ can be identified with a sequence $x_1,\ldots x_d$ of $R$-linear maps
  $x_l: R \to R$. It is a classical result in commutative algebra that the Koszul complex $K_\bullet (x)$ is acyclic
  if $x_1,\ldots x_d$ is a regular sequence that is if $x_l$ is a not a zero-divisor on $R/(x_1,\ldots , x_{l-1})$
  for $l = 1,\ldots , d$. In this case, $H_0 (K_\bullet (x)) $ coincides with the quotient ring
  $S = R /(x_1,\ldots ,x_d)$, and the Koszul complex is a free resolution of $S$ in the category of $R$-modules. 
\end{Example}

\begin{Example}\label{AppCEcomplex}
  Let $\g$ be a Lie algebra and $V$ a $\g$-module.
  The \emph{Chevalley--Eilenberg complex} $\left( \CE^\bullet (\g, V) , \delta\right)$  then is the cochain complex
  \[
    \CE^\bullet (\g, V) = \Hom \left( \Lambda^k \g , V \right) : 0 \rightarrow
    N \overset{\delta}{\rightarrow} V \otimes \g^* \rightarrow \ldots  \rightarrow V \otimes \Lambda^\ell \g^* \rightarrow 0 
  \]
  with the \emph{Chevalley--Eilenberg coboundary} $\delta : \CE^k (\g,V) \to \CE^{k+1} (\g,V)$  given by
\begin{equation*}
  \begin{split}
     \delta f \, (X_1,\ldots ,X_{k+1}) = \, & \sum_{1\leq i \leq k+1} (-1)^{i+1} X_i \cdot
     f(X_1,\ldots ,\widehat{X}_i , \ldots , X_{k+1} )  \\
     & + \sum_{1\leq i < j \leq k+1} (-1)^{i+j} f([X_i,X_j], X_1,\ldots ,\widehat{X}_i , \ldots
     \widehat{X}_j , \ldots , X_{k+1} ) \ ,
  \end{split}
\end{equation*}
for all $f \in \CE^k (\g,V)$ and $X_1,\ldots,X_{k+1} \in \g$. Chevalley and Eilenberg showed in \cite{CheEil}
that $\delta^2 =0$, so $\big( \CE^\bullet (\g,V), \delta \big)$ is a cochain complex indeed. 
Its cohomology is the \emph{Lie algebra cohomology} of $\g$ with values in the
$\g$-module $V$ and is denoted $H^\bullet (\g,V)$. Note that $H^0 (\g, V )$ coincides with the invariant part $V^\g$. 
\end{Example}

Of particular importance for our considerations is the following concept.

\begin{Definition}[cf.~{\cite[2.1 \& 2.3]{CrainicPerturbation}}]
  By a \emph{deformation retract} one understands a triple $(i,p,h)$ 
  consisting of a quasi-isomorphism of cochain complexes $i:(C^\bullet,\delta) \to (D^\bullet,d)$, a
  quasi-inverse $p: (D^\bullet,d) \to (C^\bullet,\delta)$ so that $p \circ i =\id_{C^\bullet}$ and a degree $-1$ graded map
  $h =(h^k)_{k\in \Z}$ which is a chain homotopy from $i \circ p$ to $\id_{C^\bullet}$
  that is which satisfies
  \[
    i \circ p - \id = h \circ d + d \circ h  \ .
  \]
  One usually denotes a deformation retract in the form 
  \begin{equation}
    \label{eq:DR}
    \left( (C^\bullet,\delta) \overset{i}{\underset{p}{\rightleftarrows}} (D^\bullet,d) , h \right) \ .
  \end{equation}
  A deformation retract $(i,p,h)$ is called \emph{special} if the conditions
\begin{equation}
  \label{eq:sideconditions}
  h \circ h = 0  \ , \quad h \circ i = 0 \ ,\quad  p \circ h = 0 \ ,
\end{equation}
are satisfied. Sometimes these conditions are referred to as
\emph{side conditions 1, 2}, and \emph{3}, respectively. 
\end{Definition}

Note that by properly changing the homotopy $h$ of a deformation retract one can achieve that the three side conditions hold
true, see \cite{LamHPT}.

In homological perturbation theory \cite{LamHPT,CrainicPerturbation,HueCAS,HueOBTHH} one studies the behavior of a deformation retract $(i,p,h)$
under perturbation. By the latter one understands a differential $\textsf{D} : D^\bullet \to D^\bullet$ of the form $\textsf{D} = d + t $ which
means that $t$ is a graded map of the same degree as $d$ and $(d + t)^2 = 0$.

\begin{Namedthm}{Perturbation Lemma}
\label{thm:perturbationlemma}  
  Let $\left( (C^\bullet,\delta) \overset{i}{\underset{p}{\rightleftarrows}} (D^\bullet,d) , h \right) $ be a
  deformation retract of filtered complexes satisfying side condition (3) and $\textsf{D} = d + t $ a perturbation. Assume that $\tau := pti$ satisfies
  $\tau p = p t$ and that $th +ht$ raises the filtration. Then $\id + \, th +ht$ is invertible, $\Delta = \delta + \tau$ is a differential on
  $C^\bullet$ and 
  \begin{equation}
    \label{eq:pDR}
    \left( (C^\bullet, \Delta) \overset{\textsf{i}}{\underset{p}{\rightleftarrows}} ( D^\bullet, \textsf{D}) , H \right) 
  \end{equation}
  is a deformation retract with $H=h (\id + th + ht)^{-1}$ and $\textsf{i} = i -H(ti - i\tau)$.
  If all side conditions hold for $(i,p,h)$, then they hold for \eqref{eq:pDR}.
\end{Namedthm}

\begin{proof}
  See the appendix of \cite{BHP}. 
\end{proof}
There exists a number of variants of the perturbation lemma for which we refer to the literature, in particular to
\cite[Lemmata A.1 \& A.2]{BHP} and \cite[2.4 \& 3.2]{CrainicPerturbation}. 

\begin{Remark}
  All of the above definitions and constructions can be performed when replacing
  the category of $R$-modules by an arbitrary abelian category or, with some additional care,
  by an additive subcategory of an abelian category.
  This is of relevance for homological reduction since there one essentially
  works within the category of Fr\'echet spaces which is an additive but not abelian
  subcategory of the abelian category of vector spaces over the field of real respectively
  complex numbers.
  In particular this means that the $k$-th (co)homology of a complex of Fr\'echet spaces
  might not be a Fr\'echet space again. But in any case it is still a vector space
  equipped with a compatible (possibly non-Hausdorff) vector space topology. In our construction
  of homological reduction and quantization we will point out this issue when necessary.  
\end{Remark}


\end{document}